\documentclass[11pt]{article}
\usepackage[utf8]{inputenc}
\usepackage[dvipsnames,table]{xcolor}
\usepackage{amsmath,amssymb,amsthm}
\usepackage{graphicx}
\usepackage{enumerate} 
\usepackage{caption,subcaption}
\usepackage{multicol}
\usepackage{floatpag}
\usepackage{placeins}

\usepackage{tikz}
\usepackage{pgfplots}
\usetikzlibrary{patterns}
\usetikzlibrary{decorations.pathreplacing}
\usetikzlibrary{shapes}
\usetikzlibrary{cd}
\usetikzlibrary{fit}
\usetikzlibrary{decorations.markings}
\usetikzlibrary{shapes.geometric}
\tikzset{
    rubberduck/.style={
        draw=black,
        shape=isosceles triangle,
        fill=white,
        minimum height=1.5cm,
        minimum width=0.5cm,
        shape border rotate=#1,
        isosceles triangle stretches,
        inner sep=0pt,
    },
    rubber/.style={rubberduck=+90}}
\tikzset{->-/.style={decoration={
  markings,
  mark=at position .5 with {\arrow{>}}},postaction={decorate}}}
\definecolor{ourpurple}{rgb}{0.60,0.50,0.71}
\usepackage{hyperref}
\hypersetup{
colorlinks=true,
breaklinks=true,
urlcolor= ourpurple,
linkcolor= ourpurple,
citecolor= ourpurple,
}
\usepackage{fullpage}
\usepackage[ruled,vlined,linesnumbered]{algorithm2e} 
\DontPrintSemicolon
\DeclareMathOperator{\parent}{\text{\normalfont\bf parent}}
\DeclareMathOperator{\children}{\text{\normalfont\bf children}}
\DeclareMathOperator{\depth}{\text{\normalfont\bf depth}}
\DeclareMathOperator{\leaves}{\text{\normalfont\bf leaves}}
\DeclareMathOperator{\theight}{\text{\normalfont\bf height}}
\DeclareMathOperator{\troot}{\text{\normalfont\bf root}}
\DeclareMathOperator{\degr}{\text{\normalfont\bf deg}}
\DeclareMathOperator{\attr}{\mathcal{A}}

\newcommand{\bfB}{\mathbf{B}}
\newcommand{\bfC}{\mathbf{C}}

\DeclareMathOperator{\argmin}{\arg\min}
\DeclareMathOperator{\argmax}{\arg\max}

\DeclareMathOperator{\isom}{\phi}

\DeclareMathOperator{\red}{\mathfrak{R}}
\DeclareMathOperator{\origin}{\text{\normalfont\bf origin}}

\DeclareMathOperator{\id}{\text{\normalfont \bf id}}
\DeclareMathOperator{\bags}{\mathbb{B}}
\DeclareMathOperator{\collections}{\mathbb{C}}

\DeclareMathOperator{\dom}{\text{\normalfont dom}}

\DeclareMathOperator{\supp}{\text{\normalfont supp}}

\colorlet{lblue}{blue!50!white}
\colorlet{lred}{red!50!white}
\colorlet{lgreen}{green!50!white}
\colorlet{lpurple}{purple!50!white}
\colorlet{lorange}{orange!50!white}
\colorlet{lpink}{pink!50!white}
\colorlet{lbrown}{brown!50!white}
\colorlet{lyellow}{yellow!50!white}
\colorlet{lolive}{olive!50!white}
\colorlet{dblue}{blue!75!white}
\colorlet{dred}{red!75!white}
\colorlet{dgreen}{green!75!white}
\colorlet{dpurple}{purple!75!white}
\colorlet{dorange}{orange!75!white}
\colorlet{dpink}{pink!75!white}
\colorlet{dbrown}{brown!75!white}
\colorlet{dyellow}{yellow!75!white}
\colorlet{dolive}{olive!75!white}

\newcommand{\node}[1][white]{%
\begin{tikzpicture}%
\tikzstyle{noeud}=[draw,circle,fill=#1,scale=0.7]%
\node[noeud] (0) at ({0},{0}) {};%
\end{tikzpicture}}

\newcommand{\inexI}{
\begin{tikzpicture}
\tikzstyle{noeud}=[draw,circle,scale=0.6]
\tikzstyle{arc}=[-,ultra thick,>=latex]
\node[noeud,fill=lblue] (u1) at (-1.875,1) {};
\node[noeud,fill=white,draw=white] (phantom) at (-3,1) {};

    \begin{axis}[
        width=3.5cm,
        height=3.5cm,
        xmin=0, xmax=615,
        ymin=0, ymax=100,
        xtick = \empty,
        area style,
        label style={font=\tiny},
        tick label style={font=\tiny}
    ]
    \addplot[ybar interval,fill=gray,draw=gray] coordinates {%
    (0,100) %
    (20,0) % must be keeped to display the last bar
    };
    \addplot[ybar interval,fill=black] coordinates {%
(20,12.3651452282158) %
(25,12.3858921161826) %
(30,7.5103734439834) %
(35,7.5103734439834) %
(40,5.33195020746888) %
(45,10.746887966805) %
(50,10.1244813278008) %
(55,12.1161825726141) %
(60,8.44398340248963) %
(65,11.9294605809129) %
(70,10.2697095435685) %
(75,11.4730290456432) %
(80,7.9253112033195) %
(85,10.7676348547718) %
(90,7.05394190871369) %
(95,5.4149377593361) %
(100,11.9087136929461) %
(105,11.701244813278) %
(110,11.3070539419087) %
(115,9.12863070539419) %
(120,6.55601659751037) %
(125,17.3236514522822) %
(130,17.1576763485477) %
(135,27.0331950207469) %
(140,25.8713692946058) %
(145,24.6887966804979) %
(150,12.655601659751) %
(155,15.6016597510373) %
(160,6.47302904564315) %
(165,23.298755186722) %
(170,21.3692946058091) %
(175,24.3983402489627) %
(180,13.7551867219917) %
(185,21.49377593361) %
(190,24.2738589211618) %
(195,18.7966804979253) %
(200,22.6348547717842) %
(205,23.9004149377593) %
(210,16.7219917012448) %
(215,22.0331950207469) %
(220,25.3112033195021) %
(225,23.6721991701245) %
(230,14.3153526970954) %
(235,10.7261410788382) %
(240,23.5477178423236) %
(245,14.7717842323651) %
(250,13.4232365145228) %
(255,38.0705394190871) %
(260,27.655601659751) %
(265,35.8298755186722) %
(270,38.0497925311203) %
(275,43.3402489626556) %
(280,19.6265560165975) %
(285,44.9170124481328) %
(290,15.9128630705394) %
(295,29.6473029045643) %
(300,16.3485477178423) %
(305,19.3153526970954) %
(310,5.24896265560166) %
(315,13.50622406639) %
(320,7.65560165975104) %
(325,12.3236514522822) %
(330,12.3236514522822) %
(335,6.74273858921162) %
(340,11.6597510373444) %
(345,8.85892116182573) %
(350,10.0207468879668) %
(355,5.93360995850622) %
(360,12.344398340249) %
(365,11.6804979253112) %
(370,12.7178423236515) %
(375,12.5933609958506) %
(380,10.4564315352697) %
(385,10.1659751037344) %
(390,6.34854771784232) %
(395,11.8879668049793) %
(400,13.9211618257261) %
(405,10.8091286307054) %
(410,16.3485477178423) %
(415,10.643153526971) %
(420,6.53526970954357) %
(425,6.97095435684647) %
(430,9.64730290456432) %
(435,13.9626556016598) %
(440,11.9502074688797) %
(445,13.9211618257261) %
(450,12.7385892116183) %
(455,10.3319502074689) %
(460,19.2116182572614) %
(465,6.53526970954357) %
(470,10.4979253112033) %
(475,7.26141078838174) %
(480,8.83817427385892) %
(485,7.84232365145228) %
(490,6.53526970954357) %
(495,8.87966804979253) %
(500,6.43153526970954) %
(505,5.850622406639) %
(510,8.15352697095436) %
(515,8.15352697095436) %
(520,8.15352697095436) %
(525,9.89626556016598) %
(530,9.89626556016598) %
(535,10.6224066390041) %
(540,12.3236514522822) %
(545,8.65145228215768) %
(550,8.23651452282158) %
(555,7.46887966804979) %
(560,12.5726141078838) %
(565,7.03319502074689) %
(570,5.72614107883817) %
(575,5.70539419087137) %
(580,10.6224066390041) %
(585,8.90041493775934) %
(590,6.72199170124481) %
(595,9.23236514522822) %
(600,8.00829875518672) %
(605,5.58091286307054) %
(610,5.53941908713693) %
    (615,0) % must be keeped to display the last bar
    };
    \addplot[draw=red,ultra thick] plot coordinates {
    (0,5) (1000,5)
    };
\end{axis}
\end{tikzpicture}
}

\newcommand{\inexA}{
\begin{tikzpicture}
\tikzstyle{noeud}=[draw,circle,scale=0.6]
\tikzstyle{arc}=[-,ultra thick,>=latex]
\node[noeud,fill=lblue] (u1) at (-1.875,1.5) {};
\node[noeud,fill=lred] (u2) at (-2.75,0.5) {};
\node[noeud,fill=lgreen] (u3) at (-1,0.5) {};
\node[noeud,fill=white,draw=white] (phantom) at (-3,1) {};
\draw[arc] (u1)--(u2);
\draw[arc] (u1)--(u3);

    \begin{axis}[
        width=3.5cm,
        height=3.5cm,
        xmin=0, xmax=100, %%% 
        ymin=0, ymax=100,
        xtick = \empty,
        area style,
        label style={font=\tiny},
        tick label style={font=\tiny}
    ]
    \addplot[ybar interval,fill=gray,draw=gray] coordinates {%
    (0,96.1203319502075) %
    (5,0) % must be keeped to display the last bar
    };
    \addplot[ybar interval,fill=black] coordinates {%
(5,7.5103734439834) %
(10,12.655601659751) %
(15,14.0248962655602) %
(20,7.0954356846473) %
(25,20.1659751037344) %
(30,25.2489626556017) %
(35,19.6265560165975) %
(40,15.9128630705394) %
(45,7.5103734439834) %
(50,6.74273858921162) %
(55,6.14107883817427) %
(60,9.62655601659751) %
(65,7.94605809128631) %
(70,5) %
(75,9.23236514522822) %
(80,7.4896265560166) %
(85,7.21991701244813) %
(90,8.7344398340249) %
(95,7.78008298755187) %
    (100,0) % must be keeped to display the last bar
    };
    \addplot[draw=red,ultra thick] plot coordinates {
    (0,5) (100,5)
    };
\end{axis}
\end{tikzpicture}
}

\newcommand{\inexB}{
\begin{tikzpicture}
\tikzstyle{noeud}=[draw,circle,scale=0.6]
\tikzstyle{arc}=[-,ultra thick,>=latex]
\node[noeud,fill=lblue] (u1) at (-1.875,1.5) {};
\node[noeud,fill=lred] (u2) at (-2.75,0.5) {};
\node[noeud,fill=lgreen] (u3) at (-1.875,0.5) {};
\node[noeud,fill=lpurple] (u4) at (-1,0.5) {};
\node[noeud,fill=white,draw=white] (phantom) at (-3,1) {};
\draw[arc] (u1)--(u2);
\draw[arc] (u1)--(u3);
\draw[arc] (u1)--(u4);

    \begin{axis}[
        width=3.5cm,
        height=3.5cm,
        xmin=0, xmax=30, %%% 
        ymin=0, ymax=100,
        xtick = \empty,
        area style,
        label style={font=\tiny},
        tick label style={font=\tiny}
    ]
    \addplot[ybar interval,fill=gray,draw=gray] coordinates {%
    (0,34) %
    (5,0) % must be keeped to display the last bar
    };
    \addplot[ybar interval,fill=black] coordinates {%
    (5,6.3) %
    (10,5.1)
    (15,13.9) 
    (20,9.8)
    (25,5.5)
    (30,0) % must be keeped to display the last bar
    };
    \addplot[draw=red,ultra thick] plot coordinates {
    (0,5) (100,5)
    };
\end{axis}
\end{tikzpicture}
}

\newcommand{\inexC}{
\begin{tikzpicture}
\tikzstyle{noeud}=[draw,circle,scale=0.6]
\tikzstyle{arc}=[-,ultra thick,>=latex]
\node[noeud,fill=lblue] (u1) at (-1.875,1.5) {};
\node[noeud,fill=lred] (u11) at (-3,0.5) {};
\node[noeud,fill=lgreen] (u12) at (-2.25,0.5) {};
\node[noeud,fill=lpurple] (u13) at (-1.5,0.5) {};
\node[noeud,fill=lorange] (u14) at (-0.75,0.5) {};

\draw[arc] (u1)--(u11);
\draw[arc] (u1)--(u12);
\draw[arc] (u1)--(u13);
\draw[arc] (u1)--(u14);

    \begin{axis}[
        width=3.5cm,
        height=3.5cm,
        xmin=0, xmax=10,
        ymin=0, ymax=100,
        xtick = \empty,
        area style,
        label style={font=\tiny},
        tick label style={font=\tiny}
    ]
    \addplot[ybar interval,fill=gray,draw=gray] coordinates {%
    (0,18.2) %
    (5,0) % must be keeped to display the last bar
    };
    \addplot[ybar interval,fill=black] coordinates {%
    (5,7.5) %
    (10,0) % must be keeped to display the last bar
    };
    \addplot[draw=red,ultra thick] plot coordinates {
    (0,5) (100,5)
    };
\end{axis}
\end{tikzpicture}
}

\newcommand{\inexD}{
\begin{tikzpicture}
\tikzstyle{noeud}=[draw,circle,scale=0.6]
\tikzstyle{arc}=[-,ultra thick,>=latex]
\node[noeud,fill=lblue] (u1) at (-1.875,1.5) {};
\node[noeud,fill=lred] (u11) at (-3,0.5) {};
\node[noeud,fill=lgreen] (u12) at (-2.4375,0.5) {};
\node[noeud,fill=lpurple] (u13) at (-1.8750,0.5) {};
\node[noeud,fill=lorange] (u14) at (-1.3125,0.5) {};
\node[noeud,fill=lpink] (u15) at (-0.75000,0.5) {};

\draw[arc] (u1)--(u11);
\draw[arc] (u1)--(u12);
\draw[arc] (u1)--(u13);
\draw[arc] (u1)--(u14);
\draw[arc] (u1)--(u15);

    \begin{axis}[
        width=3.5cm,
        height=3.5cm,
        xmin=0, xmax=10,
        ymin=0, ymax=100,
        xtick = \empty,
        area style,
        label style={font=\tiny},
        tick label style={font=\tiny}
    ]
    \addplot[ybar interval,fill=gray,draw=gray] coordinates {%
    (0,8.195020746887967) %
    (5,0) % must be keeped to display the last bar
    };
    \addplot[ybar interval,fill=black] coordinates {%
    (5,0) %
    (10,0) % must be keeped to display the last bar
    };
    \addplot[draw=red,ultra thick] plot coordinates {
    (0,5) (100,5)
    };
\end{axis}
\end{tikzpicture}
}

\newcommand{\inexE}{
\begin{tikzpicture}
\tikzstyle{noeud}=[draw,circle,scale=0.6]
\tikzstyle{arc}=[-,ultra thick,>=latex]
\node[noeud,fill=lblue] (u1) at (-1.875,1.75) {};
\node[noeud,fill=lred] (u11) at (-2.875,1) {};
\node[noeud,fill=lgreen] (u12) at (-2.34375,1) {};
\node[noeud,fill=lpurple] (u13) at (-1.81250,1) {};
\node[noeud,fill=lorange] (u14) at (-1.28125,1) {};
\node[noeud,fill=lpink] (u15) at (-0.75000,1) {};

\node[noeud,fill=dblue] (u111) at (-3,0.25) {};
\node[noeud,fill=dred] (u112) at (-2.75,0.25) {};

\draw[arc] (u1)--(u11);
\draw[arc] (u1)--(u12);
\draw[arc] (u1)--(u13);
\draw[arc] (u1)--(u14);
\draw[arc] (u1)--(u15);
\draw[arc] (u11)--(u111);
\draw[arc] (u11)--(u112);

    \begin{axis}[
        width=3.5cm,
        height=3.5cm,
        xmin=0, xmax=10,
        ymin=0, ymax=100,
        xtick = \empty,
        area style,
        label style={font=\tiny},
        tick label style={font=\tiny}
    ]
    \addplot[ybar interval,fill=gray,draw=gray] coordinates {%
    (0,6.8) %
    (5,0) % must be keeped to display the last bar
    };
    \addplot[ybar interval,fill=black] coordinates {%
    (5,0) %
    (10,0) % must be keeped to display the last bar
    };
    \addplot[draw=red,ultra thick] plot coordinates {
    (0,5) (100,5)
    };
\end{axis}
\end{tikzpicture}
}

\newcommand{\inexF}{
\begin{tikzpicture}
\tikzstyle{noeud}=[draw,circle,scale=0.6]
\tikzstyle{arc}=[-,ultra thick,>=latex]
\node[noeud,fill=lblue] (u1) at (-1.875,1.75) {};
\node[noeud,fill=lred] (u11) at (-2.875,1) {};
\node[noeud,fill=lgreen] (u12) at (-2.375,1) {};
\node[noeud,fill=lpurple] (u13) at (-2,1) {};
\node[noeud,fill=lorange] (u14) at (-1.583333,1) {};
\node[noeud,fill=lpink] (u15) at (-1.166667,1) {};
\node[noeud,fill=lbrown] (u16) at (-0.75,1) {};

\node[noeud,fill=dblue] (u111) at (-3,0.25) {};
\node[noeud,fill=dred] (u112) at (-2.75,0.25) {};
\node[noeud,fill=dpurple] (u121) at (-2.5,0.25) {};
\node[noeud,fill=dorange] (u122) at (-2.25,0.25) {};

\draw[arc] (u1)--(u11);
\draw[arc] (u1)--(u12);
\draw[arc] (u1)--(u13);
\draw[arc] (u1)--(u14);
\draw[arc] (u1)--(u15);
\draw[arc] (u1)--(u16);
\draw[arc] (u11)--(u111);
\draw[arc] (u11)--(u112);
\draw[arc] (u12)--(u121);
\draw[arc] (u12)--(u122);

    \begin{axis}[
        width=3.5cm,
        height=3.5cm,
        xmin=0, xmax=10,
        ymin=0, ymax=100,
        xtick = \empty,
        area style,
        label style={font=\tiny},
        tick label style={font=\tiny}
    ]
    \addplot[ybar interval,fill=gray,draw=gray] coordinates {%
    (0,8.195020746887967) %
    (5,0) % must be keeped to display the last bar
    };
    \addplot[ybar interval,fill=black] coordinates {%
    (5,0) %
    (10,0) % must be keeped to display the last bar
    };
    \addplot[draw=red,ultra thick] plot coordinates {
    (0,5) (100,5)
    };
\end{axis}
\end{tikzpicture}
}

\newcommand{\inexG}{
\begin{tikzpicture}
\tikzstyle{noeud}=[draw,circle,scale=0.6]
\tikzstyle{arc}=[-,ultra thick,>=latex]
\node[noeud,fill=lblue] (u1) at (-1.875,1.75) {};
\node[noeud,fill=lred] (u11) at (-2.875,1) {};
\node[noeud,fill=lgreen] (u12) at (-2.375,1) {};
\node[noeud,fill=lpurple] (u13) at (-2,1) {};
\node[noeud,fill=lorange] (u14) at (-1.6875,1) {};
\node[noeud,fill=lpink] (u15) at ( -1.3750,1) {};
\node[noeud,fill=lbrown] (u16) at (-1.0625,1) {};
\node[noeud,fill=lyellow] (u17) at (-0.75,1) {};

\node[noeud,fill=dblue] (u111) at (-3,0.25) {};
\node[noeud,fill=dred] (u112) at (-2.75,0.25) {};
\node[noeud,fill=dpurple] (u121) at (-2.5,0.25) {};
\node[noeud,fill=dorange] (u122) at (-2.25,0.25) {};

\draw[arc] (u1)--(u11);
\draw[arc] (u1)--(u12);
\draw[arc] (u1)--(u13);
\draw[arc] (u1)--(u14);
\draw[arc] (u1)--(u15);
\draw[arc] (u1)--(u16);
\draw[arc] (u1)--(u17);
\draw[arc] (u11)--(u111);
\draw[arc] (u11)--(u112);
\draw[arc] (u12)--(u121);
\draw[arc] (u12)--(u122);

    \begin{axis}[
        width=3.5cm,
        height=3.5cm,
        xmin=0, xmax=10,
        ymin=0, ymax=100,
        xtick = \empty,
        area style,
        label style={font=\tiny},
        tick label style={font=\tiny}
    ]
    \addplot[ybar interval,fill=gray,draw=gray] coordinates {%
    (0,6.804979253112033) %
    (5,0) % must be keeped to display the last bar
    };
    \addplot[ybar interval,fill=black] coordinates {%
    (5,0) %
    (10,0) % must be keeped to display the last bar
    };
    \addplot[draw=red,ultra thick] plot coordinates {
    (0,5) (100,5)
    };
\end{axis}
\end{tikzpicture}
}

\newcommand{\inexH}{
\begin{tikzpicture}
\tikzstyle{noeud}=[draw,circle,scale=0.6]
\tikzstyle{arc}=[-,ultra thick,>=latex]
\node[noeud,fill=lblue] (u1) at (-1.875,1.75) {};
\node[noeud,fill=lred] (u11) at (-2.875,1) {};
\node[noeud,fill=lgreen] (u12) at (-2.375,1) {};
\node[noeud,fill=lpurple] (u13) at (-2,1) {};
\node[noeud,fill=lorange] (u14) at (-1.75,1) {};
\node[noeud,fill=lpink] (u15) at ( -1.5,1) {};
\node[noeud,fill=lbrown] (u16) at (-1.25,1) {};
\node[noeud,fill=lyellow] (u17) at (-1,1) {};
\node[noeud,fill=lolive] (u18) at (-0.75,1) {};

\node[noeud,fill=dblue] (u111) at (-3,0.25) {};
\node[noeud,fill=dred] (u112) at (-2.75,0.25) {};
\node[noeud,fill=dpurple] (u121) at (-2.5,0.25) {};
\node[noeud,fill=dorange] (u122) at (-2.25,0.25) {};

\draw[arc] (u1)--(u11);
\draw[arc] (u1)--(u12);
\draw[arc] (u1)--(u13);
\draw[arc] (u1)--(u14);
\draw[arc] (u1)--(u15);
\draw[arc] (u1)--(u16);
\draw[arc] (u1)--(u17);
\draw[arc] (u1)--(u18);
\draw[arc] (u11)--(u111);
\draw[arc] (u11)--(u112);
\draw[arc] (u12)--(u121);
\draw[arc] (u12)--(u122);

    \begin{axis}[
        width=3.5cm,
        height=3.5cm,
        xmin=0, xmax=10,
        ymin=0, ymax=100,
        xtick = \empty,
        area style,
        label style={font=\tiny},
        tick label style={font=\tiny}
    ]
    \addplot[ybar interval,fill=gray,draw=gray] coordinates {%
    (0,9.3) %
    (5,0) % must be keeped to display the last bar
    };
    \addplot[ybar interval,fill=black] coordinates {%
    (5,0) %
    (10,0) % must be keeped to display the last bar
    };
    \addplot[draw=red,ultra thick] plot coordinates {
    (0,5) (100,5)
    };
\end{axis}
\end{tikzpicture}
}
\newcommand{\inexxM}{
\begin{tikzpicture}
\tikzstyle{noeud}=[draw,circle,scale=0.6]
\tikzstyle{arc}=[-,ultra thick,>=latex]
\node[noeud,fill=lblue] (u1) at (-1.875,1) {};
\node[noeud,fill=white,draw=white] (phantom) at (-3,1) {};

    \begin{axis}[
        width=3.5cm,
        height=3.5cm,
        xmin=0, xmax=60,
        ymin=0, ymax=100,
        xtick = \empty,
        area style,
        label style={font=\tiny},
        tick label style={font=\tiny}
    ]
    \addplot[ybar interval,fill=gray,draw=gray] coordinates {%
    (0,100) %
    (5,0) % must be keeped to display the last bar
    };
    \addplot[ybar interval,fill=black] coordinates {%
    (5,52) %
    (10,98) %
    (15,88) %
    (20,83) %
    (25,51) %
    (30,59) %
    (35,29) %
    (40,25) %
    (45,49) %
    (50,14) %
    (55,31) %
    (60,0) % must be keeped to display the last bar
    };
    \addplot[draw=red,ultra thick] plot coordinates {
    (0,5) (100,5)
    };
\end{axis}
\end{tikzpicture}
}

\newcommand{\inexxA}{
\begin{tikzpicture}
\tikzstyle{noeud}=[draw,circle,scale=0.6]
\tikzstyle{arc}=[-,ultra thick,>=latex]
\node[noeud,fill=lblue] (u1) at (-1.875,1.5) {};
\node[noeud,fill=lblue] (u2) at (-1.875,0.5) {};
\node[noeud,fill=white,draw=white] (phantom) at (-3,1) {};
\draw[arc] (u1)--(u2);

    \begin{axis}[
        width=3.5cm,
        height=3.5cm,
        xmin=0, xmax=10, %%% 
        ymin=0, ymax=100,
        xtick = \empty,
        area style,
        label style={font=\tiny},
        tick label style={font=\tiny}
    ]
    \addplot[ybar interval,fill=gray,draw=gray] coordinates {%
    (0,10) %
    (5,0) % must be keeped to display the last bar
    };
    \addplot[ybar interval,fill=black] coordinates {%
    (5,10) %
    (10,0) % must be keeped to display the last bar
    };
    \addplot[draw=red,ultra thick] plot coordinates {
    (0,5) (100,5)
    };
\end{axis}
\end{tikzpicture}
}

\newcommand{\inexxB}{
\begin{tikzpicture}
\tikzstyle{noeud}=[draw,circle,scale=0.6]
\tikzstyle{arc}=[-,ultra thick,>=latex]
\node[noeud,fill=lblue] (u1) at (-1.875,1.5) {};
\node[noeud,fill=lred] (u2) at (-1.875,0.5) {};
\node[noeud,fill=white,draw=white] (phantom) at (-3,1) {};
\draw[arc] (u1)--(u2);

    \begin{axis}[
        width=3.5cm,
        height=3.5cm,
        xmin=0, xmax=45, %%% 
        ymin=0, ymax=100,
        xtick = \empty,
        area style,
        label style={font=\tiny},
        tick label style={font=\tiny}
    ]
    \addplot[ybar interval,fill=gray,draw=gray] coordinates {%
    (0,86) %
    (5,0) % must be keeped to display the last bar
    };
    \addplot[ybar interval,fill=black] coordinates {%
    (5,12) %
    (10,23) %
    (15,8) %
    (20,11) %
    (25,15) %
    (30,24) %
    (35,21) %
    (40,58) %
    (45,0) % must be keeped to display the last bar
    };
    \addplot[draw=red,ultra thick] plot coordinates {
    (0,5) (100,5)
    };
\end{axis}
\end{tikzpicture}
}

\newcommand{\inexxE}{
\begin{tikzpicture}
\tikzstyle{noeud}=[draw,circle,scale=0.6]
\tikzstyle{arc}=[-,ultra thick,>=latex]
\node[noeud,fill=lblue] (u1) at (-1.875,1.5) {};
\node[noeud,fill=lred] (u2) at (-2.75,0.5) {};
\node[noeud,fill=lgreen] (u3) at (-1,0.5) {};
\node[noeud,fill=white,draw=white] (phantom) at (-3,1) {};
\draw[arc] (u1)--(u2);
\draw[arc] (u1)--(u3);

    \begin{axis}[
        width=3.5cm,
        height=3.5cm,
        xmin=0, xmax=25, %%% 
        ymin=0, ymax=100,
        xtick = \empty,
        area style,
        label style={font=\tiny},
        tick label style={font=\tiny}
    ]
    \addplot[ybar interval,fill=gray,draw=gray] coordinates {%
    (0,90) %
    (5,0) % must be keeped to display the last bar
    };
    \addplot[ybar interval,fill=black] coordinates {%
    (5,11) %
    (10,8) %
    (15,17) %
    (20,64) %
    (25,0) % must be keeped to display the last bar
    };
    \addplot[draw=red,ultra thick] plot coordinates {
    (0,5) (100,5)
    };
\end{axis}
\end{tikzpicture}
}

\newcommand{\inexxF}{
\begin{tikzpicture}
\tikzstyle{noeud}=[draw,circle,scale=0.6]
\tikzstyle{arc}=[-,ultra thick,>=latex]
\node[noeud,fill=lblue] (u1) at (-1.875,1.5) {};
\node[noeud,fill=lred] (u2) at (-2.75,0.5) {};
\node[noeud,fill=lgreen] (u3) at (-1.875,0.5) {};
\node[noeud,fill=lpurple] (u4) at (-1,0.5) {};
\node[noeud,fill=white,draw=white] (phantom) at (-3,1) {};
\draw[arc] (u1)--(u2);
\draw[arc] (u1)--(u3);
\draw[arc] (u1)--(u4);

    \begin{axis}[
        width=3.5cm,
        height=3.5cm,
        xmin=0, xmax=10, %%% 
        ymin=0, ymax=100,
        xtick = \empty,
        area style,
        label style={font=\tiny},
        tick label style={font=\tiny}
    ]
    \addplot[ybar interval,fill=gray,draw=gray] coordinates {%
    (0,25) %
    (5,0) % must be keeped to display the last bar
    };
    \addplot[ybar interval,fill=black] coordinates {%
    (5,22) %
    (10,0) % must be keeped to display the last bar
    };
    \addplot[draw=red,ultra thick] plot coordinates {
    (0,5) (100,5)
    };
\end{axis}
\end{tikzpicture}
}

\newcommand{\inexxC}{
\begin{tikzpicture}
\tikzstyle{noeud}=[draw,circle,scale=0.6]
\tikzstyle{arc}=[-,ultra thick,>=latex]
\node[noeud,fill=lblue] (u1) at (-1.875,1.75) {};
\node[noeud,fill=lred] (u2) at (-1.875,1) {};
\node[noeud,fill=lred] (u3) at (-1.875,0.25) {};
\node[noeud,fill=white,draw=white] (phantom) at (-3,1) {};
\draw[arc] (u1)--(u2);
\draw[arc] (u2)--(u3);

    \begin{axis}[
        width=3.5cm,
        height=3.5cm,
        xmin=0, xmax=10, %%% 
        ymin=0, ymax=100,
        xtick = \empty,
        area style,
        label style={font=\tiny},
        tick label style={font=\tiny}
    ]
    \addplot[ybar interval,fill=gray,draw=gray] coordinates {%
    (0,7) %
    (5,0) % must be keeped to display the last bar
    };
    \addplot[ybar interval,fill=black] coordinates {%
    (5,7) %
    (10,0) % must be keeped to display the last bar
    };
    \addplot[draw=red,ultra thick] plot coordinates {
    (0,5) (100,5)
    };
\end{axis}
\end{tikzpicture}
}

\newcommand{\inexxD}{
\begin{tikzpicture}
\tikzstyle{noeud}=[draw,circle,scale=0.6]
\tikzstyle{arc}=[-,ultra thick,>=latex]
\node[noeud,fill=lblue] (u1) at (-1.875,1.75) {};
\node[noeud,fill=lred] (u2) at (-1.875,1) {};
\node[noeud,fill=lgreen] (u3) at (-1.875,0.25) {};
\node[noeud,fill=white,draw=white] (phantom) at (-3,1) {};
\draw[arc] (u1)--(u2);
\draw[arc] (u2)--(u3);

    \begin{axis}[
        width=3.5cm,
        height=3.5cm,
        xmin=0, xmax=15, %%% 
        ymin=0, ymax=100,
        xtick = \empty,
        area style,
        label style={font=\tiny},
        tick label style={font=\tiny}
    ]
    \addplot[ybar interval,fill=gray,draw=gray] coordinates {%
    (0,31) %
    (5,0) % must be keeped to display the last bar
    };
    \addplot[ybar interval,fill=black] coordinates {%
    (5,16) %
    (10,15) %
    (15,0) % must be keeped to display the last bar
    };
    \addplot[draw=red,ultra thick] plot coordinates {
    (0,5) (100,5)
    };
\end{axis}
\end{tikzpicture}
}

\newcommand{\inexxH}{
\begin{tikzpicture}
\tikzstyle{noeud}=[draw,circle,scale=0.6]
\tikzstyle{arc}=[-,ultra thick,>=latex]
\node[noeud,fill=lblue] (u1) at (-1.875,1.75) {};
\node[noeud,fill=lred] (u2) at (-1.875,1) {};
\node[noeud,fill=lgreen] (u3) at (-2.75,0.25) {};
\node[noeud,fill=lpurple] (u4) at (-1,0.25) {};
\node[noeud,fill=white,draw=white] (phantom) at (-3,1) {};
\draw[arc] (u1)--(u2);
\draw[arc] (u2)--(u3);
\draw[arc] (u2)--(u4);

    \begin{axis}[
        width=3.5cm,
        height=3.5cm,
        xmin=0, xmax=10, %%% 
        ymin=0, ymax=100,
        xtick = \empty,
        area style,
        label style={font=\tiny},
        tick label style={font=\tiny}
    ]
    \addplot[ybar interval,fill=gray,draw=gray] coordinates {%
    (0,9) %
    (5,0) % must be keeped to display the last bar
    };
    \addplot[ybar interval,fill=black] coordinates {%
    (5,5) %
    (10,0) % must be keeped to display the last bar
    };
    \addplot[draw=red,ultra thick] plot coordinates {
    (0,5) (100,5)
    };
\end{axis}
\end{tikzpicture}
}

\newcommand{\inexxI}{
\begin{tikzpicture}
\tikzstyle{noeud}=[draw,circle,scale=0.6]
\tikzstyle{arc}=[-,ultra thick,>=latex]
\node[noeud,fill=lblue] (u1) at (-1.875,1.75) {};
\node[noeud,fill=lred] (u2) at (-2.75,1) {};
\node[noeud,fill=lgreen] (u3) at (-1,1) {};
\node[noeud,fill=lpurple] (u4) at (-2.75,0.25) {};
\node[noeud,fill=white,draw=white] (phantom) at (-3,1) {};
\draw[arc] (u1)--(u2);
\draw[arc] (u1)--(u3);
\draw[arc] (u2)--(u4);

    \begin{axis}[
        width=3.5cm,
        height=3.5cm,
        xmin=0, xmax=10, %%% 
        ymin=0, ymax=100,
        xtick = \empty,
        area style,
        label style={font=\tiny},
        tick label style={font=\tiny}
    ]
    \addplot[ybar interval,fill=gray,draw=gray] coordinates {%
    (0,8) %
    (5,0) % must be keeped to display the last bar
    };
    \addplot[ybar interval,fill=black] coordinates {%
    (5,7) %
    (10,0) % must be keeped to display the last bar
    };
    \addplot[draw=red,ultra thick] plot coordinates {
    (0,5) (100,5)
    };
\end{axis}
\end{tikzpicture}
}

\newcommand{\inexxJ}{
\begin{tikzpicture}
\tikzstyle{noeud}=[draw,circle,scale=0.6]
\tikzstyle{arc}=[-,ultra thick,>=latex]
\node[noeud,fill=lblue] (u1) at (-1.875,1.75) {};
\node[noeud,fill=lred] (u2) at (-2.75,1) {};
\node[noeud,fill=lgreen] (u3) at (-1.875,1) {};
\node[noeud,fill=lpurple] (u4) at (-1,1) {};
\node[noeud,fill=lorange] (u5) at (-2.75,0.25) {};

\node[noeud,fill=white,draw=white] (phantom) at (-3,1) {};
\draw[arc] (u1)--(u2);
\draw[arc] (u1)--(u3);
\draw[arc] (u1)--(u4);
\draw[arc] (u2)--(u5);

    \begin{axis}[
        width=3.5cm,
        height=3.5cm,
        xmin=0, xmax=15, %%% 
        ymin=0, ymax=100,
        xtick = \empty,
        area style,
        label style={font=\tiny},
        tick label style={font=\tiny}
    ]
    \addplot[ybar interval,fill=gray,draw=gray] coordinates {%
    (0,32) %
    (5,0) % must be keeped to display the last bar
    };
    \addplot[ybar interval,fill=black] coordinates {%
    (5,7) %
    (10,25) %
    (15,0) % must be keeped to display the last bar
    };
    \addplot[draw=red,ultra thick] plot coordinates {
    (0,5) (100,5)
    };
\end{axis}
\end{tikzpicture}
}

\newcommand{\inexxG}{
\begin{tikzpicture}
\tikzstyle{noeud}=[draw,circle,scale=0.6]
\tikzstyle{arc}=[-,ultra thick,>=latex]
\node[noeud,fill=lblue] (u1) at (-1.875,1.75) {};
\node[noeud,fill=lred] (u2) at (-2.75,1) {};
\node[noeud,fill=lgreen] (u3) at (-1.875,1) {};
\node[noeud,fill=lpurple] (u4) at (-1,1) {};
\node[noeud,fill=lorange] (u5) at (-2.75,0.25) {};
\node[noeud,fill=lpink] (u6) at (-1.875,0.25) {};

\node[noeud,fill=white,draw=white] (phantom) at (-3,1) {};
\draw[arc] (u1)--(u2);
\draw[arc] (u1)--(u3);
\draw[arc] (u1)--(u4);
\draw[arc] (u2)--(u5);
\draw[arc] (u3)--(u6);

    \begin{axis}[
        width=3.5cm,
        height=3.5cm,
        xmin=0, xmax=10, %%% 
        ymin=0, ymax=100,
        xtick = \empty,
        area style,
        label style={font=\tiny},
        tick label style={font=\tiny}
    ]
    \addplot[ybar interval,fill=gray,draw=gray] coordinates {%
    (0,5) %
    (5,0) % must be keeped to display the last bar
    };
    \addplot[ybar interval,fill=black] coordinates {%
    (5,0) %
    (10,0) % must be keeped to display the last bar
    };
    \addplot[draw=red,ultra thick] plot coordinates {
    (0,5) (100,5)
    };
\end{axis}
\end{tikzpicture}
}

\newcommand{\inexxK}{
\begin{tikzpicture}
\tikzstyle{noeud}=[draw,circle,scale=0.6]
\tikzstyle{arc}=[-,ultra thick,>=latex]
\node[noeud,fill=lblue] (u1) at (-1.875,1.75) {};
\node[noeud,fill=lred] (u2) at (-3,1) {};
\node[noeud,fill=lgreen] (u3) at (-2.25,1) {};
\node[noeud,fill=lpurple] (u4) at (-1.5,1) {};
\node[noeud,fill=lorange] (u5) at (-0.75,1) {};
\node[noeud,fill=lpink] (u6) at (-3,0.25) {};

\draw[arc] (u1)--(u2);
\draw[arc] (u1)--(u3);
\draw[arc] (u1)--(u4);
\draw[arc] (u1)--(u5);
\draw[arc] (u2)--(u6);

    \begin{axis}[
        width=3.5cm,
        height=3.5cm,
        xmin=0, xmax=15, %%% 
        ymin=0, ymax=100,
        xtick = \empty,
        area style,
        label style={font=\tiny},
        tick label style={font=\tiny}
    ]
    \addplot[ybar interval,fill=gray,draw=gray] coordinates {%
    (0,17) %
    (5,0) % must be keeped to display the last bar
    };
    \addplot[ybar interval,fill=black] coordinates {%
    (5,8) %
    (10,9) %
    (15,0) % must be keeped to display the last bar
    };
    \addplot[draw=red,ultra thick] plot coordinates {
    (0,5) (100,5)
    };
\end{axis}
\end{tikzpicture}
}

\newcommand{\inexxL}{
\begin{tikzpicture}
\tikzstyle{noeud}=[draw,circle,scale=0.6]
\tikzstyle{arc}=[-,ultra thick,>=latex]
\node[noeud,fill=lblue] (u1) at (-1.875,1.75) {};
\node[noeud,fill=lred] (u11) at (-3,1) {};
\node[noeud,fill=lgreen] (u12) at (-2.25,1) {};
\node[noeud,fill=lpurple] (u13) at (-1.5,1) {};
\node[noeud,fill=lorange] (u14) at (-0.75,1) {};
\node[noeud,fill=lpink] (u111) at (-3,0.25) {};
\node[noeud,fill=lbrown] (u121) at (-2.25,0.25) {};
\draw[arc] (u1)--(u11);
\draw[arc] (u1)--(u12);
\draw[arc] (u1)--(u13);
\draw[arc] (u1)--(u14);
\draw[arc] (u11)--(u111);
\draw[arc] (u12)--(u121);

    \begin{axis}[
        width=3.5cm,
        height=3.5cm,
        xmin=0, xmax=10, %%% 
        ymin=0, ymax=100,
        xtick = \empty,
        area style,
        label style={font=\tiny},
        tick label style={font=\tiny}
    ]
    \addplot[ybar interval,fill=gray,draw=gray] coordinates {%
    (0,6) %
    (5,0) % must be keeped to display the last bar
    };
    \addplot[ybar interval,fill=black] coordinates {%
    (5,6) %
    (10,0) % must be keeped to display the last bar
    };
    \addplot[draw=red,ultra thick] plot coordinates {
    (0,5) (100,5)
    };
\end{axis}
\end{tikzpicture}
}

\usepackage{euler} % police maths
\usepackage{palatino} % police
\usepackage{fourier-orns} % leaf
\usepackage{pifont} % xmark
\usepackage[colored]{shadethm}
\definecolor{shadethmcolor}{rgb}{0.96,0.93,1.0} % couleur du fond
\definecolor{shaderulecolor}{rgb}{1,1,1} % couleur de l'encadré

\newshadetheorem{theorem}{Theorem}[section]
\newshadetheorem{proposition}{Proposition}[section]
\newshadetheorem{lemma}[proposition]{Lemma}
\newshadetheorem{corollary}[proposition]{Corollary}
\newshadetheorem{remark}{Remark}[section]
\newshadetheorem{definition}{Definition}[section]
\newshadetheorem{observation}{Observation}
\newshadetheorem{deduction}{Deduction Rule}

 \renewenvironment{proof}{{\noindent\bfseries Proof.}}{\hfill{\color{ourpurple}{\decoone}}}
\title{\Large\textsc{Detection of Common Subtrees}\\ \textsc{with Identical Label Distribution}}
\author{
Romain Azaïs \\ \href{mailto:romain.azais@inria.fr}{romain.azais@inria.fr}
\and
Florian Ingels \\ \href{mailto:florian.ingels@gmail.com}{florian.ingels@gmail.com}}
\date{\small Inria team MOSAIC, Lyon, France}

\begin{document}

\maketitle

%%%%%%%%%%%%%%%%%%%%%%%%%%%%%%%%%%%%%%%%%%%%%%%%%%%%%%%%%%%%

{\small
\begin{center}\textbf{Abstract}\end{center}

\vspace{-0.2cm}

\noindent
Frequent pattern mining is a relevant method to analyse structured data, like sequences, trees or graphs. It consists in identifying characteristic substructures of a dataset. This paper deals with a new type of patterns for tree data: common subtrees with identical label distribution. Their detection is far from obvious since the underlying isomorphism problem is graph isomorphism complete. An elaborated search algorithm is developed and analysed from both theoretical and numerical perspectives. Based on this, the enumeration of patterns is performed through a new lossless compression scheme for trees, called DAG-RW, whose complexity is investigated as well. The method shows very good properties, both in terms of computation times and analysis of real datasets from the literature. Compared to other substructures like topological subtrees and labelled subtrees for which the isomorphism problem is linear, the patterns found provide a more parsimonious representation of the data.

\smallskip

}

\section{Introduction}

\subsection{Context}

Tree data are ubiquitous, especially in biology and computer science, but also non-Euclidean \cite{bmmv2002}, which prevents them from being analysed by classical statistical methods adapted to multi-dimensional data. Therefore, they require the development of specific tools that take into account their structured nature. Among such techniques, frequent pattern mining \cite{aggarwal2014frequent} consists in identifying patterns, i.e. substructures, that appear often in the data. The more elaborate the patterns searched, the more difficult the problem is: the issue is to preserve a reasonable algorithmic complexity that allows the search of a given family of patterns in a reasonable time.

Different types of patterns have been considered in the literature to analyse tree data (see the survey \cite{jimenez2010} and the references therein) with a strong interest in a specific family of patterns called subtrees \cite{asai2003discovering,zaki2002}. In these two papers, only subtrees that appear more often than a given threshold are considered. Reverse search \cite{avis1996reverse} is a generic approach for enumerating frequent patterns in a dataset that consists in (i) building an enumeration tree of substructures, and then (ii) pruning it to keep only frequent patterns. It was notably used in \cite{asai2003discovering} for subtrees and in \cite{ingels2022enumeration} for subforests, i.e. sets of subtrees. Another technique is to enumerate all the substructures, then evaluate their frequency. For subtrees, this can be done with DAG compression \cite{valiente2001}, a tree compression scheme that eliminates subtree redundancies. Interestingly, the subtree kernel (a similarity function based on comparison of subtrees) can be evaluated from the enumeration of all the subtrees obtained from the DAG compression of the dataset \cite{azais2020weight}, showing incidentally that this is a credible alternative to reverse search when the number of substructures is linear.

Subtrees can be considered with or without their labels. Considering the labels, the condition of subtree redundancy is more stringent and one may not find frequent patterns that characterise the data. When labels are ignored, the condition is more flexible but some information is lost. In this paper, we propose to fill the gap between labelled and unlabelled subtrees by considering a new type of patterns: subtrees with identical label distribution. The frequent pattern mining algorithm that we introduce and analyse is based on a generalisation of DAG compression to these substructures. It allows the detection of richer patterns than labelled subtrees while preserving the information of the labels unlike unlabelled subtrees.

\subsection{Background}
\label{ss:intro:dagcompression}

A (rooted) tree is a connected directed graph with no undirected cycle and such that there exists a special node called the root, which has no parent, and such that any node different from the root has exactly one parent.
In the setting of this paper, each node $u$ of a tree $T$ carries a label, denoted by $\overline{u}$, the nature of which is not assumed. The set of labels of $T$, also called alphabet, is defined as $\attr(T) = \lbrace \overline{u} : u\in T \rbrace$.
The structures of two trees $T_1$ and $T_2$ can be compared taking into account or not their labels.

\begin{definition}\label{def:isom}
A bijection $\isom : T_1\to T_2$ is a tree isomorphism if, for any $u,v\in T_1$, $u$ is a child of $v$ in $T_1$ if and only if $\isom(u)$ is a child of $\isom(v)$ in $T_2$. If in addition, for any $u\in T_1$, $\overline{u}=\overline{\isom(u)}$, then $\isom$ is a labelled tree isomorphism.
\end{definition}

The existence of a (labelled or not) tree isomorphism (which can be established in linear time \cite{aho1974design,valiente2013algorithms}) induces an equivalence relation on the set of trees. $T_1\simeq T_2$ means that there exists a tree isomorphism between $T_1$ and $T_2$.
Adding $l$ as a subscript to symbol $\simeq$ means that tree isomorphism is considered with labels.

Subtrees are an important class of tree substructures. A subtree $T(v)$ of a tree $T$ is the tree made of the node $v$ and all its descendants in $T$. $T$ can have internal repetitions in its subtrees, regardless of the isomorphism used to evaluate them. DAG compression aims at eliminating these redundancies through the construction of a new graph $\red_\ast(T)$ (where $\ast\in\{\simeq,\simeq_l\}$ denotes the selected equivalence relation) which contains each of the subtrees of $T$ only once, up to isomorphism \cite{valiente2001}. Consequently, parsimonious enumeration of subtrees for frequent pattern mining can be realised from the graph $\red_\ast(T)$.

The set of vertices of $\red_\ast(T)$ corresponds to the set of equivalence classes of the subtrees of $T$. Let $s:\red_\ast(T)\to T$ be a section, i.e. an injective function such that the equivalence class of $s(a)$ is $a$. If we consider labelled tree isomorphisms, each node $a$ of $\red_{\simeq_l}(T)$ carries the label of $s(a)$. The set of edges of $\red_\ast(T)$ is built as follows. For any nodes $a$ and $b$ in $\red_\ast(T)$, there are as many edges from $a$ to $b$ as the number of apparitions (up to isomorphism) of $s(b)$ as a child of $s(a)$. The graph $\red_\ast(T)$ can be computed in linear time \cite{valiente2001}. It is a DAG that compresses lossless the structure of $T$. Lossless means that a tree in the same equivalence class (with respect to the selected equivalence relation $\ast$) as $T$ can be reconstructed. However, it should be noted that labels of $T$ can not be inferred when working with unlabelled tree isomorphisms: the DAG $\red_{\simeq}(T)$ does not carry any piece of information about them. Figure~\ref{fig:dag_compression_example} illustrates DAG compression on an example for both $\simeq$ and $\simeq_l$.

\begin{figure}[h]
\def\xscale{0.67}
\def\yscale{0.9}
\def\nodescale{0.7}
\centering
\begin{tikzpicture}[xscale=\xscale,yscale=\yscale]
\tikzstyle{noeud}=[draw,circle,fill=white,scale=\nodescale*1]
\tikzstyle{arc}=[->-,>=latex]
\tikzstyle{fleche}=[->,>=latex,red,thick]

\node[noeud,fill=lorange] (u0) at (0,3) {$0$};

\node[noeud,fill=lblue] (u1) at (-3.5,2) {$1$};
\node[noeud,fill=lblue] (u2) at (3.5,2) {$2$};

\node[noeud,fill=lred] (u3) at (-5.5,1) {$2$};
\node[noeud,fill=lred] (u4) at (-3.5,1) {$4$};
\node[noeud,fill=lred] (u5) at (-1.5,1) {$3$};
\node[noeud,fill=lred] (u6) at (1.5,1) {$4$};
\node[noeud,fill=lred] (u7) at (3.5,1) {$8$};
\node[noeud,fill=lred] (u8) at (5.5,1) {$6$};

\node[noeud,fill=lgreen] (u9) at (-6,0) {$3$};
\node[noeud,fill=lgreen] (u10) at (-5,0) {$4$};
\node[noeud,fill=lgreen] (u11) at (-4,0) {$9$};
\node[noeud,fill=lgreen] (u12) at (-3,0) {$16$};
\node[noeud,fill=lgreen] (u13) at (-2,0) {$3$};
\node[noeud,fill=lgreen] (u14) at (-1,0) {$4$};
\node[noeud,fill=lgreen] (u15) at (1,0) {$6$};
\node[noeud,fill=lgreen] (u16) at (2,0) {$8$};
\node[noeud,fill=lgreen] (u17) at (3,0) {$18$};
\node[noeud,fill=lgreen] (u18) at (4,0) {$32$};
\node[noeud,fill=lgreen] (u19) at (5,0) {$6$};
\node[noeud,fill=lgreen] (u20) at (6,0) {$8$};

\node (t1) at (0,4) {$T$};

\draw (u0)--(u1);
\draw (u0)--(u2);
\draw (u1)--(u3);
\draw (u1)--(u4);
\draw (u1)--(u5);
\draw (u2)--(u6);
\draw (u2)--(u7);
\draw (u2)--(u8);
\draw (u3)--(u9);
\draw (u3)--(u10);
\draw (u4)--(u11);
\draw (u4)--(u12);
\draw (u5)--(u13);
\draw (u5)--(u14);
\draw (u6)--(u15);
\draw (u6)--(u16);
\draw (u7)--(u17);
\draw (u7)--(u18);
\draw (u8)--(u19);
\draw (u8)--(u20);

\def\xshift{7.5}

\node[noeud,fill=lorange] (u0) at (\xshift,3) {};
\node[noeud,fill=lblue] (u1) at (\xshift,2) {};
\node[noeud,fill=lred] (u2) at (\xshift,1) {};
\node[noeud,fill=lgreen] (u3) at (\xshift,0) {};

\node (t1) at (\xshift,4) {$\red_\simeq(T)$};

\draw[arc] (u0) to[bend left=45] (u1);
\draw[arc] (u0) to[bend right=45] (u1);
\draw[arc] (u1) to[bend left=45] (u2);
\draw[arc] (u1) to[bend right=45] (u2);
\draw[arc] (u1) to (u2);
\draw[arc] (u2) to[bend left=45] (u3);
\draw[arc] (u2) to[bend right=45] (u3);

\def\xshift{13}

\node[noeud,fill=lorange] (u0) at (0+\xshift,3) {$0$};

\node[noeud,fill=lblue] (u1) at (-2.25+\xshift,2) {$1$};
\node[noeud,fill=lblue] (u2) at (2.25+\xshift,2) {$2$};

\node[noeud,fill=lred] (u3) at (-4+\xshift,1) {$2$};
\node[noeud,fill=lred] (u4) at (-1.5+\xshift,1) {$4$};
\node[noeud,fill=lred] (u5) at (-3+\xshift,1) {$3$};
\node[noeud,fill=lred] (u6) at (1+\xshift,1) {$4$};
\node[noeud,fill=lred] (u7) at (3.5+\xshift,1) {$8$};
\node[noeud,fill=lred] (u8) at (2+\xshift,1) {$6$};

\node[noeud,fill=lgreen] (u9) at (-4+\xshift,0) {$3$};
\node[noeud,fill=lgreen] (u10) at (-3+\xshift,0) {$4$};
\node[noeud,fill=lgreen] (u11) at (-2+\xshift,0) {$9$};
\node[noeud,fill=lgreen] (u12) at (-1+\xshift,0) {$16$};
\node[noeud,fill=lgreen] (u15) at (1+\xshift,0) {$6$};
\node[noeud,fill=lgreen] (u16) at (2+\xshift,0) {$8$};
\node[noeud,fill=lgreen] (u17) at (3+\xshift,0) {$18$};
\node[noeud,fill=lgreen] (u18) at (4+\xshift,0) {$32$};

\node (t1) at (0+\xshift,4) {$\red_{\simeq_l}(T)$};

\draw[arc] (u0)--(u1);
\draw[arc] (u0)--(u2);
\draw[arc] (u1)--(u3);
\draw[arc] (u1)--(u4);
\draw[arc] (u1)--(u5);
\draw[arc] (u2)--(u6);
\draw[arc] (u2)--(u7);
\draw[arc] (u2)--(u8);
\draw[arc] (u3)--(u9);
\draw[arc] (u3)--(u10);
\draw[arc] (u4)--(u11);
\draw[arc] (u4)--(u12);
\draw[arc] (u5)--(u9);
\draw[arc] (u5)--(u10);
\draw[arc] (u6)--(u15);
\draw[arc] (u6)--(u16);
\draw[arc] (u7)--(u17);
\draw[arc] (u7)--(u18);
\draw[arc] (u8)--(u15);
\draw[arc] (u8)--(u16);
\end{tikzpicture}

\caption{From left to right: a labelled tree $T$, its unlabelled DAG compression $\red_\simeq(T)$ and its labelled DAG compression $\red_{\simeq_l}(T)$. Nodes with same equivalence class (with respect to $\simeq$) are colored accordingly.}
\label{fig:dag_compression_example}
\end{figure}
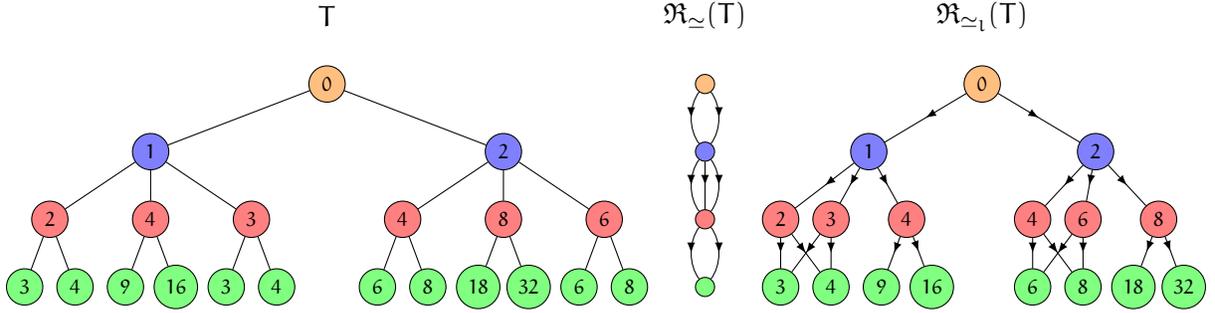

By definition, equivalence relation $\simeq_l$ is more stringent than $\simeq$, i.e. $T_1\simeq_lT_2\,\Rightarrow\,T_1\simeq T_2$. This means that patterns detected for $\simeq$ are more frequent and can help to detect more similarities in the data, but without taking into account the labels. This is also reflected in compression levels, because, for any tree $T$, $\red_{\simeq}(T)\subset \red_{\simeq_l}(T)$. Nevertheless, $\red_\simeq(T)$ loses the piece of information carried by the labels of $T$.

A tree isomorphism $\isom$ between two trees $T_1$ and $T_2$ naturally induces a binary relation $\mathrel{R_{\isom}}$ over $\attr(T_1)\times\attr(T_2)$, defined as, for any $x\in\attr(T_1)$ and $y\in\attr(T_2)$,
\begin{equation*}
x \mathrel{R_{\isom}} y\quad\iff\quad\exists\,u\in T_1,~(x=\overline{u} )\wedge (y=\overline{\isom(u)}).
\end{equation*}
Such a relation $\mathrel{R_{\isom}}$ is said to be a bijection if and only if, for any $x\in\attr(T_1)$, there exists a unique $y\in\attr(T_2)$, which we denote $f_{\isom}(x)$, so that $x\mathrel{R_{\isom}} y$, and conversely if, for any $y\in \attr(T_2)$, there exists a unique $x\in \attr(T_1)$ so that $x\mathrel{R_{\isom}} y$.

\begin{definition}\label{def:cipher_tree_isom}
A tree isomorphism $\isom$ is said to be a tree ciphering if and only if $\mathrel{R_{\isom}}$ is a bijection. The bijective function $f_{\isom}:\attr(T_1) \to \attr(T_2)$ is called a cipher.
\end{definition}

An example of tree ciphering is shown in Figure~\ref{fig:tree_ciphering_example}.

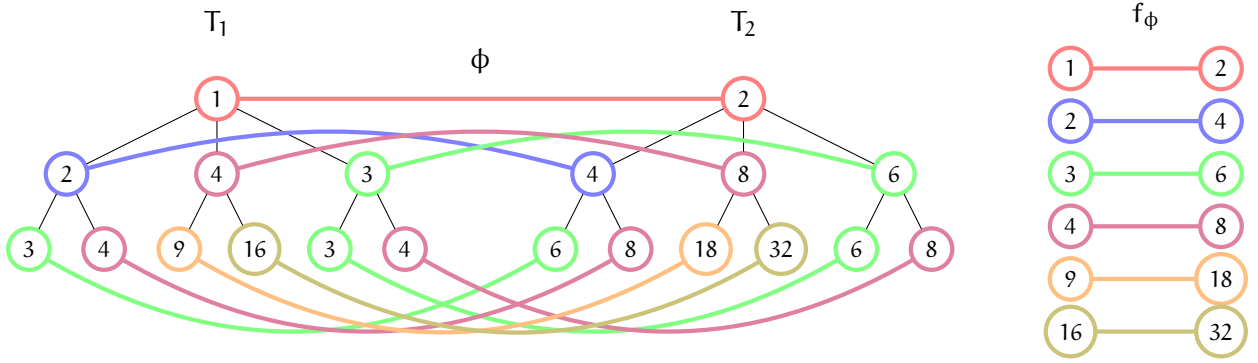
\begin{figure}[h]
\def\xscale{1}
\def\yscale{1}
\def\nodescale{0.8}
\centering
\begin{tikzpicture}[xscale=\xscale,yscale=\yscale]
\tikzstyle{noeud}=[draw,circle,fill=white,scale=\nodescale*1, ultra thick]
\tikzstyle{fleche}=[-,>=latex,red,ultra thick]

\node[noeud,draw=lred] (u1) at (-3.5,2) {$1$};
\node[noeud,draw=lred] (u2) at (3.5,2) {$2$};

\node[noeud,draw=lblue] (u3) at (-5.5,1) {$2$};
\node[noeud,draw=lpurple] (u4) at (-3.5,1) {$4$};
\node[noeud,draw=lgreen] (u5) at (-1.5,1) {$3$};
\node[noeud,draw=lblue] (u6) at (1.5,1) {$4$};
\node[noeud,draw=lpurple] (u7) at (3.5,1) {$8$};
\node[noeud,draw=lgreen] (u8) at (5.5,1) {$6$};

\node[noeud,draw=lgreen] (u9) at (-6,0) {$3$};
\node[noeud,draw=lpurple] (u10) at (-5,0) {$4$};
\node[noeud,draw=lorange] (u11) at (-4,0) {$9$};
\node[noeud,draw=lolive] (u12) at (-3,0) {$16$};
\node[noeud,draw=lgreen] (u13) at (-2,0) {$3$};
\node[noeud,draw=lpurple] (u14) at (-1,0) {$4$};

\node[noeud,draw=lgreen] (u15) at (1,0) {$6$};
\node[noeud,draw=lpurple] (u16) at (2,0) {$8$};
\node[noeud,draw=lorange] (u17) at (3,0) {$18$};
\node[noeud,draw=lolive] (u18) at (4,0) {$32$};
\node[noeud,draw=lgreen] (u19) at (5,0) {$6$};
\node[noeud,draw=lpurple] (u20) at (6,0) {$8$};

\node (t1) at (-3.5,3) {$T_1$};
\node (t2) at (3.5,3) {$T_2$};

\draw (u1)--(u3);
\draw (u1)--(u4);
\draw (u1)--(u5);
\draw (u2)--(u6);
\draw (u2)--(u7);
\draw (u2)--(u8);
\draw (u3)--(u9);
\draw (u3)--(u10);
\draw (u4)--(u11);
\draw (u4)--(u12);
\draw (u5)--(u13);
\draw (u5)--(u14);
\draw (u6)--(u15);
\draw (u6)--(u16);
\draw (u7)--(u17);
\draw (u7)--(u18);
\draw (u8)--(u19);
\draw (u8)--(u20);

\node (phi) at (0,2.5) {$\isom$};

\draw[fleche,lred] (u1)--(u2);
\draw[fleche,lblue] (u3) to[bend left=15](u6);
\draw[fleche,lgreen] (u5) to[bend left=15](u8);
\draw[fleche,lgreen] (u9) to[bend left=-30](u15);
\draw[fleche,lgreen] (u13) to[bend left=-30](u19);

\draw[fleche,lpurple] (u4) to[bend left=15](u7);
\draw[fleche,lpurple] (u14) to[bend left=-30](u20);
\draw[fleche,lpurple] (u10) to[bend left=-30](u16);

\draw[fleche,lorange] (u11) to[bend left=-30](u17);

\draw[fleche,lolive] (u12) to[bend left=-30](u18);

\end{tikzpicture}\hfill
\begin{tikzpicture}[xscale=\xscale,yscale=0.7*\yscale]
\tikzstyle{noeud}=[draw,ultra thick,circle,fill=white,scale=\nodescale*1]
\tikzstyle{fleche}=[-,>=latex,ultra thick]

\node[noeud,draw=lred] (u1) at (0,0) {$1$};
\node[noeud,draw=lblue] (u2) at (0,-1) {$2$};
\node[noeud,draw=lgreen] (u3) at (0,-2) {$3$};
\node[noeud,draw=lpurple] (u4) at (0,-3) {$4$};
\node[noeud,draw=lorange] (u9) at (0,-4) {$9$};
\node[noeud,draw=lolive] (u16) at (0,-5) {$16$};

\node[noeud,draw=lred] (v2) at (2,0) {$2$};
\node[noeud,draw=lblue] (v4) at (2,-1) {$4$};
\node[noeud,draw=lgreen] (v6) at (2,-2) {$6$};
\node[noeud,draw=lpurple] (v8) at (2,-3) {$8$};
\node[noeud,draw=lorange] (v18) at (2,-4) {$18$};
\node[noeud,draw=lolive] (v32) at (2,-5) {$32$};

\node (t1) at (1,1) {$f_{\isom}$};

\draw[fleche,lred] (u1)--(v2);
\draw[fleche,lblue] (u2)--(v4);
\draw[fleche,lgreen] (u3)--(v6);
\draw[fleche,lpurple] (u4)--(v8);
\draw[fleche,lorange] (u9)--(v18);
\draw[fleche,lolive] (u16)--(v32);

\end{tikzpicture}
\caption{Two topologically isomorphic labelled trees $T_1$ and $T_2$, as well as an example of tree isomorphism $\isom$ between them (left). $\isom$ is also a tree ciphering, as the binary relation $R_{\isom}$ is bijective and induces a cipher $f_{\isom}$ (right). Nodes and mapping $\isom$ are colored according to the corresponding relations between labels in $f_{\isom}$. Note that $\isom$ is not the only tree isomorphism that yields a tree ciphering for these particular trees. }
\label{fig:tree_ciphering_example}
\end{figure}

The existence of a tree ciphering isomorphism induces an equivalence relation on the set of trees \cite{ingels2021isomorphic} which we denote $\sim$. By definition, relation $\sim$ is more stringent than $\simeq$ but also less stringent than $\simeq_l$, i.e
\begin{equation}\label{eq:eqrels}
T_1\simeq_l T_2\,\Rightarrow\,T_1\sim T_2\,\Rightarrow\,T_1\simeq T_2.\end{equation}

\subsection{Objectives of the paper}

As aforementioned, $\simeq_l$ is less flexible than $\simeq$ which implies that corresponding patterns are less frequent and can miss some forms of similarity. However, $\simeq_l$ keeps the piece of information carried by the labels. In light of \eqref{eq:eqrels}, equivalence relation $\sim$ is somewhere in-between making related patterns, i.e. subtrees with identical label distribution, of great interest for tree data analysis. If it exists, the compression scheme $\red_\sim$ based on $\sim$ should be a good solution for parsimonious enumeration of such patterns, and therefore for frequent pattern mining. In addition, it should achieve a better compression level than $\red_{\simeq_l}$ while preserving the labels contrary to $\red_{\simeq}$.

However, before addressing the issue of compression, it is important to note that the tree ciphering isomorphism problem (which consists in determining if a tree ciphering exists between two trees) itself is far from trivial: referred to as marked tree isomorphism problem, it is graph isomorphism complete \cite{booth1979problems}, i.e. as difficult as the graph isomorphism problem, to which no polynomial algorithm nor proof of NP-completeness is known \cite{schoning1988}, although a quasi-polynomial algorithm was recently exhibited \cite{babai2016graph}.

In addition, one can already expect from $\red_\sim$ a more complex conceptual definition than for $\red_{\simeq_l}$ and $\red_\simeq$. Indeed, the ciphers involved in the tree ciphering isomorphisms will have to be transferred to the compression version to allow the lossless reconstruction of the tree and its labels. However, ciphers correspond to pairwise comparisons of subtrees, for which there is no equivalent in DAG compressions $\red_{\simeq_l}$ and $\red_{\simeq}$. In particular, there is not a single cipher per equivalence class contrary to labels in $\red_{\simeq_l}$.

This paper aims to develop DAG-RW compression, a new compression scheme that generalises DAG compression to the equivalence relation $\sim$ (RW stands for rewritable to indicate that the labels are evaluated through a cipher) as an algorithmic solution for detection of frequent subtrees with identical pattern distribution. For this purpose, we address the following points.

\paragraph{Algorithm for tree ciphering isomorphism} If one wants to establish a compression method based on tree ciphering isomorphism, one must be able to quickly compute the corresponding equivalence classes, and this while the problem is graph isomorphism complete.
    
The reduction from a marked tree to a graph being linear \cite{booth1979problems}, we could consider using graph isomorphism resolution methods rather than algorithms built to directly deal with marked trees. Even if there exist graph isomorphism algorithms that are very efficient in practice \cite{mckay2014practical}, the majority of them do not construct an isomorphism (while they are required in our approach), but seek to compute a certificate of non-isomorphism, often based on vertex invariants. For example, in Weisfeiler-Lehman algorithms, also known as color refinement algorithms \cite{weisfeiler1968reduction}, graphs are colored according to some rules, and the color histograms are compared afterwards: if they diverge, the graphs are not isomorphic. However, the test is incomplete in the sense that there exist non-isomorphic graphs that share the same coloring. The distinguishability of those algorithms is constantly improved \cite{grohe2021deep} but does not yet answer the problem for any graph. Interestingly, the histograms can be used as the initial partition of vertices of a backtracking procedure that aims to exhibit an isomorphism, as in \cite[3.3 Solving the isomorphism problem directly]{fortin1996graph}: the invariants reduce the number of candidates that can be mapped with a given vertex.

\begin{quote}
\textit{It turns out that many isomorphism problems can be directly solved in less time than it takes to calculate the more powerful invariants (...).} --- Scott Fortin \cite[3.2 Vertex invariants]{fortin1996graph}
\end{quote}

In the present paper, we take this piece of advice on board and propose an efficient search algorithm to solve the tree ciphering isomorphism problem directly from the objects of interest (which also has the benefit of avoiding the reduction from marked trees to graphs): given two trees $T_1$ and $T_2$, it explores the huge space of tree isomorphisms between $T_1$ and $T_2$ and checks if they induce a tree ciphering or not. The algorithm processes in two main steps, both based on invariants developed specifically for marked trees: first, a search space reduction (which relies on previous work \cite{ingels2021isomorphic}) is performed using the relations between tree cipherings and ciphers; second, the remaining search space is explored using a sophisticated backtracking strategy.
    
\paragraph{DAG-RW compression} We define the original concept of DAG-RW compression, which is lossless with respect to both the topology and the labels, and theoretically achieves better compression rates than conventional DAG compression with labels. Furthermore, we use the aforementioned search algorithm for tree ciphering isomorphism to propose a compression algorithm, which we also analyse the time-complexity. Our compression method does not depend on this specific answer to the tree ciphering isomorphism problem but requires the computation of a tree ciphering and related cipher when subtrees are isomorphic, which drastically limits the contributions on graph isomorphism available in the literature and argues for the development of a new approach.

\paragraph{Real data analysis} On top of this, we illustrate on two real datasets that DAG-RW compression can be used to detect frequent common subtrees with identical label distribution. We compare the results with labelled and unlabelled subtrees and show that this new type of patterns is relevant because it captures new forms of similarity, making DAG-RW compression a serious candidate to perform parsimonious analysis of tree data.

The paper is organised as follows. Forthcoming Subsection~\ref{ss:nots} sums up the notation and vocabulary used in the rest of the article. Our algorithm for tree ciphering isomorphism is investigated in Section~\ref{s:treeciphering}. Section~\ref{s:dagrw} is devoted to DAG-RW compression: concept, algorithm, and application to frequent pattern mining on two datasets of reference. Appendixes~\ref{app:size_search_size} and \ref{proof:backtracking_tree} present technical details and proofs related to the algorithm for tree ciphering isomorphism. All the methods and algorithms developed in this paper were implemented in the \texttt{Python} library \texttt{treex} \cite{azais2019treex}, which was used in the numerical sections of the paper.

\subsection{Notation and vocabulary}
\label{ss:nots}

\paragraph{Labelled rooted trees} A rooted tree $T$ is a connected directed graph without undirected cycle and such that there exists a special node $\troot(T)$, which has no parent and such that any node $v\neq\troot(T)$ has exactly one parent: $\parent(v)$. $\children(v)$ denotes the set of nodes that share $v$ as parent. With a slight abuse of notation, ``$v$ is a node of $T$'' is denoted $v\in T$. $\leaves(T)$ is the set of nodes of $T$ without any children. $T(v)$ is the subtree of $T$ formed by $v$ and all its descendants in $T$. Finally, each node $v$ of $T$ carries a label $\overline{v}$, whose set, called alphabet, is denoted $\attr(T)$, i.e. $\attr(T) = \{\overline{v}\,:\,v\in T\}$.

\paragraph{Isomorphisms} Let $T_1$ and $T_2$ be two labelled rooted trees. The paper considers three types of tree isomorphisms: unlabelled and labelled tree isomorphisms given in Definition~\ref{def:isom} and tree ciphering given in Definition~\ref{def:cipher_tree_isom}.
\begin{itemize}
    \item $T_1\simeq T_2$: there exists an unlabelled tree isomorphism between $T_1$ and $T_2$, which means that they share the same topology without any information about their labels.
    \item $T_1\sim T_2$: there exists a tree ciphering $\isom$ from $T_1$ into $T_2$, meaning that they share the same topology and the same labels up to a cipher $f_{\isom}$, i.e. $\overline{\isom(u)} = f_{\isom}( \overline{u} )$ for any $u\in T_1$.
    \item $T_1\simeq_l T_2$: there exists a labelled tree isomorphism $\isom$ from $T_1$ into $T_2$, meaning that they share the same topology and the same labels, i.e. $\overline{\isom(u)} = \overline{u}$, for any $u\in T_1$.
\end{itemize}
The inclusion relation between these isomorphisms is given in \eqref{eq:eqrels}. Related equivalence classes are denoted $[T]_\ast$, $\ast\in\{\simeq,\sim,\simeq_l\}$. When there is no ambiguity on the underlying tree $T$, equivalence classes of subtrees of $T$ are noted in short $[u]_\ast$ instead of $[T(u)]_\ast$.

\paragraph{Compression} $\red_\ast(T)$ denotes the DAG compression of the tree $T$, which is described in Subsection~\ref{ss:intro:dagcompression} for $\ast\in\{\simeq,\simeq_l\}$ and remains to be introduced for $\ast=\,\sim$. We recall that $\red_\simeq(T)$ compresses the topology of $T$ but loses the pieces of information carried by the labels of $T$, while $\red_{\simeq_l}(T)$ is a lossless compression of $T$.

\paragraph{Topological properties}  The depth of a node $u$ is the length of the path from the root of $T$ to $u$,
$$ \depth(u) = 
\left\{\begin{array}{ll}
0 & \text{if }u=\troot(T),\\
\displaystyle 1+\depth(\parent(u)) & \text{otherwise.}
\end{array}\right.
$$
The height $\theight(T)$ and the degree $\degr(T)$ of $T$ can then be defined as the maximal depth and the maximal number of children of its nodes, i.e.
$$\theight(T) = \max_{u\in T}\,\depth(u)\qquad\text{and}\qquad \degr(T) = \max_{u\in T}\,\#\children(u).$$

\section{Algorithm for tree ciphering isomorphism}
\label{s:treeciphering}

This section presents our algorithm to find, if it exists, a tree ciphering between two trees $T_1$ and $T_2$. Exploring tree isomorphisms, the algorithm reduces the search space using back and forth deductions on the topology and the labels. This part relies on and improves our previous work \cite{ingels2021isomorphic}. Then the remaining space is explored using a backtracking strategy.

\subsection{Incremental construction of a tree ciphering}

Starting from trivial objects, the algorithm incrementally builds, if possible, a bijection $\isom$ on the nodes of the trees and a bijection $f$ on the labels of the nodes. At any step of the algorithm:
\begin{itemize}
    \item $\isom$ maps a subset of nodes of tree $T_1$ to a subset of nodes of tree $T_2$ so that, if $u$ is a child of $v$ in $T_1$ and both are in the domain of $\isom$, then $\isom(u)$ is a child of $\isom(v)$ in $T_2$;
    \item $f$ maps a subset of labels of $\attr(T_1)$ to a subset of labels of $\attr(T_2)$ so that $f(\overline{u})=\overline{v}$ as soon as $\isom(u)=v$.
\end{itemize}
If the algorithm goes until the end, i.e. if $\dom(\isom)=T_1$ and $\dom(f)=\attr(T_1)$, then, in light of Definitions~\ref{def:isom} and \ref{def:cipher_tree_isom}, $\isom$ and $f$ correspond to a tree ciphering and related cipher. The incremental update of the domains of $\isom$ and $f$ is performed using the following method.

\begin{minipage}[c]{0.53\textwidth}
Let $\psi : A_\psi \to B_\psi$ be a bijection with $A_\psi\subset A$ and $B_\psi\subset B$, and $(a,b)\in A\times B$. If $a\in A_\psi$ and $\psi(a)=b$ or $a\notin A_\psi$ and $b\notin B_\psi$, then $\psi$ can be extended as the bijection $\psi_{(a,b)}$ from $A_\psi\cup \lbrace a\rbrace$ into $B_\psi\cup\lbrace b\rbrace$ as
$$ \psi_{(a,b)}(\alpha) =\begin{cases}
\psi(\alpha) & \text{if } \alpha\in A_\psi, \\
b & \text{if } \alpha=a.
\end{cases}$$
Algorithm~\ref{algo:extbij} tests the extension condition and extends the bijection when possible.
\end{minipage}\hfill
\begin{minipage}[c]{0.42\textwidth}
\begin{algorithm}[H]
\caption{\textsc{ExtBij}}
\label{algo:extbij}
\KwIn{$\psi: A_\psi\to B_\psi,a\in A,b\in B$}
\KwOut{$\top$ if $\psi$ can be extended}
\uIf{$a\in A_\psi$ {\bf and} $\psi(a)=b$}{
\Return $\top$
}{
\uElseIf{$a\notin A_\psi$ {\bf and} $b\notin B_\psi$}
{
$\psi \gets \psi_{(a,b)}$\;
\Return $\top$
}
\Else{
\Return $\bot$
}
}
\end{algorithm}
\end{minipage}

Now the question is to wisely propose elements to extend the domains of $\isom$ and $f$. We proceed in three main steps: deductions from topology given in Subsection~\ref{ss:ded:topol}, deductions from labels given in Subsection~\ref{ss:ded:label}, backtracking given in Subsection~\ref{ss:backtracking}.

\subsection{Deductions from topology}
\label{ss:ded:topol}

\paragraph{Partition of nodes into bags} We assume here that some nodes and labels have already been mapped through $\isom$ and $f$ and aim to increase their domains. To this end, nodes of $T_1$ and $T_2$ that are not mapped via $\isom$ can be sorted into bags.

\begin{definition}\label{def:bag}
A bag $\bfB$ is a couple $(B_1,B_2)$, with $B_i\subset T_i$, so that $\#B_1=\#B_2$, which is denoted $\#\bfB$.
\end{definition}

A bag $\bfB$ contains nodes of $T_1$ and $T_2$ which we know (in a trivial way or by intricate deductions) must be mapped together to satisfy the conditions of a tree ciphering. For instance, if $\isom(u) = v$, then $(\children(u),\children(v))$ forms a bag: $\isom$ can not be a tree isomorphism if a child of $u$ is not mapped with a child of $v$. Another elementary bag is $(\leaves(T_1),\leaves(T_2))$. At this point, we assume that the nodes of $T_1$ and $T_2$ not mapped via $\isom$ are sorted into bags, which set is denoted $\bags$. The number of possible completions of $\isom$ is therefore given by
\begin{equation}\label{size_bags_only}
    N(\bags) = \prod_{\bfB\in\bags} \#\bfB!,
\end{equation}
which evaluates the size of the remaining search space.

If a bag $\bfB$ is of size $1$, i.e. $\bfB=(u,v)\in T_1\times T_2$, then one can call Algorithm~\ref{algo:extbij} to try to extend $\isom$ and $f$ with $\isom(u)=v$ and $f(\overline{u}) = \overline{v}$. If Algorithm~\ref{algo:extbij} returns $\bot$, then it means $T_1\nsim T_2$. Otherwise, the domains of $\isom$ and $f$ can be extended and this new piece of information can be used to refine the partition of the nodes into finer bags. The recursive improvements in bags due to such a mapping attempt are performed by function \textsc{MapNodes} given in the next paragraph.

\begin{deduction}\label{ded:bags}
As long as there exist bags of the form $\bfB=(u,v)\in T_1\times T_2$, call \textsc{MapNodes}$(u,v,\isom,f)$, and delete $\bfB$ from $\bags$.
\end{deduction}

\paragraph{Mapping nodes} If two nodes $u \in T_1$ and $v\in T_2$ can be mapped as $\isom(u)=v$, i.e. if $\textsc{ExtBij}(f,\overline{u},\overline{v})$ and $\textsc{ExtBij}(\isom,u,v)$ return $\top$, then (i) $\isom(\parent(u)) = \parent(v)$, (ii) the elements of $\children(u)$ must be mapped with those of $\children(v)$, and (iii) $f(\overline{u})=\overline{v}$. These observations are compiled into Algorithm~\ref{algo:mapnodes}, which is used whenever nodes should be mapped together. It uses the subroutine \textsc{SplitChildren}, given in Algorithm~\ref{algo:split_children} and illustrated in Figure~\ref{fig:bag_exple}, to address condition (ii), and recursively separate the children of the mapped nodes from the other nodes with which they are partitioned. The recursive nature of \textsc{SplitChildren} represents a significant improvement compared to \cite[3.3 Mapping nodes]{ingels2021isomorphic}. It should be noted that, if at any step \textsc{MapNodes} returns $\bot$, then we can conclude that $T_1$ and $T_2$ are not isomorphic with respect to $\sim$.

\begin{figure}[h!]

\begin{minipage}[c]{0.62\textwidth}
\begin{algorithm}[H]
\caption{\textsc{MapNodes}}\label{algo:mapnodes}

\SetKw{KwAnd}{and}
\KwIn{$u\in T_1,v\in T_2,\isom,f$}
\eIf{\textup{\textsc{ExtBij}}$(f,\overline{u},\overline{v})$ \KwAnd \textup{\textsc{ExtBij}}$(\isom,u,v)$}{
\textsc{SplitChildren}$(\lbrace u\rbrace,\lbrace v\rbrace)$\\
\eIf{$\isom(\parent(u))=\parent(v)$}{\Return $\top$}{
\Return \textsc{MapNodes}($\parent(u),\parent(v),\isom,f$)}
}{
\Return $\bot$
}
\end{algorithm}

   \begin{algorithm}[H]
\caption{\textsc{SplitChildren}}\label{algo:split_children}
\KwIn{$S_u,S_v$}
\SetKw{KwAnd}{and}
\SetKw{KwST}{so that}
$\mathcal{C}_u \gets \bigcup_{w\in S_u} \children(w)$\\
$\mathcal{C}_v \gets \bigcup_{w\in S_v} \children(w)$\\
\For{$\bfB\in\bags$}{
$B_u\gets B_1\cap \mathcal{C}_u$\\
$B_v\gets B_2\cap \mathcal{C}_v$\\
\If{$B_u\neq\emptyset$ \KwAnd $B_1\setminus B_u\neq \emptyset$}{
Add $(B_u,B_v)$ and $(B_1\setminus B_u, B_2\setminus B_v)$ to $\bags$\\
\textsc{SplitChildren}$(B_u,B_v)$\\
\textsc{SplitChildren}$(B_1\setminus B_u,\allowbreak B_2\setminus B_v)$\\
Delete $\bfB$ from $\bags$
}
}
\end{algorithm}
\end{minipage}\hfill
\begin{minipage}[c]{0.33\textwidth}
\centering
\def\xscale{0.5}
\def\yscale{0.5}
\def\nodescale{0.8}

\begin{tikzpicture}[xscale=\xscale,yscale=\yscale]
\tikzstyle{noeud}=[draw,circle,fill=white,scale=\nodescale*1]
\tikzstyle{attribut}=[scale=\nodescale*1,font=\bf]
\tikzstyle{arc}=[-,>=latex]
\tikzstyle{fleche}=[->,>=latex,red,ultra thick]

\node[noeud] (u1) at (0,0) {};
\node[noeud] (u2) at (1,0) {};
\node[noeud] (u3) at (2,0) {};
\node[noeud] (u4) at (3,0) {};

\node[noeud,ultra thick,draw=red] (u6) at (1,2) {$u$};
\draw[arc] (u6)--(u1);
\draw[arc] (u6)--(u2);
\draw[arc] (u6)--(u3);

\node[noeud] (u1b) at (0,-2) {};
\node[noeud] (u2b) at (1,-2) {};
\node[noeud] (u3b) at (2,-2) {};
\node[noeud] (u4b) at (3,-2) {};
\draw[arc] (u1)--(u1b);
\draw[arc] (u2)--(u2b);
\draw[arc] (u3)--(u3b);
\draw[arc] (u4)--(u4b);

\node[noeud] (v1) at (5,0) {};
\node[noeud] (v2) at (6,0) {};
\node[noeud] (v3) at (7,0) {};
\node[noeud] (v4) at (8,0) {};

\node[noeud] (v1b) at (5,-2) {};
\node[noeud] (v2b) at (6,-2) {};
\node[noeud] (v3b) at (7,-2) {};
\node[noeud] (v4b) at (8,-2) {};
\draw[arc] (v1)--(v1b);
\draw[arc] (v2)--(v2b);
\draw[arc] (v3)--(v3b);
\draw[arc] (v4)--(v4b);

\node[noeud,ultra thick,draw=red] (v6) at (6,2) {$v$};

\draw[fleche] (u6)--(v6) node[midway,above] {$\isom$};
\draw[arc] (v6)--(v1);
\draw[arc] (v6)--(v2);
\draw[arc] (v6)--(v3);

\draw[rounded corners,ultra thick,gray] (-.5,-.5) rectangle (3.5,.5) node[above] {$B_1$};

\draw[ultra thick,gray] (3.5,0)--(4.5,0);

\draw[rounded corners,ultra thick,gray] (4.5,-.5) rectangle (8.5,.5) node[above] {$B_2$};

\draw[rounded corners,ultra thick,gray] (-.5,-2.5) rectangle (3.5,-1.5); %node[below] {$B'_1$};

\draw[ultra thick,gray] (3.5,0-2)--(4.5,-2);

\draw[rounded corners,ultra thick,gray] (4.5,-2.5) rectangle (8.5,-1.5); %node[below] {$B'_2$};

\def\yshift{-2.5}

\node[noeud] (u1b) at (0,-2.5+\yshift) {};
\node[noeud] (u2b) at (1,-2.5+\yshift) {};
\node[noeud] (u3b) at (2,-2.5+\yshift) {};
\node[noeud] (u4b) at (3,-2.5+\yshift) {};

\node[noeud] (v1b) at (5,-2.5+\yshift) {};
\node[noeud] (v2b) at (6,-2.5+\yshift) {};
\node[noeud] (v3b) at (7,-2.5+\yshift) {};
\node[noeud] (v4b) at (8,-2.5+\yshift) {};

\draw[rounded corners,ultra thick,gray] (-.5,-3+\yshift) rectangle (2.5,-2+\yshift) ;

\node[gray] at (1,-1.5+\yshift) {$B_u$} ;

\draw[ultra thick,gray] (1,-3+\yshift) to[bend right=30] (6,-3+\yshift);

\draw[rounded corners,ultra thick,gray] (4.5,-3+\yshift) rectangle (7.5,-2+\yshift);

\node[gray] at (6,-1.5+\yshift) {$B_v$};

\draw[rounded corners,ultra thick,gray] (2.5,-3+\yshift) rectangle (3.5,-2+\yshift);

\node[gray] at (3.5,-1.5+\yshift) {$B_1\setminus B_u$} ;

\draw[ultra thick,gray] (3,-3+\yshift) to[bend right=30] (8,-3+\yshift);

\draw[rounded corners,ultra thick,gray] (7.5,-3+\yshift) rectangle (8.5,-2+\yshift);

\node[gray] at (8.5,-1.5+\yshift) {$B_2\setminus B_v$} ;

\draw[->,ultra thick,black] (9,-1)to[bend left=45] (9,-3.5+\yshift);

\node[noeud] (u1bb) at (0,-4.5+\yshift) {};
\node[noeud] (u2bb) at (1,-4.5+\yshift) {};
\node[noeud] (u3bb) at (2,-4.5+\yshift) {};
\node[noeud] (u4bb) at (3,-4.5+\yshift) {};

\node[noeud] (v1bb) at (5,-4.5+\yshift) {};
\node[noeud] (v2bb) at (6,-4.5+\yshift) {};
\node[noeud] (v3bb) at (7,-4.5+\yshift) {};
\node[noeud] (v4bb) at (8,-4.5+\yshift) {};

\draw[arc] (u1b)--(u1bb);
\draw[arc] (u2b)--(u2bb);
\draw[arc] (u3b)--(u3bb);
\draw[arc] (u4b)--(u4bb);
\draw[arc] (v1b)--(v1bb);
\draw[arc] (v2b)--(v2bb);
\draw[arc] (v3b)--(v3bb);
\draw[arc] (v4b)--(v4bb);

\draw[rounded corners,ultra thick,gray] (-.5,-5+\yshift) rectangle (2.5,-4+\yshift) ;

% \node[gray] at (1,-3.5+\yshift) {$B_u$} ;

\draw[ultra thick,gray] (1,-5+\yshift) to[bend right=30] (6,-5+\yshift);

\draw[rounded corners,ultra thick,gray] (4.5,-5+\yshift) rectangle (7.5,-4+\yshift);

% \node[gray] at (6,-3.5+\yshift) {$B_v$};

\draw[rounded corners,ultra thick,gray] (2.5,-5+\yshift) rectangle (3.5,-4+\yshift) ;%node[above] {$B_1\setminus B_u$};

\draw[ultra thick,gray] (3,-5+\yshift) to[bend right=30] (8,-5+\yshift);

\draw[rounded corners,ultra thick,gray] (7.5,-5+\yshift) rectangle (8.5,-4+\yshift) ;%node[above] {$B_2\setminus B_v$};
\end{tikzpicture}
\captionof{figure}{Illustration of \textsc{SplitChildren}. The bag $(B_1,B_2)$ (above) is split into two bags after the mapping of $u$ and $v$, one composed of the children of $u$ and $v$, and the other with the remaining nodes (below). Recursively, the bags of their descendants are also split. %
} \label{fig:bag_exple}
\end{minipage}
\end{figure}

\paragraph{Reduction of search space} To determine whether $T_1\sim T_2$, it is required that $T_1\simeq T_2$, which we can check in linear time \cite[Theorem~3.3]{aho1974design} together with assigning to the subtrees of $T_1$ and $T_2$ their equivalence class. The following steps are then taken in order. They will be illustrated on a followed example, beginning with Figure~\ref{preproc_exple}. In addition, Appendix~\ref{app:size_search_size} indicates how each of the manipulations of the system theoretically reduces the size of the search space.

\textit{A. Histogram of labels.} We first proceed to a verification (computable in linear time) which allows to eliminate some obvious cases where $T_1\nsim T_2$. In a traversal of $T_i$, we can count the number $N_i(a)$ of occurrences of each label $a\in\attr(T_i)$. Then we consider $H_i : n\in\mathbb{N} \mapsto \lbrace a\in \attr(T_i) : N_i(a)=n\rbrace$, which represents the histogram of labels of $T_i$. If the two histograms do not coincide, i.e. if there exists $n\in\mathbb{N}$ such that $\#H_1(n)\neq \#H_2(n)$, then we can conclude (in linear time) that $T_1\nsim T_2$. Otherwise, $\supp(H)$ denotes the support of the histograms,
$$\supp(H) = \left\{n\in\mathbb{N}~:~\#H_1(n)=\#H_2(n)>0\right\}.$$
We initialise the different parameters of the algorithm as follows: $\dom(\isom)=\dom(f)=\emptyset$ and all the nodes of $T_1$ and $T_2$ have to be sorted into bags, for which we work from the histograms of labels. Since it is not possible to map nodes whose labels do not have the same number of occurrences in their respective tree, the set $\bags$ is formed of bags $\bfB(n)$, $n\in\supp(H)$, such that
$$B_i(n) = \lbrace u\in T_i~:~\overline{u}\in H_i(n)\rbrace.$$
See Figure~\ref{preproc_exple} for an example. It should be noted that, in the case where all labels appear the same number of times each, we would then have a single bag containing all the nodes.

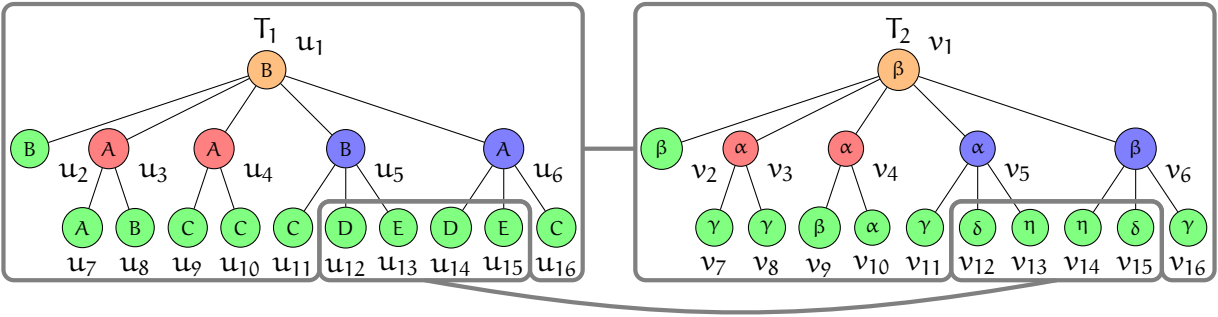
\begin{figure}[h!]
\centering
\def\xscale{0.35}
\def\yscale{0.7}
\def\nodescale{0.7}
\begin{tikzpicture}[xscale=\xscale,yscale=\yscale]
\tikzstyle{noeud}=[draw,circle,fill=white,scale=\nodescale*1]
\tikzstyle{attribut}=[scale=\nodescale*1,font=\bf]
\tikzstyle{arc}=[-,>=latex]
\tikzstyle{fleche}=[->,>=latex,red,thick]

\def\x{2}
\def\y{1.5}

\node[noeud,fill=lorange,label=15:$u_1$] (u1) at (4.5*\x,2*\y) {$B$};
\node[noeud,fill=lgreen,label=-15:$u_2$] (u2) at (0,1*\y) {$B$};
\node[noeud,fill=lred,label=-15:$u_3$] (u3) at (1.5*\x,1*\y) {$A$};
\node[noeud,fill=lred,label=-15:$u_4$] (u4) at (3.5*\x,1*\y) {$A$};
\node[noeud,fill=lblue,label=-15:$u_5$] (u5) at (6*\x,1*\y) {$B$};
\node[noeud,fill=lblue,label=-15:$u_6$] (u6) at (9*\x,1*\y) {$A$};
\node[noeud,fill=lgreen,label=-90:$u_7$] (u7) at (1*\x,0) {$A$};
\node[noeud,fill=lgreen,label=-90:$u_8$] (u8) at (2*\x,0) {$B$};
\node[noeud,fill=lgreen,label=-90:$u_9$] (u9) at (3*\x,0) {$C$};
\node[noeud,fill=lgreen,label=-90:$u_{10}$] (u10) at (4*\x,0) {$C$};
\node[noeud,fill=lgreen,label=-90:$u_{11}$] (u11) at (5*\x,0) {$C$};
\node[noeud,fill=lgreen,label=-90:$u_{12}$] (u12) at (6*\x,0) {$D$};
\node[noeud,fill=lgreen,label=-90:$u_{13}$] (u13) at (7*\x,0) {$E$};
\node[noeud,fill=lgreen,label=-90:$u_{14}$] (u14) at (8*\x,0) {$D$};
\node[noeud,fill=lgreen,label=-90:$u_{15}$] (u15) at (9*\x,0) {$E$};
\node[noeud,fill=lgreen,label=-90:$u_{16}$] (u16) at (10*\x,0) {$C$};

\node (t1) at (4.5*\x,2.5*\y) {$T_1$};

\draw[arc] (u1)--(u2) ;
\draw[arc] (u1)--(u3) ;
\draw[arc] (u1)--(u4) ;
\draw[arc] (u1)--(u5) ;
\draw[arc] (u1)--(u6) ;
\draw[arc] (u3)--(u7) ;
\draw[arc] (u3)--(u8) ;
\draw[arc] (u4)--(u9) ;
\draw[arc] (u4)--(u10) ;
\draw[arc] (u5)--(u11) ;
\draw[arc] (u5)--(u12) ;
\draw[arc] (u5)--(u13) ;
\draw[arc] (u6)--(u14) ;
\draw[arc] (u6)--(u15) ;
\draw[arc] (u6)--(u16) ;

\def\xshift{12*\x}

\node[noeud,fill=lorange,label=15:$v_1$] (u1) at (4.5*\x+\xshift,2*\y) {$\beta$};
\node[noeud,fill=lgreen,label=-15:$v_2$] (u2) at (0+\xshift,1*\y) {$\beta$};
\node[noeud,fill=lred,label=-15:$v_3$] (u3) at (1.5*\x+\xshift,1*\y) {$\alpha$};
\node[noeud,fill=lred,label=-15:$v_4$] (u4) at (3.5*\x+\xshift,1*\y) {$\alpha$};
\node[noeud,fill=lblue,label=-15:$v_5$] (u5) at (6*\x+\xshift,1*\y) {$\alpha$};
\node[noeud,fill=lblue,label=-15:$v_6$] (u6) at (9*\x+\xshift,1*\y) {$\beta$};
\node[noeud,fill=lgreen,label=-90:$v_7$] (u7) at (1*\x+\xshift,0) {$\gamma$};
\node[noeud,fill=lgreen,label=-90:$v_8$] (u8) at (2*\x+\xshift,0) {$\gamma$};
\node[noeud,fill=lgreen,label=-90:$v_9$] (u9) at (3*\x+\xshift,0) {$\beta$};
\node[noeud,fill=lgreen,label=-90:$v_{10}$] (u10) at (4*\x+\xshift,0) {$\alpha$};
\node[noeud,fill=lgreen,label=-90:$v_{11}$] (u11) at (5*\x+\xshift,0) {$\gamma$};
\node[noeud,fill=lgreen,label=-90:$v_{12}$] (u12) at (6*\x+\xshift,0) {$\delta$};
\node[noeud,fill=lgreen,label=-90:$v_{13}$] (u13) at (7*\x+\xshift,0) {$\eta$};
\node[noeud,fill=lgreen,label=-90:$v_{14}$] (u14) at (8*\x+\xshift,0) {$\eta$};
\node[noeud,fill=lgreen,label=-90:$v_{15}$] (u15) at (9*\x+\xshift,0) {$\delta$};
\node[noeud,fill=lgreen,label=-90:$v_{16}$] (u16) at (10*\x+\xshift,0) {$\gamma$};

\node (t2) at (4.5*\x+\xshift,2.5*\y) {$T_2$};

\draw[arc] (u1)--(u2) ;
\draw[arc] (u1)--(u3) ;
\draw[arc] (u1)--(u4) ;
\draw[arc] (u1)--(u5) ;
\draw[arc] (u1)--(u6) ;
\draw[arc] (u3)--(u7) ;
\draw[arc] (u3)--(u8) ;
\draw[arc] (u4)--(u9) ;
\draw[arc] (u4)--(u10) ;
\draw[arc] (u5)--(u11) ;
\draw[arc] (u5)--(u12) ;
\draw[arc] (u5)--(u13) ;
\draw[arc] (u6)--(u14) ;
\draw[arc] (u6)--(u15) ;
\draw[arc] (u6)--(u16) ;

\draw[ultra thick,rounded corners,gray] (6*\x-1,-1) rectangle (9*\x+1,.5);
\draw[ultra thick,rounded corners,gray] (\xshift+6*\x-1,-1) rectangle (\xshift+9*\x+1,.5);
\draw[ultra thick,gray] (7.5*\x,-1) to[bend right=5](\xshift+7.5*\x,-1);

\draw[ultra thick,rounded corners,gray] (-1,-1)--(-1,2.5*\y+.5)--(10*\x+1,2.5*\y+.5)--(10*\x+1,-1)--(9*\x+1,-1)--(9*\x+1,.5)--(6*\x-1,.5)--(6*\x-1,-1)--cycle;
\draw[ultra thick,rounded corners,gray] (\xshift-1,-1)--(\xshift-1,2.5*\y+.5)--(\xshift+10*\x+1,2.5*\y+.5)--(\xshift+10*\x+1,-1)--(\xshift+9*\x+1,-1)--(\xshift+9*\x+1,.5)--(\xshift+6*\x-1,.5)--(\xshift+6*\x-1,-1)--cycle;
\draw[ultra thick,gray] (10*\x+1,\y)--(\xshift-1,\y);

\end{tikzpicture}

\caption{Running example: histogram of labels (step 1/6). Two topologically isomorphic labelled trees $T_1$ (left) and $T_2$ (right). The color on nodes indicates the equivalence class under $\simeq$. The nodes have been numbered from $u_1$ to $u_{16}$ in $T_1$ (resp. from $v_1$ to $v_{16}$ in $T_2$) in breadth-first search order. The histograms of labels give $H_1(2)=\lbrace D,E\rbrace$, $H_1(4) = \lbrace A,B,C\rbrace$ and $H_2(2) =\lbrace \delta,\eta\rbrace$, $H_2(4) = \lbrace \alpha,\beta,\gamma\rbrace$. Since the histograms coincide, we initialise the bags (gray boxes) by grouping together the nodes whose labels appear with the same frequency. The initial size of the search space is $N(\bags)=12! \times 4! =11,496,038,400$.}
\label{preproc_exple}
\end{figure}

\FloatBarrier

\textit{B. Depth.} We distinguish nodes according to their depth: two mapped nodes must be at the same depth. For any $\bfB=(B_1,B_2)\in\bags$, we create the new bags $\bfB(d)$, $d\in[\![ 0,\theight(T_i)]\!]$, such that $B_i(d) = \lbrace u\in B_i : \depth(u)=d\rbrace$, then we remove $\bfB$ from $\bags$. After this operation, we apply Deduction~Rule~\ref{ded:bags}. Since the roots are the only nodes of depth zero, the rule will at least apply to their bag and they will be mapped together. See Figure~\ref{preproc_depth} for an example.

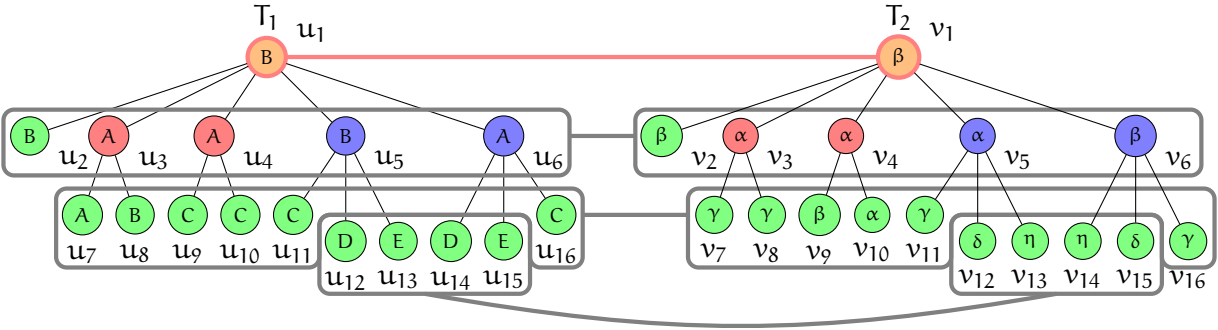
\begin{figure}[h!]
\centering
\def\xscale{0.35}
\def\yscale{0.7}
\def\nodescale{0.7}
\begin{tikzpicture}[xscale=\xscale,yscale=\yscale]
\tikzstyle{noeud}=[draw,circle,fill=white,scale=\nodescale*1]
\tikzstyle{attribut}=[scale=\nodescale*1,font=\bf]
\tikzstyle{arc}=[-,>=latex]
\tikzstyle{fleche}=[-,>=latex,red,ultra thick]

\def\x{2}
\def\y{1.5}

\node[noeud,fill=lorange,label=15:$u_1$,draw=lred, ultra thick] (u1) at (4.5*\x,2*\y) {$B$};
\node[noeud,fill=lgreen,label=-15:$u_2$] (u2) at (0,1*\y) {$B$};
\node[noeud,fill=lred,label=-15:$u_3$] (u3) at (1.5*\x,1*\y) {$A$};
\node[noeud,fill=lred,label=-15:$u_4$] (u4) at (3.5*\x,1*\y) {$A$};
\node[noeud,fill=lblue,label=-15:$u_5$] (u5) at (6*\x,1*\y) {$B$};
\node[noeud,fill=lblue,label=-15:$u_6$] (u6) at (9*\x,1*\y) {$A$};
\node[noeud,fill=lgreen,label=-90:$u_7$] (u7) at (1*\x,0) {$A$};
\node[noeud,fill=lgreen,label=-90:$u_8$] (u8) at (2*\x,0) {$B$};
\node[noeud,fill=lgreen,label=-90:$u_9$] (u9) at (3*\x,0) {$C$};
\node[noeud,fill=lgreen,label=-90:$u_{10}$] (u10) at (4*\x,0) {$C$};
\node[noeud,fill=lgreen,label=-90:$u_{11}$] (u11) at (5*\x,0) {$C$};
\node[noeud,fill=lgreen,label=-90:$u_{12}$] (u12) at (6*\x,-.5) {$D$};
\node[noeud,fill=lgreen,label=-90:$u_{13}$] (u13) at (7*\x,-.5) {$E$};
\node[noeud,fill=lgreen,label=-90:$u_{14}$] (u14) at (8*\x,-.5) {$D$};
\node[noeud,fill=lgreen,label=-90:$u_{15}$] (u15) at (9*\x,-.5) {$E$};
\node[noeud,fill=lgreen,label=-90:$u_{16}$] (u16) at (10*\x,0) {$C$};

\node (t1) at (4.5*\x,2.5*\y) {$T_1$};

\draw[arc] (u1)--(u2) ;
\draw[arc] (u1)--(u3) ;
\draw[arc] (u1)--(u4) ;
\draw[arc] (u1)--(u5) ;
\draw[arc] (u1)--(u6) ;
\draw[arc] (u3)--(u7) ;
\draw[arc] (u3)--(u8) ;
\draw[arc] (u4)--(u9) ;
\draw[arc] (u4)--(u10) ;
\draw[arc] (u5)--(u11) ;
\draw[arc] (u5)--(u12) ;
\draw[arc] (u5)--(u13) ;
\draw[arc] (u6)--(u14) ;
\draw[arc] (u6)--(u15) ;
\draw[arc] (u6)--(u16) ;

\def\xshift{12*\x}

\draw[ultra thick,rounded corners,gray] (\x-1,-1)--(\x-1,.5)--(10*\x+1,.5)--(10*\x+1,-1)--(10*\x-1,-1)--(10*\x-1,0)--(5*\x+1,0)--(5*\x+1,-1)--cycle;
\draw[ultra thick,rounded corners,gray] (\xshift+\x-1,-1)--(\xshift+\x-1,.5)--(\xshift+10*\x+1,.5)--(\xshift+10*\x+1,-1)--(\xshift+10*\x-1,-1)--(\xshift+10*\x-1,0)--(\xshift+5*\x+1,0)--(\xshift+5*\x+1,-1)--cycle;
\draw[ultra thick,gray] (10*\x+1,0)--(\xshift+\x-1,0);

\draw[ultra thick,rounded corners,gray] (-1,\y-.75) rectangle (10*\x+.5,\y+.5);
\draw[ultra thick,rounded corners,gray] (\xshift-1,\y-.75) rectangle (\xshift+10*\x+.5,\y+.5);
\draw[ultra thick,gray] (10*\x+.5,\y)--(\xshift+-1,\y);

\draw[ultra thick,rounded corners,gray] (6*\x-1,-1.5) rectangle (9*\x+1,0);
\draw[ultra thick,rounded corners,gray] (\xshift+6*\x-1,-1.5) rectangle (\xshift+9*\x+1,0);
\draw[ultra thick,gray] (7.5*\x,-1.5) to[bend right=5](\xshift+7.5*\x,-1.5);

\node[noeud,fill=lorange,label=15:$v_1$,draw=lred, ultra thick] (v1) at (4.5*\x+\xshift,2*\y) {$\beta$};
\node[noeud,fill=lgreen,label=-15:$v_2$] (v2) at (0+\xshift,1*\y) {$\beta$};
\node[noeud,fill=lred,label=-15:$v_3$] (v3) at (1.5*\x+\xshift,1*\y) {$\alpha$};
\node[noeud,fill=lred,label=-15:$v_4$] (v4) at (3.5*\x+\xshift,1*\y) {$\alpha$};
\node[noeud,fill=lblue,label=-15:$v_5$] (v5) at (6*\x+\xshift,1*\y) {$\alpha$};
\node[noeud,fill=lblue,label=-15:$v_6$] (v6) at (9*\x+\xshift,1*\y) {$\beta$};
\node[noeud,fill=lgreen,label=-90:$v_7$] (v7) at (1*\x+\xshift,0) {$\gamma$};
\node[noeud,fill=lgreen,label=-90:$v_8$] (v8) at (2*\x+\xshift,0) {$\gamma$};
\node[noeud,fill=lgreen,label=-90:$v_9$] (v9) at (3*\x+\xshift,0) {$\beta$};
\node[noeud,fill=lgreen,label=-90:$v_{10}$] (v10) at (4*\x+\xshift,0) {$\alpha$};
\node[noeud,fill=lgreen,label=-90:$v_{11}$] (v11) at (5*\x+\xshift,0) {$\gamma$};
\node[noeud,fill=lgreen,label=-90:$v_{12}$] (v12) at (6*\x+\xshift,-.5) {$\delta$};
\node[noeud,fill=lgreen,label=-90:$v_{13}$] (v13) at (7*\x+\xshift,-.5) {$\eta$};
\node[noeud,fill=lgreen,label=-90:$v_{14}$] (v14) at (8*\x+\xshift,-.5) {$\eta$};
\node[noeud,fill=lgreen,label=-90:$v_{15}$] (v15) at (9*\x+\xshift,-.5) {$\delta$};
\node[noeud,fill=lgreen,label=-90:$v_{16}$] (v16) at (10*\x+\xshift,-.5) {$\gamma$};

\node (t2) at (4.5*\x+\xshift,2.5*\y) {$T_2$};

\draw[arc] (v1)--(v2) ;
\draw[arc] (v1)--(v3) ;
\draw[arc] (v1)--(v4) ;
\draw[arc] (v1)--(v5) ;
\draw[arc] (v1)--(v6) ;
\draw[arc] (v3)--(v7) ;
\draw[arc] (v3)--(v8) ;
\draw[arc] (v4)--(v9) ;
\draw[arc] (v4)--(v10) ;
\draw[arc] (v5)--(v11) ;
\draw[arc] (v5)--(v12) ;
\draw[arc] (v5)--(v13) ;
\draw[arc] (v6)--(v14) ;
\draw[arc] (v6)--(v15) ;
\draw[arc] (v6)--(v16) ;

\draw[fleche,lred] (u1)--(v1);

\end{tikzpicture}

\caption{Running example: depth (step 2/6). Since $u_1$ and $v_1$ are the only nodes with depth zero, via Deduction~Rule~\ref{ded:bags}, we map $\isom(u_1)=v_1$ and $f(B)=\beta$. The children of $u_1$ and $v_1$ should be set aside from the other nodes according to the \textsc{SplitChildren} procedure, but they are already alone in their bag. The size of the search space is now $N(\bags)=5! \times 6! \times 4! = 2,073,600$.}
\label{preproc_depth}
\end{figure}

\textit{C. Equivalence class.} We distinguish nodes according to their equivalence class (with respect to $\simeq$): two mapped nodes must be the roots of isomorphic subtrees. For any $\bfB=(B_1,B_2)\in\bags$, we define $B_i(c) = \lbrace u\in B_i : [u]_\simeq=c\rbrace$. We add a new bag $(B_1(c),B_2(c))$ to $\bags$ for each such $c$, and remove $\bfB$ from $\bags$. Once all bags have been partitioned, we apply Deduction~Rule~\ref{ded:bags}. See Figure~\ref{preproc_equiv_class} for an example.

\begin{figure}[h!]
\centering
\def\xscale{0.35}
\def\yscale{0.7}
\def\nodescale{0.7}
\begin{tikzpicture}[xscale=\xscale,yscale=\yscale]
\tikzstyle{noeud}=[draw,circle,fill=white,scale=\nodescale*1]
\tikzstyle{attribut}=[scale=\nodescale*1,font=\bf]
\tikzstyle{arc}=[-,>=latex]
\tikzstyle{fleche}=[-,>=latex,red,ultra thick]

\def\x{2}
\def\y{1.5}
\node[noeud,fill=lorange,label=15:$u_1$,draw=lred, ultra thick] (u1) at (4.5*\x,2*\y) {$B$};
\node[noeud,fill=lgreen,label=-15:$u_2$,draw=lred, ultra thick] (u2) at (0,1*\y) {$B$};
\node[noeud,fill=lred,label=-15:$u_3$] (u3) at (1.5*\x,1*\y) {$A$};
\node[noeud,fill=lred,label=-15:$u_4$] (u4) at (3.5*\x,1*\y) {$A$};
\node[noeud,fill=lblue,label=-15:$u_5$] (u5) at (6*\x,1*\y) {$B$};
\node[noeud,fill=lblue,label=-15:$u_6$] (u6) at (9*\x,1*\y) {$A$};
\node[noeud,fill=lgreen,label=-90:$u_7$] (u7) at (1*\x,0) {$A$};
\node[noeud,fill=lgreen,label=-90:$u_8$] (u8) at (2*\x,0) {$B$};
\node[noeud,fill=lgreen,label=-90:$u_9$] (u9) at (3*\x,0) {$C$};
\node[noeud,fill=lgreen,label=-90:$u_{10}$] (u10) at (4*\x,0) {$C$};
\node[noeud,fill=lgreen,label=-90:$u_{11}$] (u11) at (5*\x,0) {$C$};
\node[noeud,fill=lgreen,label=-90:$u_{12}$] (u12) at (6*\x,-.5) {$D$};
\node[noeud,fill=lgreen,label=-90:$u_{13}$] (u13) at (7*\x,-.5) {$E$};
\node[noeud,fill=lgreen,label=-90:$u_{14}$] (u14) at (8*\x,-.5) {$D$};
\node[noeud,fill=lgreen,label=-90:$u_{15}$] (u15) at (9*\x,-.5) {$E$};
\node[noeud,fill=lgreen,label=-90:$u_{16}$] (u16) at (10*\x,0) {$C$};

\node (t1) at (4.5*\x,2.5*\y) {$T_1$};

\draw[arc] (u1)--(u2) ;
\draw[arc] (u1)--(u3) ;
\draw[arc] (u1)--(u4) ;
\draw[arc] (u1)--(u5) ;
\draw[arc] (u1)--(u6) ;
\draw[arc] (u3)--(u7) ;
\draw[arc] (u3)--(u8) ;
\draw[arc] (u4)--(u9) ;
\draw[arc] (u4)--(u10) ;
\draw[arc] (u5)--(u11) ;
\draw[arc] (u5)--(u12) ;
\draw[arc] (u5)--(u13) ;
\draw[arc] (u6)--(u14) ;
\draw[arc] (u6)--(u15) ;
\draw[arc] (u6)--(u16) ;

\def\xshift{12*\x}

\draw[ultra thick,rounded corners,gray] (1.5*\x-1,\y-.75) rectangle (4.5*\x+.5,\y+.5);
\draw[ultra thick,rounded corners,gray] (\xshift+1.5*\x-1,\y-.75) rectangle (\xshift+4.5*\x+.5,\y+.5);
\draw[ultra thick,gray] (4.5*\x+.5,\y+.25)to[bend left=15](\xshift+1.5*\x-1,\y+.25);

\draw[ultra thick,rounded corners,gray] (\x-1,-1)--(\x-1,.5)--(10*\x+1,.5)--(10*\x+1,-1)--(10*\x-1,-1)--(10*\x-1,0)--(5*\x+1,0)--(5*\x+1,-1)--cycle;
\draw[ultra thick,rounded corners,gray] (\xshift+\x-1,-1)--(\xshift+\x-1,.5)--(\xshift+10*\x+1,.5)--(\xshift+10*\x+1,-1)--(\xshift+10*\x-1,-1)--(\xshift+10*\x-1,0)--(\xshift+5*\x+1,0)--(\xshift+5*\x+1,-1)--cycle;
\draw[ultra thick,gray] (10*\x+1,0)--(\xshift+\x-1,0);

\draw[ultra thick,rounded corners,gray] (6*\x-1,-1.5) rectangle (9*\x+1,0);
\draw[ultra thick,rounded corners,gray] (\xshift+6*\x-1,-1.5) rectangle (\xshift+9*\x+1,0);
\draw[ultra thick,gray] (7.5*\x,-1.5) to[bend right=5](\xshift+7.5*\x,-1.5);

\draw[ultra thick,rounded corners,gray] (6*\x-1,\y-.75) rectangle (10*\x+.5,\y+.5);
\draw[ultra thick,rounded corners,gray] (\xshift+6*\x-1,\y-.75) rectangle (\xshift+10*\x+.5,\y+.5);
\draw[ultra thick,gray] (10*\x+.5,\y+.25)to[bend left=15](\xshift+6*\x-1,\y+.25);

\node[noeud,fill=lorange,label=15:$v_1$,draw=lred, ultra thick] (v1) at (4.5*\x+\xshift,2*\y) {$\beta$};
\node[noeud,fill=lgreen,label=-15:$v_2$,draw=lred,ultra thick] (v2) at (0+\xshift,1*\y) {$\beta$};
\node[noeud,fill=lred,label=-15:$v_3$] (v3) at (1.5*\x+\xshift,1*\y) {$\alpha$};
\node[noeud,fill=lred,label=-15:$v_4$] (v4) at (3.5*\x+\xshift,1*\y) {$\alpha$};
\node[noeud,fill=lblue,label=-15:$v_5$] (v5) at (6*\x+\xshift,1*\y) {$\alpha$};
\node[noeud,fill=lblue,label=-15:$v_6$] (v6) at (9*\x+\xshift,1*\y) {$\beta$};
\node[noeud,fill=lgreen,label=-90:$v_7$] (v7) at (1*\x+\xshift,0) {$\gamma$};
\node[noeud,fill=lgreen,label=-90:$v_8$] (v8) at (2*\x+\xshift,0) {$\gamma$};
\node[noeud,fill=lgreen,label=-90:$v_9$] (v9) at (3*\x+\xshift,0) {$\beta$};
\node[noeud,fill=lgreen,label=-90:$v_{10}$] (v10) at (4*\x+\xshift,0) {$\alpha$};
\node[noeud,fill=lgreen,label=-90:$v_{11}$] (v11) at (5*\x+\xshift,0) {$\gamma$};
\node[noeud,fill=lgreen,label=-90:$v_{12}$] (v12) at (6*\x+\xshift,-.5) {$\delta$};
\node[noeud,fill=lgreen,label=-90:$v_{13}$] (v13) at (7*\x+\xshift,-.5) {$\eta$};
\node[noeud,fill=lgreen,label=-90:$v_{14}$] (v14) at (8*\x+\xshift,-.5) {$\eta$};
\node[noeud,fill=lgreen,label=-90:$v_{15}$] (v15) at (9*\x+\xshift,-.5) {$\delta$};
\node[noeud,fill=lgreen,label=-90:$v_{16}$] (v16) at (10*\x+\xshift,0) {$\gamma$};

\node (t2) at (4.5*\x+\xshift,2.5*\y) {$T_2$};

\draw[arc] (v1)--(v2) ;
\draw[arc] (v1)--(v3) ;
\draw[arc] (v1)--(v4) ;
\draw[arc] (v1)--(v5) ;
\draw[arc] (v1)--(v6) ;
\draw[arc] (v3)--(v7) ;
\draw[arc] (v3)--(v8) ;
\draw[arc] (v4)--(v9) ;
\draw[arc] (v4)--(v10) ;
\draw[arc] (v5)--(v11) ;
\draw[arc] (v5)--(v12) ;
\draw[arc] (v5)--(v13) ;
\draw[arc] (v6)--(v14) ;
\draw[arc] (v6)--(v15) ;
\draw[arc] (v6)--(v16) ;

\draw[fleche,lred] (u1)--(v1);
\draw[fleche,lred] (u2)to[bend left=25](v2);

\end{tikzpicture}

\caption{Running example: equivalence class (step 3/6). Since $u_2$ and $v_2$ are the only nodes with class \node[lgreen] and depth 1, via Deduction~Rule~\ref{ded:bags}, we map $\isom(u_2)=v_2$ (we already have $f(B)=\beta$). Also, since $u_2$ and $v_2$ have no children and their parents are already mapped together, the process stops. The size of the search space is now $N(\bags)=(2!)^2\times 6!\times 4!=69,120$.}
\label{preproc_equiv_class}
\end{figure}
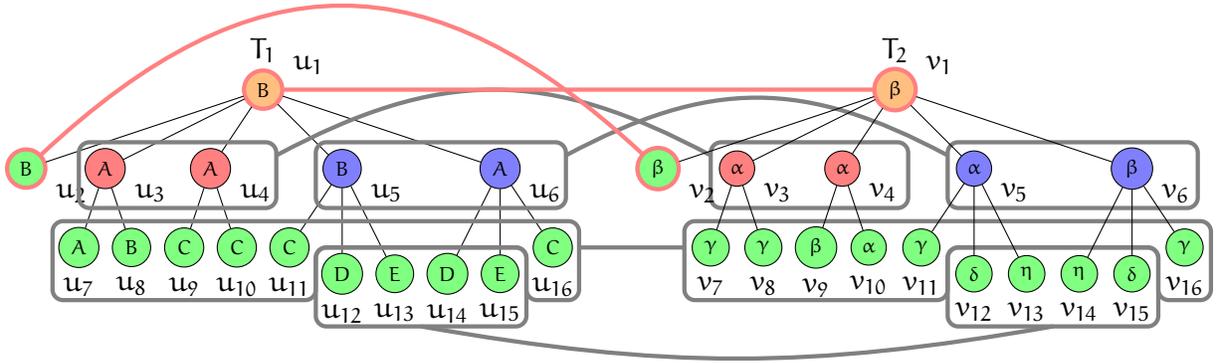

\FloatBarrier

\textit{D. Parents.} If two nodes are mapped together, so are the parents. We inspect each bag $\bfB=(B_1,B_2)$ by increasing depth and replace it by the new bags $\bfB(s)$, $s\in\bags\setminus\{\bfB\}$, such that $B_i(s) = \lbrace u\in B_i : \parent(u)\in s\rbrace$. Visiting the bags by increasing depth ensures that we have subdivided the bags containing the parents before turning to the children. Once all bags have been partitioned, we apply Deduction~Rule~\ref{ded:bags}. See Figure~\ref{preproc_parents} for an example.

\begin{figure}[h!]
\centering
\def\xscale{0.35}
\def\yscale{0.7}
\def\nodescale{0.7}
\begin{tikzpicture}[xscale=\xscale,yscale=\yscale]
\tikzstyle{noeud}=[draw,circle,fill=white,scale=\nodescale*1]
\tikzstyle{attribut}=[scale=\nodescale*1,font=\bf]
\tikzstyle{arc}=[-,>=latex]
\tikzstyle{fleche}=[-,>=latex,red,ultra thick]

\def\x{2}
\def\y{1.5}

\node[noeud,fill=lorange,label=15:$u_1$,draw=lred, ultra thick] (u1) at (4.5*\x,2*\y) {$B$};
\node[noeud,fill=lgreen,label=-15:$u_2$,draw=lred, ultra thick] (u2) at (0,1*\y) {$B$};
\node[noeud,fill=lred,label=-15:$u_3$] (u3) at (1.5*\x,1*\y) {$A$};
\node[noeud,fill=lred,label=-15:$u_4$] (u4) at (3.5*\x,1*\y) {$A$};
\node[noeud,fill=lblue,label=-15:$u_5$] (u5) at (6*\x,1*\y) {$B$};
\node[noeud,fill=lblue,label=-15:$u_6$] (u6) at (9*\x,1*\y) {$A$};
\node[noeud,fill=lgreen,label=-90:$u_7$] (u7) at (1*\x,0) {$A$};
\node[noeud,fill=lgreen,label=-90:$u_8$] (u8) at (2*\x,0) {$B$};
\node[noeud,fill=lgreen,label=-90:$u_9$] (u9) at (3*\x,0) {$C$};
\node[noeud,fill=lgreen,label=-90:$u_{10}$] (u10) at (4*\x,0) {$C$};
\node[noeud,fill=lgreen,label=-90:$u_{11}$] (u11) at (5*\x,0) {$C$};
\node[noeud,fill=lgreen,label=-90:$u_{12}$] (u12) at (6*\x,-.5) {$D$};
\node[noeud,fill=lgreen,label=-90:$u_{13}$] (u13) at (7*\x,-.5) {$E$};
\node[noeud,fill=lgreen,label=-90:$u_{14}$] (u14) at (8*\x,-.5) {$D$};
\node[noeud,fill=lgreen,label=-90:$u_{15}$] (u15) at (9*\x,-.5) {$E$};
\node[noeud,fill=lgreen,label=-90:$u_{16}$] (u16) at (10*\x,0) {$C$};

\node (t1) at (4.5*\x,2.5*\y) {$T_1$};

\draw[arc] (u1)--(u2) ;
\draw[arc] (u1)--(u3) ;
\draw[arc] (u1)--(u4) ;
\draw[arc] (u1)--(u5) ;
\draw[arc] (u1)--(u6) ;
\draw[arc] (u3)--(u7) ;
\draw[arc] (u3)--(u8) ;
\draw[arc] (u4)--(u9) ;
\draw[arc] (u4)--(u10) ;
\draw[arc] (u5)--(u11) ;
\draw[arc] (u5)--(u12) ;
\draw[arc] (u5)--(u13) ;
\draw[arc] (u6)--(u14) ;
\draw[arc] (u6)--(u15) ;
\draw[arc] (u6)--(u16) ;

\def\xshift{12*\x}

\draw[ultra thick,rounded corners,gray] (\x-1,-1) rectangle (4*\x+1,.5);
\draw[ultra thick,rounded corners,gray] (\xshift+\x-1,-1) rectangle (\xshift+4*\x+1,.5);
\draw[ultra thick,gray] (2.5*\x,-1)to[bend right=10](\xshift+2.5\x,-1);

\draw[ultra thick,rounded corners,gray] (5*\x-1,-1)--(5*\x-1,.5)--(10*\x+1,.5)--(10*\x+1,-1)--(10*\x-1,-1)--(10*\x-1,0)--(5*\x+1,0)--(5*\x+1,-1)--cycle;
\draw[ultra thick,rounded corners,gray] (\xshift+5*\x-1,-1)--(\xshift+5*\x-1,.5)--(\xshift+10*\x+1,.5)--(\xshift+10*\x+1,-1)--(\xshift+10*\x-1,-1)--(\xshift+10*\x-1,0)--(\xshift+5*\x+1,0)--(\xshift+5*\x+1,-1)--cycle;
\draw[ultra thick,gray] (10*\x,-1) to[bend right=10](\xshift+5*\x,-1);

\draw[ultra thick,rounded corners,gray] (6*\x-1,-1.5) rectangle (9*\x+1,0);
\draw[ultra thick,rounded corners,gray] (\xshift+6*\x-1,-1.5) rectangle (\xshift+9*\x+1,0);
\draw[ultra thick,gray] (7.5*\x,-1.5) to[bend right=5](\xshift+7.5*\x,-1.5);

\draw[ultra thick,rounded corners,gray] (1.5*\x-1,\y-.75) rectangle (4.5*\x+.5,\y+.5);
\draw[ultra thick,rounded corners,gray] (\xshift+1.5*\x-1,\y-.75) rectangle (\xshift+4.5*\x+.5,\y+.5);
\draw[ultra thick,gray] (4.5*\x+.5,\y+.25)to[bend left=15](\xshift+1.5*\x-1,\y+.25);

\draw[ultra thick,rounded corners,gray] (6*\x-1,\y-.75) rectangle (10*\x+.5,\y+.5);
\draw[ultra thick,rounded corners,gray] (\xshift+6*\x-1,\y-.75) rectangle (\xshift+10*\x+.5,\y+.5);
\draw[ultra thick,gray] (10*\x+.5,\y+.25)to[bend left=15](\xshift+6*\x-1,\y+.25);

\node[noeud,fill=lorange,label=15:$v_1$,draw=lred, ultra thick] (v1) at (4.5*\x+\xshift,2*\y) {$\beta$};
\node[noeud,fill=lgreen,label=-15:$v_2$,draw=lred,ultra thick] (v2) at (0+\xshift,1*\y) {$\beta$};
\node[noeud,fill=lred,label=-15:$v_3$] (v3) at (1.5*\x+\xshift,1*\y) {$\alpha$};
\node[noeud,fill=lred,label=-15:$v_4$] (v4) at (3.5*\x+\xshift,1*\y) {$\alpha$};
\node[noeud,fill=lblue,label=-15:$v_5$] (v5) at (6*\x+\xshift,1*\y) {$\alpha$};
\node[noeud,fill=lblue,label=-15:$v_6$] (v6) at (9*\x+\xshift,1*\y) {$\beta$};
\node[noeud,fill=lgreen,label=-90:$v_7$] (v7) at (1*\x+\xshift,0) {$\gamma$};
\node[noeud,fill=lgreen,label=-90:$v_8$] (v8) at (2*\x+\xshift,0) {$\gamma$};
\node[noeud,fill=lgreen,label=-90:$v_9$] (v9) at (3*\x+\xshift,0) {$\beta$};
\node[noeud,fill=lgreen,label=-90:$v_{10}$] (v10) at (4*\x+\xshift,0) {$\alpha$};
\node[noeud,fill=lgreen,label=-90:$v_{11}$] (v11) at (5*\x+\xshift,0) {$\gamma$};
\node[noeud,fill=lgreen,label=-90:$v_{12}$] (v12) at (6*\x+\xshift,-.5) {$\delta$};
\node[noeud,fill=lgreen,label=-90:$v_{13}$] (v13) at (7*\x+\xshift,-.5) {$\eta$};
\node[noeud,fill=lgreen,label=-90:$v_{14}$] (v14) at (8*\x+\xshift,-.5) {$\eta$};
\node[noeud,fill=lgreen,label=-90:$v_{15}$] (v15) at (9*\x+\xshift,-.5) {$\delta$};
\node[noeud,fill=lgreen,label=-90:$v_{16}$] (v16) at (10*\x+\xshift,0) {$\gamma$};

\node (t2) at (4.5*\x+\xshift,2.5*\y) {$T_2$};

\draw[arc] (v1)--(v2) ;
\draw[arc] (v1)--(v3) ;
\draw[arc] (v1)--(v4) ;
\draw[arc] (v1)--(v5) ;
\draw[arc] (v1)--(v6) ;
\draw[arc] (v3)--(v7) ;
\draw[arc] (v3)--(v8) ;
\draw[arc] (v4)--(v9) ;
\draw[arc] (v4)--(v10) ;
\draw[arc] (v5)--(v11) ;
\draw[arc] (v5)--(v12) ;
\draw[arc] (v5)--(v13) ;
\draw[arc] (v6)--(v14) ;
\draw[arc] (v6)--(v15) ;
\draw[arc] (v6)--(v16) ;

\draw[fleche,lred] (u1)--(v1);
\draw[fleche,lred] (u2)to[bend left=25](v2);
\end{tikzpicture}

\caption{Running example: parents (step 4/6). The only modification happens when looking at nodes of depth 2: the bag $(\{u_7,u_8,u_9,u_{10},u_{11},u_{16}\},\{v_7,v_8,v_9,v_{10},v_{11},v_{16}\})$ is split according to the bags of the parents $(\{u_3,u_4\},\{v_3,v_4\})$ and $(\{u_5,u_6\},\{v_5,v_6\})$. The size of the search space is now $N(\bags)=(2!)^3\times (4!)^2=4,608$.}
\label{preproc_parents}
\end{figure}
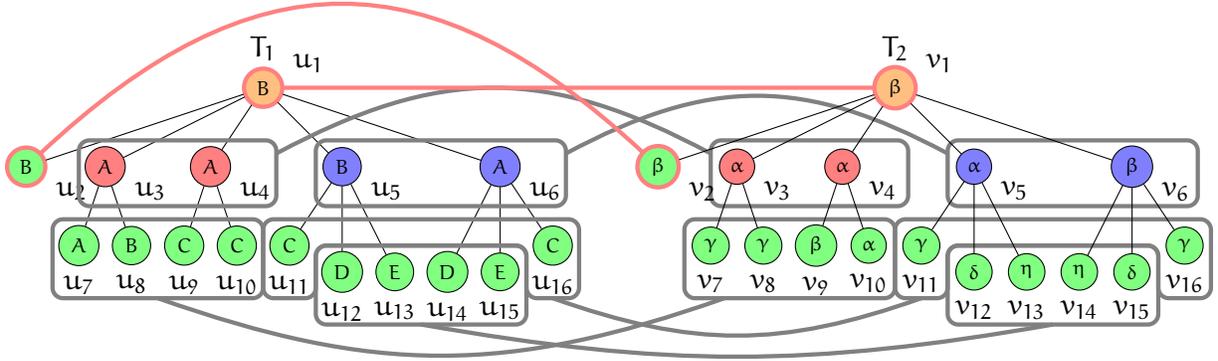

\subsection{Deductions from labels}\label{ss:ded:label}

\paragraph{From bags to collections} The initialisation and the first four steps of the algorithm given in Subsection~\ref{ss:ded:topol} only rely on topological properties. Now we use the labels to process the remaining bags: denoting $B_i(a) = \{u\in B_i\,:\,\overline{u}=a\}$, each bag $\bfB$ is replaced by
$$ \bfC = \left(\left\{B_1(a)~:~a\in\{\overline{u}\,:\,u\in B_1\}\right\}~,~\{B_2(a)~:~a\in\{\overline{u}\,:\,u\in B_2\}\right),$$
which we call a collection.

\begin{definition}\label{def:collection}
A collection $\bfC$ is a couple $(C_1,C_2)$, with $C_i \subset 2^{T_i}$, so that:
\begin{itemize}
\item for any $P\in C_i$, all the nodes of $P$ share the same label, which is denoted $\overline{P}$;
\item for any $n\in\mathbb{N}$, the sets $C_{1,n} = \{P\in C_1\,:\,\#P=n\}$ and $C_{2,n} = \{P\in C_2\,:\,\#P=n\}$ share the same cardinality. In that case, $\bfC_n = (C_{1,n},C_{2,n})$ and $\#\bfC_n = \#C_{i,n}$.
\end{itemize}
$\supp(\bfC)$ denotes the set of integers $n$ such that $\#\bfC_n>0$.
\end{definition}

Collections, which set is denoted $\collections$, are composed of sets of nodes that share the same label and are susceptible to be put together in a bag. When converting bags into collections, if the second condition is not satisfied, then we can conclude $T_1\nsim T_2$. On the other hand, the first condition is necessarily true by construction. It should be noted that, at this stage, the set $\bags$ of bags is now empty. The transition from bags to collections is illustrated on the running example in Figure~\ref{preproc_collec}.

\begin{figure}[h!]
\def\xscale{0.7}
\def\yscale{0.7}
\def\nodescale{0.7}
\centering
\begin{tikzpicture}[xscale=\xscale,yscale=\yscale]
\tikzstyle{noeud}=[draw,circle,fill=white,scale=\nodescale*1]
\tikzstyle{attribut}=[scale=\nodescale*1,font=\bf]
\tikzstyle{arc}=[->-,>=latex]

\def\x{1.5}
\def\y{2}

%%%%% Collection 1 %%%%%

\node[noeud,fill=lred,label=-15:$u_3$] (u3) at (\x,1*\y) {$A$};
\node[noeud,fill=lred,label=-15:$u_4$] (u4) at (2*\x,1*\y) {$A$};
\draw[rounded corners,ultra thick,gray] (\x-.5,\y-.75) rectangle (2*\x+1.25,\y+.5);

\node[noeud,fill=lred,label=-15:$v_3$] (v3) at (\x,0) {$\alpha$};
\node[noeud,fill=lred,label=-15:$v_4$] (v4) at (2*\x,0) {$\alpha$};
\draw[rounded corners,ultra thick,gray] (\x-.5,-.75) rectangle (2*\x+1.25,.5);

\node at (1.5*\x,\y+1) {$n=2$};
\node at (0,\y) {$C_1$};
\node at (0,0) {$C_2$};
\draw[ultra thick,lightgray,dashed] (-.5,\y/2)--(2*\x+1.5,\y/2);

\draw[ultra thick,lightgray,rounded corners] (-.5,-1) rectangle (2*\x+1.5,\y+1.5);

\node[ultra thick,circle,draw=lightgray,fill=white] at (-.5,\y+1.5) {$\bfC$};

%%%%% Collection 2 %%%%%

\def\xstep{4*\x}
\def\ystep{0*\y}

\node[noeud,fill=lgreen,label=-15:$u_7$] (u7) at (\x+\xstep,1*\y+\ystep) {$A$};
\draw[rounded corners,ultra thick,gray] (\x-.5+\xstep,\y-.75+\ystep) rectangle (1*\x+1.25+\xstep,\y+.5+\ystep);
\node[noeud,fill=lgreen,label=-15:$u_8$] (u8) at (2.5*\x+\xstep,1*\y+\ystep) {$B$};
\draw[rounded corners,ultra thick,gray] (2.5*\x-.5+\xstep,\y-.75+\ystep) rectangle (2.5*\x+1.25+\xstep,\y+.5+\ystep);

\draw[rounded corners,ultra thick,gray] (4.5*\x-.5+\xstep,\y-.75+\ystep) rectangle (5.5*\x+1.25+\xstep,\y+.5+\ystep);
\node[noeud,fill=lgreen,label=-15:$u_9$] (u9) at (4.5*\x+\xstep,1*\y+\ystep) {$C$};
\node[noeud,fill=lgreen,label=-15:$u_{10}$] (u10) at (5.5*\x+\xstep,1*\y+\ystep) {$C$};

\node[noeud,fill=lgreen,label=-15:$v_9$] (v9) at (\x+\xstep,+\ystep) {$\beta$};
\draw[rounded corners,ultra thick,gray] (\x-.5+\xstep,-.75+\ystep) rectangle (1*\x+1.25+\xstep,+.5+\ystep);
\node[noeud,fill=lgreen,label=-15:$v_{10}$] (v10) at (2.5*\x+\xstep,\ystep) {$\alpha$};
\draw[rounded corners,ultra thick,gray] (2.5*\x-.5+\xstep,-.75+\ystep) rectangle (2.5*\x+1.25+\xstep,+.5+\ystep);

\node[noeud,fill=lgreen,label=-15:$v_7$] (v7) at (4.5*\x+\xstep,+\ystep) {$\gamma$};
\node[noeud,fill=lgreen,label=-15:$v_{8}$] (v8) at (5.5*\x+\xstep,+\ystep) {$\gamma$};
\draw[rounded corners,ultra thick,gray] (4.5*\x-.5+\xstep,-.75+\ystep) rectangle (5.5*\x+1.25+\xstep,+.5+\ystep);

\node at (1.75*\x+\xstep,\y+1+\ystep) {$n=1$};
\node at (5*\x+\xstep,\y+1+\ystep) {$n=2$};
\draw[ultra thick,lightgray,dashed] (3.75*\x+\xstep,\ystep-1)--(3.75*\x+\xstep,\y+1.5+\ystep);
\node at (+\xstep,\y+\ystep) {$C'_1$};
\node at (+\xstep,0+\ystep) {$C'_2$};
\draw[ultra thick,lightgray,dashed] (-.5+\xstep,\y/2+\ystep)--(5.5*\x+1.5+\xstep,\y/2+\ystep);

\draw[ultra thick,lightgray,rounded corners] (-.5+\xstep,-1+\ystep) rectangle (5.5*\x+1.5+\xstep,\y+1.5+\ystep);

\node[ultra thick,circle,draw=lightgray,fill=white] at (-.5+\xstep,\y+1.5+\ystep) {$\mathbf{C'}$};

%%%%% f %%%%

\tikzstyle{fleche}=[-,>=latex, ultra thick]

\def\xstep{4.5*\x}

% \draw[rounded corners,thick] (7*\x-.5+\xstep,-\y/2-.5) rectangle (8.5*\x+0.5+\xstep,1.5*\y+1);
\node at (7.75*\x+\xstep,1.5*\y+.75) {$f$};
\node[noeud] (a) at (7*\x+\xstep,1.5*\y) {$A$};
\node[noeud,draw=lred,ultra thick] (b) at (7*\x+\xstep,\y) {$B$};
\node[noeud] (c) at (7*\x+\xstep,\y/2) {$C$};
\node[noeud] (d) at (7*\x+\xstep,0) {$D$};
\node[noeud] (e) at (7*\x+\xstep,-\y/2) {$E$};
\node[noeud] (alpha) at (8.5*\x+\xstep,1.5*\y) {$\alpha$};
\node[noeud,draw=lred,ultra thick] (beta) at (8.5*\x+\xstep,\y) {$\beta$};
\node[noeud] (gamma) at (8.5*\x+\xstep,\y/2) {$\gamma$};
\node[noeud] (delta) at (8.5*\x+\xstep,0) {$\delta$};
\node[noeud] (eta) at (8.5*\x+\xstep,-\y/2) {$\eta$};
\draw[fleche,lred] (b)--(beta);
%\draw[fleche] (c)--(gamma);

%%%%% Collection 3 %%%%%

\def\xstep{0*\x}
\def\ystep{-2.75*\y}

\draw[rounded corners,ultra thick,gray] (\x-.5+\xstep,\y-.75+\ystep) rectangle (2*\x+1.25+\xstep,\y+.5+\ystep);
\node[noeud,fill=lgreen,label=-15:$u_{11}$] (u11) at (\x+\xstep,1*\y+\ystep) {$C$};
\node[noeud,fill=lgreen,label=-15:$u_{16}$] (u14) at (2*\x+\xstep,1*\y+\ystep) {$C$};

\node[noeud,fill=lgreen,label=-15:$v_{11}$] (v12) at (1*\x+\xstep,0*\y+\ystep) {$\gamma$};
\node[noeud,fill=lgreen,label=-15:$v_{16}$] (v16) at (2*\x+\xstep,0*\y+\ystep) {$\gamma$};
\draw[rounded corners,ultra thick,gray] (\x-.5+\xstep,-.75+\ystep) rectangle (2*\x+1.25+\xstep,+.5+\ystep);

\node at (1.5*\x+\xstep,\y+1+\ystep) {$n=2$};
\node at (+\xstep,\y+\ystep) {$C''_1$};
\node at (+\xstep,0+\ystep) {$C''_2$};
\draw[ultra thick,lightgray,dashed] (-.5+\xstep,\y/2+\ystep)--(2*\x+1.5+\xstep,\y/2+\ystep);

\draw[ultra thick,lightgray,rounded corners] (-.5+\xstep,-1+\ystep) rectangle (2*\x+1.5+\xstep,\y+1.5+\ystep);

\node[ultra thick,circle,draw=lightgray,fill=white] at (-.5+\xstep,\y+1.5+\ystep) {$\mathbf{C''}$};

%%%%% Collection 4 %%%%%

\def\xstep{4*\x}
\def\ystep{-2.75*\y}

\draw[rounded corners,ultra thick,gray] (1*\x-.5+\xstep,\y-.75+\ystep) rectangle (2*\x+1.25+\xstep,\y+.5+\ystep);
\node[noeud,fill=lgreen,label=-15:$u_{12}$] (u12) at (1*\x+\xstep,1*\y+\ystep) {$D$};
\node[noeud,fill=lgreen,label=-15:$u_{14}$] (u15) at (2*\x+\xstep,1*\y+\ystep) {$D$};

\draw[rounded corners,ultra thick,gray] (3.5*\x-.5+\xstep,\y-.75+\ystep) rectangle (4.5*\x+1.25+\xstep,\y+.5+\ystep);
\node[noeud,fill=lgreen,label=-15:$u_{13}$] (u13) at (3.5*\x+\xstep,1*\y+\ystep) {$E$};
\node[noeud,fill=lgreen,label=-15:$u_{15}$] (u16) at (4.5*\x+\xstep,1*\y+\ystep) {$E$};

\node[noeud,fill=lgreen,label=-15:$v_{12}$] (v11) at (1*\x+\xstep,0*\y+\ystep) {$\delta$};
\node[noeud,fill=lgreen,label=-15:$v_{15}$] (v15) at (2*\x+\xstep,0*\y+\ystep) {$\delta$};
\draw[rounded corners,ultra thick,gray] (1*\x-.5+\xstep,-.75+\ystep) rectangle (2*\x+1.25+\xstep,+.5+\ystep);

\node[noeud,fill=lgreen,label=-15:$v_{13}$] (v13) at (3.5*\x+\xstep,0*\y+\ystep) {$\eta$};
\node[noeud,fill=lgreen,label=-15:$v_{14}$] (v14) at (4.5*\x+\xstep,0*\y+\ystep) {$\eta$};
\draw[rounded corners,ultra thick,gray] (3.5*\x-.5+\xstep,-.75+\ystep) rectangle (4.5*\x+1.25+\xstep,+.5+\ystep);

\node at (3*\x+\xstep,\y+1+\ystep) {$n=2$};
\node at (+\xstep,\y+\ystep) {$C'''_1$};
\node at (+\xstep,0+\ystep) {$C'''_2$};
\draw[ultra thick,lightgray,dashed] (-.5+\xstep,\y/2+\ystep)--(4.5*\x+1.5+\xstep,\y/2+\ystep);

\draw[ultra thick,lightgray,rounded corners] (-.5+\xstep,-1+\ystep) rectangle (4.5*\x+1.5+\xstep,\y+1.5+\ystep);

\node[ultra thick,circle,draw=lightgray,fill=white] at (-.5+\xstep,\y+1.5+\ystep) {$\mathbf{C'''}$};

%%%%% Collection 5 %%%%%

\def\xstep{10.5*\x}
\def\ystep{-2.75*\y}

\node[noeud,fill=lblue,label=-15:$u_5$] (u5) at (\x+\xstep,1*\y+\ystep) {$B$};
\draw[rounded corners,ultra thick,gray] (\x-.5+\xstep,\y-.75+\ystep) rectangle (1*\x+1.25+\xstep,\y+.5+\ystep);
\node[noeud,fill=lblue,label=-15:$u_6$] (u6) at (2.5*\x+\xstep,1*\y+\ystep) {$A$};
\draw[rounded corners,ultra thick,gray] (2.5*\x-.5+\xstep,\y-.75+\ystep) rectangle (2.5*\x+1.25+\xstep,\y+.5+\ystep);

\node[noeud,fill=lblue,label=-15:$v_5$] (v5) at (\x+\xstep,0+\ystep) {$\alpha$};
\draw[rounded corners,ultra thick,gray] (\x-.5+\xstep,-.75+\ystep) rectangle (1*\x+1.25+\xstep,+.5+\ystep);
\node[noeud,fill=lblue,label=-15:$v_6$] (v6) at (2.5*\x+\xstep,0+\ystep) {$\beta$};
\draw[rounded corners,ultra thick,gray] (2.5*\x-.5+\xstep,-.75+\ystep) rectangle (2.5*\x+1.25+\xstep,+.5+\ystep);

\node at (1.75*\x+\xstep,\y+1+\ystep) {$n=1$};
\node at (+\xstep,\y+\ystep) {$C''''_1$};
\node at (+\xstep,+\ystep) {$C''''_2$};
\draw[ultra thick,lightgray,dashed] (-.5+\xstep,\y/2+\ystep)--(2.5*\x+1.5+\xstep,\y/2+\ystep);

\draw[ultra thick,lightgray,rounded corners] (+\xstep-.5,-1+\ystep) rectangle (2.5*\x+1.5+\xstep,\y+1.5+\ystep);

\node[ultra thick,circle,draw=lightgray,fill=white] at (-.5+\xstep,\y+1.5+\ystep) {$\mathbf{C''''}$};
\end{tikzpicture}
\caption{Running example: from bags to collections (step 5/6). All the bags visible on Figure~\ref{preproc_parents} have been converted into collections (gray boxes). The size of the search space taking collections into account, given in \eqref{size}, is now $N(\bags,\collections) = 2!^8 = 256$.}
\label{preproc_collec}
\end{figure}
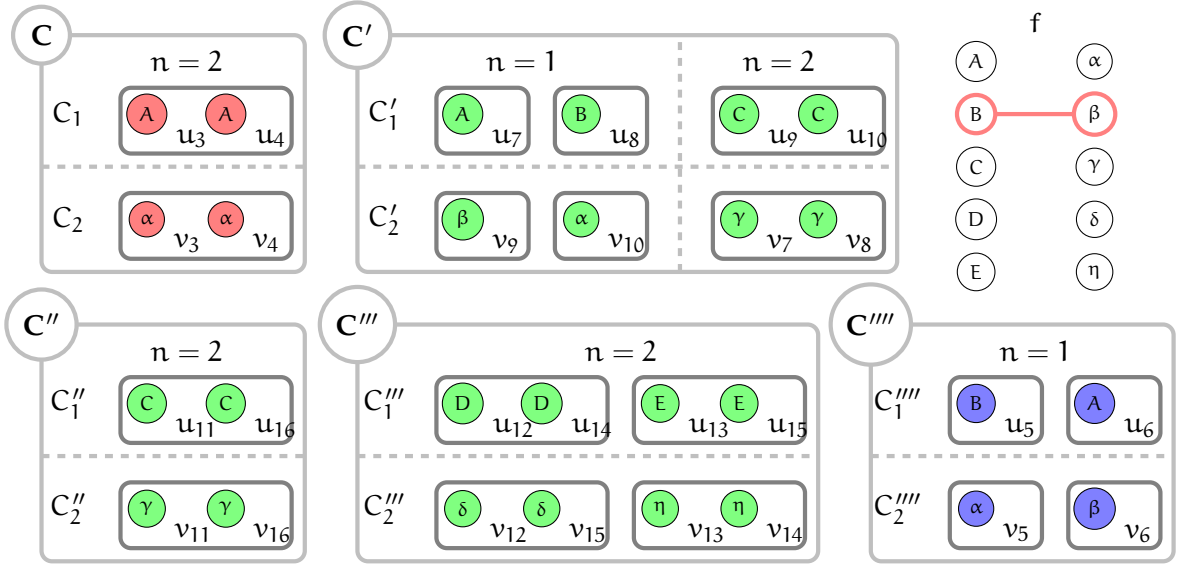

\paragraph{Deduction rules on collections}

There are three deduction rules for collections: they all deal with labels, but one of them uses mapped labels to improve the partition of nodes, one maps labels from considerations on cardinalities, while the last one both maps labels and improves the partition of nodes.

Whenever we conclude that two subsets $P$ and $Q$ of a collection should form a bag $(P,Q)$ (because their nodes should be mapped together), we know that the related labels $\overline{P}$ and $\overline{Q}$ should also be mapped together. Therefore, labels whose mapping is already known must be separated from labels that are not yet mapped. More precisely, with $C_i(a) = \lbrace P\in C_i : \overline{P}=a\rbrace$, we introduce the following deduction rule.

\begin{deduction}\label{ded:coll_label}
As long as there exist collections $\bfC$ and labels $a$ and $b$ such that $f(a)=b$, if $C_1(a)\neq \emptyset$ and $C_1\setminus C_1(a)\neq \emptyset$, replace $\bfC$ by the collections $\left(C_1(a), C_2(b)\right)$ and $\left(C_1\setminus C_1(a), C_2\setminus C_2(b)\right)$. 
\end{deduction}

The second deduction rule is complementary to the previous one. If a collection contains subsets of same cardinality that all share the same label in $T_1$ and $T_2$, then those labels must be mapped together.

\begin{deduction}\label{ded:coll_label_new}
As long as there exist collections $\bfC$, $n\in\supp(\bfC)$, and labels $a$ and $b$ not already mapped in $f$ such that, for any $P\in C_{1,n}, \overline{P}=a$ and for any $Q\in C_{2,n}, \overline{Q}=b$, call \textsc{ExtBij}$(f,a,b)$.
\end{deduction}

The last deduction rule aims at building bags: if $\#\bfC_n=1$, then the two corresponding subsets must form a bag, and their labels must be mapped together.

\begin{deduction}\label{ded:coll_cardinal}
As long as there exist collections $\bfC$ and $n\in\supp(\bfC)$ such that $C_{1,n}=\lbrace P\rbrace$ and $C_{2,n} = \lbrace Q\rbrace$, call \textsc{ExtBij}$(f,\overline{P},\overline{Q})$, add a new bag $(P,Q)$ to $\bags$ and remove $P$ and $Q$ from $\bfC$.
\end{deduction}

Starting from the organization of nodes in collections defined at the previous paragraph, Deduction Rules~\ref{ded:bags}, \ref{ded:coll_label}, \ref{ded:coll_label_new} and \ref{ded:coll_cardinal} are recursively applied: they are intertwined and each deduction made imposes to check again all the rules. The process stops when no more deductions are made, except if the conditions are fulfilled but the deduction rule can not be applied, which shows that $T_1\nsim T_2$. Deduction rules on collections are illustrated on the running example in Figure~\ref{preproc_collec_deduction}.

\begin{figure}[ht!]
\centering
\def\xscale{0.35}
\def\yscale{0.7}
\def\nodescale{0.7}
\begin{tikzpicture}[xscale=\xscale,yscale=\yscale]
\tikzstyle{noeud}=[draw,circle,fill=white,scale=\nodescale*1]
\tikzstyle{attribut}=[scale=\nodescale*1,font=\bf]
\tikzstyle{arc}=[-,>=latex]
\tikzstyle{fleche}=[-,>=latex,red,ultra thick]

\def\x{2}
\def\y{1.5}

\node[noeud,label=15:$u_1$,draw=lred,ultra thick] (u1) at (4.5*\x,2*\y) {$B$};
\node[noeud,label=-15:$u_2$,draw=lred,ultra thick] (u2) at (0,1*\y) {$B$};
\node[noeud,label=-15:$u_3$,draw=lgreen,ultra thick] (u3) at (1.5*\x,1*\y) {$A$};
\node[noeud,label=-15:$u_4$,draw=lgreen, ultra thick] (u4) at (3.5*\x,1*\y) {$A$};
\node[noeud,label=-15:$u_5$,draw=lred, ultra thick] (u5) at (6*\x,1*\y) {$B$};
\node[noeud,label=-15:$u_6$,draw=lgreen, ultra thick] (u6) at (9*\x,1*\y) {$A$};
\node[noeud,label=-90:$u_7$,draw=lgreen, ultra thick] (u7) at (1*\x,0) {$A$};
\node[noeud,label=-90:$u_8$,draw=lred,ultra thick] (u8) at (2*\x,0) {$B$};
\node[noeud,label=-90:$u_9$] (u9) at (3*\x,0) {$C$};
\node[noeud,label=-90:$u_{10}$] (u10) at (4*\x,0) {$C$};
\node[noeud,label=-90:$u_{11}$,draw=lblue,ultra thick] (u11) at (5*\x,0) {$C$};
\node[noeud,label=-90:$u_{12}$] (u12) at (6*\x,0) {$D$};
\node[noeud,label=-90:$u_{13}$] (u13) at (7*\x,0) {$E$};
\node[noeud,label=-90:$u_{14}$] (u14) at (8*\x,0) {$D$};
\node[noeud,label=-90:$u_{15}$] (u15) at (9*\x,0) {$E$};
\node[noeud,label=-90:$u_{16}$,draw=lblue, ultra thick] (u16) at (10*\x,0) {$C$};

\node (t1) at (4.5*\x,2.5*\y) {$T_1$};

\draw[arc] (u1)--(u2) ;
\draw[arc] (u1)--(u3) ;
\draw[arc] (u1)--(u4) ;
\draw[arc] (u1)--(u5) ;
\draw[arc] (u1)--(u6) ;
\draw[arc] (u3)--(u7) ;
\draw[arc] (u3)--(u8) ;
\draw[arc] (u4)--(u9) ;
\draw[arc] (u4)--(u10) ;
\draw[arc] (u5)--(u11) ;
\draw[arc] (u5)--(u12) ;
\draw[arc] (u5)--(u13) ;
\draw[arc] (u6)--(u14) ;
\draw[arc] (u6)--(u15) ;
\draw[arc] (u6)--(u16) ;

\def\xshift{12*\x}

\draw[ultra thick,rounded corners,gray] (3*\x-1,-1) rectangle (4*\x+1,.5);
\draw[ultra thick,rounded corners,gray] (\xshift+1*\x-1,-1) rectangle (\xshift+2*\x+1,.5);
\draw[ultra thick,gray] (3.5*\x,-1)to[bend right=5](\xshift+1.5\x,-1);

\node[noeud,label=15:$v_1$,draw=lred, ultra thick] (v1) at (4.5*\x+\xshift,2*\y) {$\beta$};
\node[noeud,label=-15:$v_2$,draw=lred,ultra thick] (v2) at (0+\xshift,1*\y) {$\beta$};
\node[noeud,label=-15:$v_3$,draw=lgreen, ultra thick] (v3) at (1.5*\x+\xshift,1*\y) {$\alpha$};
\node[noeud,label=-15:$v_4$,draw=lgreen,ultra thick] (v4) at (3.5*\x+\xshift,1*\y) {$\alpha$};
\node[noeud,label=-15:$v_5$,draw=lgreen, ultra thick] (v5) at (6*\x+\xshift,1*\y) {$\alpha$};
\node[noeud,label=-15:$v_6$,draw=lred, ultra thick] (v6) at (9*\x+\xshift,1*\y) {$\beta$};
\node[noeud,label=-90:$v_7$] (v7) at (1*\x+\xshift,0) {$\gamma$};
\node[noeud,label=-90:$v_8$] (v8) at (2*\x+\xshift,0) {$\gamma$};
\node[noeud,label=-90:$v_9$,draw=lred, ultra thick] (v9) at (3*\x+\xshift,0) {$\beta$};
\node[noeud,label=-90:$v_{10}$,draw=lgreen,ultra thick] (v10) at (4*\x+\xshift,0) {$\alpha$};
\node[noeud,label=-90:$v_{11}$,draw=lblue, ultra thick] (v11) at (5*\x+\xshift,0) {$\gamma$};
\node[noeud,label=-90:$v_{12}$] (v12) at (6*\x+\xshift,0) {$\delta$};
\node[noeud,label=-90:$v_{13}$] (v13) at (7*\x+\xshift,0) {$\eta$};
\node[noeud,label=-90:$v_{14}$] (v14) at (8*\x+\xshift,0) {$\eta$};
\node[noeud,label=-90:$v_{15}$] (v15) at (9*\x+\xshift,0) {$\delta$};
\node[noeud,label=-90:$v_{16}$,draw=lblue, ultra thick] (v16) at (10*\x+\xshift,0) {$\gamma$};

\node (t2) at (4.5*\x+\xshift,2.5*\y) {$T_2$};

\draw[arc] (v1)--(v2) ;
\draw[arc] (v1)--(v3) ;
\draw[arc] (v1)--(v4) ;
\draw[arc] (v1)--(v5) ;
\draw[arc] (v1)--(v6) ;
\draw[arc] (v3)--(v7) ;
\draw[arc] (v3)--(v8) ;
\draw[arc] (v4)--(v9) ;
\draw[arc] (v4)--(v10) ;
\draw[arc] (v5)--(v11) ;
\draw[arc] (v5)--(v12) ;
\draw[arc] (v5)--(v13) ;
\draw[arc] (v6)--(v14) ;
\draw[arc] (v6)--(v15) ;
\draw[arc] (v6)--(v16) ;

\draw[fleche,lred] (u1)--(v1);
\draw[fleche,lred] (u2)to[bend left=25](v2);
\draw[fleche,lred] (u8)to[bend right=15](v9);
\draw[fleche,lgreen] (u7)to[bend right=15](v10);
\draw[fleche,lgreen] (u3)to[bend left=10] (v4);
\draw[fleche,lgreen] (u4)to[bend left=10] (v3);
\draw[fleche,lred] (u5)to[bend left=5] (v6);
\draw[fleche,lgreen] (u6)to[bend left=5] (v5);
\draw[fleche,lblue] (u11)to[bend right=15](v16);
\draw[fleche,lblue] (u16)to[bend right=20](v11);

\end{tikzpicture}

\vspace{-2\baselineskip}

\def\xscale{0.7}
\def\yscale{0.7}
\def\nodescale{0.7}
\begin{tikzpicture}[xscale=\xscale,yscale=\yscale]
\tikzstyle{noeud}=[draw,circle,fill=white,scale=\nodescale*1]
\tikzstyle{attribut}=[scale=\nodescale*1,font=\bf]
\tikzstyle{arc}=[->-,>=latex]
\tikzstyle{fleche}=[-,>=latex,red,ultra thick]

\def\x{1.5}
\def\y{2}

\def\xstep{0}
\def\ystep{0}

\draw[rounded corners,ultra thick,gray] (\x-.5+\xstep,\y-.75+\ystep) rectangle (1*\x+1.25+\xstep,\y+.5+\ystep);
\node[noeud,label=-15:$u_{12}$] (u5) at (\x+\xstep,1*\y+\ystep) {$D$};

\draw[rounded corners,ultra thick,gray] (2.5*\x-.5+\xstep,\y-.75+\ystep) rectangle (2.5*\x+1.25+\xstep,\y+.5+\ystep);
\node[noeud,label=-15:$u_{13}$] (u6) at (2.5*\x+\xstep,1*\y+\ystep) {$E$};

\node[noeud,label=-15:$v_{14}$] (v5) at (\x+\xstep,0+\ystep) {$\eta$};
\draw[rounded corners,ultra thick,gray] (\x-.5+\xstep,-.75+\ystep) rectangle (1*\x+1.25+\xstep,+.5+\ystep);
\node[noeud,label=-15:$v_{15}$] (v6) at (2.5*\x+\xstep,0+\ystep) {$\delta$};
\draw[rounded corners,ultra thick,gray] (2.5*\x-.5+\xstep,-.75+\ystep) rectangle (2.5*\x+1.25+\xstep,+.5+\ystep);

\node at (1.75*\x+\xstep,\y+1+\ystep) {$n=1$};
\node at (+\xstep,\y+\ystep) {$C^\ast_1$};
\node at (+\xstep,+\ystep) {$C^\ast_2$};
\draw[ultra thick,lightgray,dashed] (-.5+\xstep,\y/2+\ystep)--(2.5*\x+1.5+\xstep,\y/2+\ystep);

\draw[ultra thick,lightgray,rounded corners] (+\xstep-.5,-1+\ystep) rectangle (2.5*\x+1.5+\xstep,\y+1.5+\ystep);

\node[ultra thick,circle,draw=lightgray,fill=white] at (-.5+\xstep,\y+1.5+\ystep) {$\mathbf{C^\ast}$};

\def\xstep{4.5*\x}
\def\ystep{0}

\draw[rounded corners,ultra thick,gray] (\x-.5+\xstep,\y-.75+\ystep) rectangle (1*\x+1.25+\xstep,\y+.5+\ystep);
\node[noeud,label=-15:$u_{14}$] (u5) at (\x+\xstep,1*\y+\ystep) {$D$};
\draw[rounded corners,ultra thick,gray] (2.5*\x-.5+\xstep,\y-.75+\ystep) rectangle (2.5*\x+1.25+\xstep,\y+.5+\ystep);
\node[noeud,label=-15:$u_{15}$] (u6) at (2.5*\x+\xstep,1*\y+\ystep) {$E$};

\node[noeud,label=-15:$v_{12}$] (v5) at (\x+\xstep,0+\ystep) {$\delta$};
\draw[rounded corners,ultra thick,gray] (\x-.5+\xstep,-.75+\ystep) rectangle (1*\x+1.25+\xstep,+.5+\ystep);
\node[noeud,label=-15:$v_{13}$] (v6) at (2.5*\x+\xstep,0+\ystep) {$\eta$};
\draw[rounded corners,ultra thick,gray] (2.5*\x-.5+\xstep,-.75+\ystep) rectangle (2.5*\x+1.25+\xstep,+.5+\ystep);

\node at (1.75*\x+\xstep,\y+1+\ystep) {$n=1$};
\node at (+\xstep,\y+\ystep) {$C^{\ast\ast}_1$};
\node at (+\xstep,+\ystep) {$C^{\ast\ast}_2$};
\draw[ultra thick,lightgray,dashed] (-.5+\xstep,\y/2+\ystep)--(2.5*\x+1.5+\xstep,\y/2+\ystep);

\draw[ultra thick,lightgray,rounded corners] (+\xstep-.5,-1+\ystep) rectangle (2.5*\x+1.5+\xstep,\y+1.5+\ystep);

\node[ultra thick,circle,draw=lightgray,fill=white] at (-.5+\xstep,\y+1.5+\ystep) {$\mathbf{C^{\ast\ast}}$};

\def\xstep{8*\x}
\def\ystep{0}

\node at (\xstep+3*\x,\ystep+\y+1.5) {$f$};

\node[noeud,draw=lgreen, ultra thick] (a) at (\xstep+\x,\ystep+\y) {$A$};
\node[noeud,draw=lred,ultra thick] (b) at (\xstep+2*\x,\ystep+\y) {$B$};
\node[noeud,draw=lblue,ultra thick] (c) at (\xstep+3*\x,\ystep+\y) {$C$};
\node[noeud] (d) at (\xstep+4*\x,\ystep+\y) {$D$};
\node[noeud] (e) at (\xstep+5*\x,\ystep+\y) {$E$};

\node[noeud,draw=lgreen, ultra thick] (a1) at (\xstep+\x,\ystep) {$\alpha$};
\node[noeud,draw=lred,ultra thick] (b1) at (\xstep+2*\x,\ystep) {$\beta$};
\node[noeud, draw=lblue, ultra thick] (c1) at (\xstep+3*\x,\ystep) {$\gamma$};
\node[noeud] (d1) at (\xstep+4*\x,\ystep) {$\delta$};
\node[noeud] (e1) at (\xstep+5*\x,\ystep) {$\eta$};

\draw[fleche,lgreen] (a)--(a1);
\draw[fleche,lred] (b)--(b1);
\draw[fleche,lblue] (c)--(c1);

\end{tikzpicture}
\caption{Running example: deductions on collections (step 6/6). Starting from the collections of Figure~\ref{preproc_collec}, Deduction Rules~\ref{ded:bags}, \ref{ded:coll_label}, \ref{ded:coll_label_new} and \ref{ded:coll_cardinal} are recursively applied. Since $f(B)=\beta$, Deduction Rule~\ref{ded:coll_label} moves nodes $(u_8,v_9)$ and $(u_5,v_6)$ into new collections; then Deduction Rule~\ref{ded:coll_label_new} allows to deduce $f(A)=\alpha$ and $f(C)=\gamma$ (for instance with collections $\bfC$ and $\mathbf{C''}$). Deduction Rule~\ref{ded:coll_cardinal} puts every node in bags except the ones in $\mathbf{C'''}$. Most of these bags are of size $1$ and therefore allow to deduce mappings via Deduction Rule~\ref{ded:bags}. The mapping $\isom(u_5)=v_6$ provokes, via \textsc{SplitChildren}, collection $\mathbf{C'''}$ to be divided into collections $\mathbf{C^\ast}$ and $\mathbf{C^{\ast\ast}}$, for which no further deduction is possible. The size of the search space after all deductions is $N(\bags,\collections)=2!^3=8$.}
\label{preproc_collec_deduction}
\end{figure}
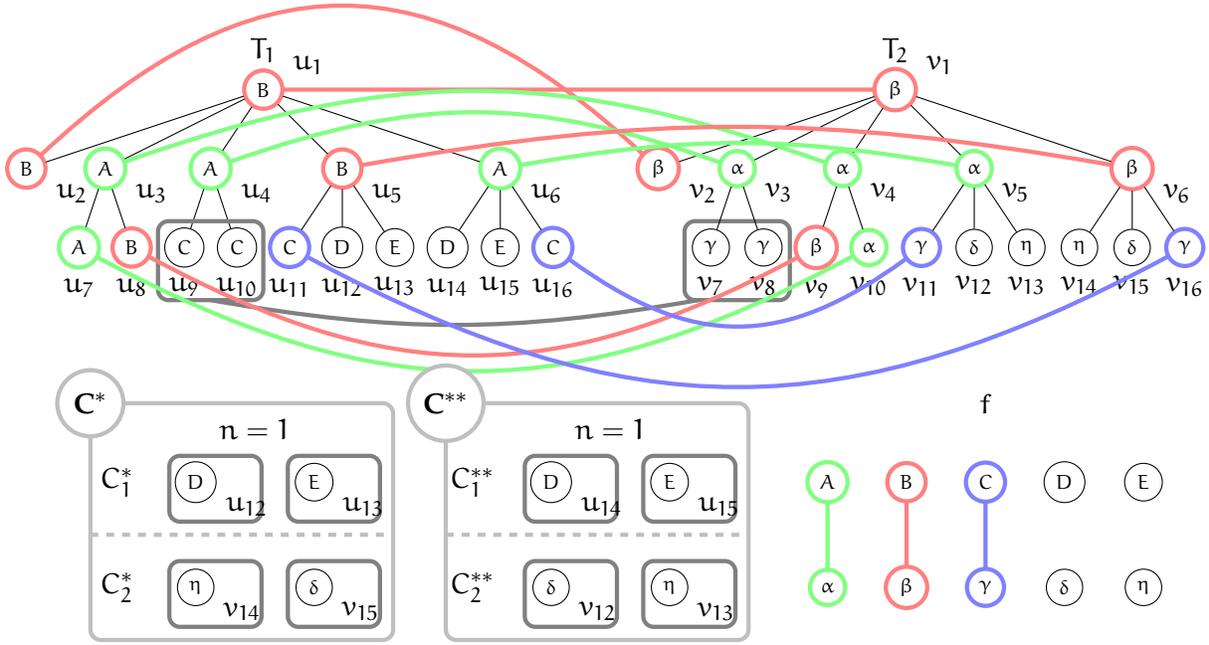

\paragraph{Redefinition of mapping nodes} Each of the deduction rules above changes the organization of the unmapped nodes of $T_1$ and $T_2$ in bags and collections. Consequently, \textsc{MapNodes} through \textsc{SplitChildren} must take into account the existence of collections. To this end, the pseudocode of Algorithm~\ref{split_children_bis} is to be added after Algorithm~\ref{algo:split_children} to redefine \textsc{SplitChildren}. An example is presented in Figure~\ref{fig:collec_exple}.

\begin{figure}[ht!]
\begin{minipage}[c]{0.51\textwidth}

\begin{algorithm}[H]
\caption{\textsc{SplitChildrenSequel}}\label{split_children_bis}
\setcounter{AlgoLine}{8 }
\SetKw{KwAnd}{and}
\For{$\bfC\in\collections$}{
\For{$n\in\supp(\bfC)$}{
Init. $\bfC'$ and $\bfC''$ as empty collections\\
Init. $S_1$ and $S'_1$ as empty sets\\
\For{$P\in C_{1,n}$}{
$P_u\gets P\cap \mathcal{C}_u$\\
\If{$P_u\neq\emptyset$ \KwAnd $P\setminus P_u\neq\emptyset$}{
Add $P_u$ to $C'_{1,\#P_u}$\\
Add $P\setminus P_u$ to $C''_{1,\#(P\setminus P_u)}$\\
$S_1\gets S_1\cup P_u$\\
$S'_1\gets S'_1\cup (P\setminus P_u)$\\
Delete $P$ from $C_{1,n}$\\
}}
Init. $S_2$ and $S'_2$ as empty sets\\
\For{$Q\in C_{2,n}$}{
$Q_v\gets Q\cap \mathcal{C}_v$\\
\If{$Q_v\neq\emptyset$ \KwAnd $Q\setminus Q_v\neq\emptyset$}{
Add $Q_v$ to $C'_{2,\#Q_v}$\\
Add $Q\setminus Q_v$ to $C''_{2,\#(Q\setminus Q_v)}$\\
$S_2\gets S_2\cup Q_v$\\
$S'_2\gets S'_2\cup (Q\setminus Q_v)$\\
Delete $Q$ from $C_{2,n}$\\
}}
Add $\bfC'$ and $\bfC''$ to $\collections$\\
\textsc{SplitChildren}$(S_1,S_2)$\\
\textsc{SplitChildren}$(S'_1,S'_2)$
}
\If{$\supp(\bfC)=\emptyset$}{Delete $\bfC$ from $\collections$}
}
\end{algorithm}
\end{minipage}\hfill
\begin{minipage}[c]{0.44\textwidth}
\def\xscale{0.65}
\def\yscale{0.7}
\def\nodescale{0.7}
\centering
\begin{tikzpicture}[xscale=\xscale,yscale=\yscale]
\tikzstyle{noeud}=[draw,circle,fill=white,scale=\nodescale*1]
\tikzstyle{attribut}=[scale=\nodescale*1,font=\bf]
\tikzstyle{arc}=[->-,>=latex]

\def\x{0.6}
\def\y{2}

%%%%% Collection 1 %%%%%

\node[noeud,fill=lred] (u1) at (\x,1*\y) {};
\node[noeud,fill=lred] (u2) at (2*\x,1*\y) {};
\node[noeud,fill=lred] (u3) at (3*\x,1*\y) {};
\node[noeud] (u4) at (4*\x,1*\y) {};
\node[noeud] (u5) at (5*\x,1*\y) {};
\draw[rounded corners,ultra thick,gray] (\x-.25,\y-.25) rectangle (5*\x+.25,\y+.25);

\node[noeud,fill=lred] (u1) at (7*\x,1*\y) {};
\node[noeud,fill=lred] (u2) at (8*\x,1*\y) {};
\node[noeud] (u3) at (9*\x,1*\y) {};
\node[noeud] (u4) at (10*\x,1*\y) {};
\node[noeud] (u5) at (11*\x,1*\y) {};
\draw[rounded corners,ultra thick,gray] (7*\x-.25,\y-.25) rectangle (11*\x+.25,\y+.25);

\node[noeud,fill=lred] (v1) at (\x,0) {};
\node[noeud,fill=lred] (v2) at (2*\x,0) {};
\node[noeud,fill=lred] (v3) at (3*\x,0) {};
\node[noeud] (v4) at (4*\x,0) {};
\node[noeud] (v5) at (5*\x,0) {};
\draw[rounded corners,ultra thick,gray] (\x-.25,-.25) rectangle (5*\x+.25,.25);

\node[noeud,fill=lred] (v1) at (7*\x,0) {};
\node[noeud,fill=lred] (v2) at (8*\x,0) {};
\node[noeud] (v3) at (9*\x,0) {};
\node[noeud] (v4) at (10*\x,0) {};
\node[noeud] (v5) at (11*\x,0) {};
\draw[rounded corners,ultra thick,gray] (7*\x-.25,-.25) rectangle (11*\x+.25,.25);

\node at (6*\x,\y+1) {$n=5$};
\node at (0,\y) {$C_1$};
\node at (0,0) {$C_2$};
\draw[ultra thick,lightgray,dashed] (-.5,\y/2)--(11*\x+.5,\y/2);

\draw[ultra thick,lightgray,rounded corners] (-.5,-1) rectangle (11*\x+.5,\y+1.5);

\node[ultra thick,circle,draw=lightgray,fill=white] at (-.5,\y+1.5) {$\bfC$};

\def\yshift{-7}

\draw[->,ultra thick,black] (6*\x,-1-.5)to (6*\x,\y+1.5+\yshift+.5);

%%%%% Collection 2 %%%%%

\node[noeud,fill=lred] (u1) at (4*\x,1*\y+\yshift) {};
\node[noeud,fill=lred] (u2) at (5*\x,1*\y+\yshift) {};
\node[noeud,fill=lred] (u3) at (6*\x,1*\y+\yshift) {};
\draw[rounded corners,ultra thick,gray] (4*\x-.25,\y-.25+\yshift) rectangle (6*\x+.25,\y+.25+\yshift);

\node[noeud,fill=lred] (v1) at (4*\x,0+\yshift) {};
\node[noeud,fill=lred] (v2) at (5*\x,0+\yshift) {};
\node[noeud,fill=lred] (v3) at (6*\x,0+\yshift) {};
\draw[rounded corners,ultra thick,gray] (4*\x-.25,-.25+\yshift) rectangle (6*\x+.25,.25+\yshift);

\node[noeud,fill=lred] (u1) at (1*\x,1*\y+\yshift) {};
\node[noeud,fill=lred] (u2) at (2*\x,1*\y+\yshift) {};
\draw[rounded corners,ultra thick,gray] (1*\x-.25,\y-.25+\yshift) rectangle (2*\x+.25,\y+.25+\yshift);

\node[noeud,fill=lred] (v1) at (1*\x,0+\yshift) {};
\node[noeud,fill=lred] (v2) at (2*\x,0+\yshift) {};
\draw[rounded corners,ultra thick,gray] (1*\x-.25,-.25+\yshift) rectangle (2*\x+.25,.25+\yshift);

\node at (1.5*\x,\y+1+\yshift) {$n=2$};
\node at (5*\x,\y+1+\yshift) {$n=3$};
\node at (0,\y+\yshift) {$C'_1$};
\node at (0,0+\yshift) {$C'_2$};

\draw[ultra thick,lightgray,dashed] (-.5,\y/2+\yshift)--(6*\x+.5,\y/2+\yshift);

\draw[ultra thick,lightgray,dashed] (3*\x,-1+\yshift)--(3*\x,\y+1.5+\yshift);

\draw[ultra thick,lightgray,rounded corners] (-.5,-1+\yshift) rectangle (6*\x+.5,\y+1.5+\yshift);

\node[ultra thick,circle,draw=lightgray,fill=white] at (-.5,\y+1.5+\yshift) {$\mathbf{C'}$};

\def\xshift{5.5}

%%%%% Collection 3 %%%%%

\node[noeud] (u1) at (4*\x+\xshift,1*\y+\yshift) {};
\node[noeud] (u2) at (5*\x+\xshift,1*\y+\yshift) {};
\node[noeud] (u3) at (6*\x+\xshift,1*\y+\yshift) {};
\draw[rounded corners,ultra thick,gray] (4*\x-.25+\xshift,\y-.25+\yshift) rectangle (6*\x+.25+\xshift,\y+.25+\yshift);

\node[noeud] (v1) at (4*\x+\xshift,0+\yshift) {};
\node[noeud] (v2) at (5*\x+\xshift,0+\yshift) {};
\node[noeud] (v3) at (6*\x+\xshift,0+\yshift) {};
\draw[rounded corners,ultra thick,gray] (4*\x-.25+\xshift,-.25+\yshift) rectangle (6*\x+.25+\xshift,.25+\yshift);

\node[noeud] (u1) at (1*\x+\xshift,1*\y+\yshift) {};
\node[noeud] (u2) at (2*\x+\xshift,1*\y+\yshift) {};
\draw[rounded corners,ultra thick,gray] (1*\x-.25+\xshift,\y-.25+\yshift) rectangle (2*\x+.25+\xshift,\y+.25+\yshift);

\node[noeud] (v1) at (1*\x+\xshift,0+\yshift) {};
\node[noeud] (v2) at (2*\x+\xshift,0+\yshift) {};
\draw[rounded corners,ultra thick,gray] (1*\x-.25+\xshift,-.25+\yshift) rectangle (2*\x+.25+\xshift,.25+\yshift);

\node at (1.5*\x+\xshift,\y+1+\yshift) {$n=2$};
\node at (5*\x+\xshift,\y+1+\yshift) {$n=3$};
\node at (0+\xshift,\y+\yshift) {$C''_1$};
\node at (0+\xshift,0+\yshift) {$C''_2$};

\draw[ultra thick,lightgray,dashed] (-.5+\xshift,\y/2+\yshift)--(6*\x+.5+\xshift,\y/2+\yshift);

\draw[ultra thick,lightgray,dashed] (3*\x+\xshift,-1+\yshift)--(3*\x+\xshift,\y+1.5+\yshift);

\draw[ultra thick,lightgray,rounded corners] (-.5+\xshift,-1+\yshift) rectangle (6*\x+.5+\xshift,\y+1.5+\yshift);

\node[ultra thick,circle,draw=lightgray,fill=white] at (-.5+\xshift,\y+1.5+\yshift) {$\mathbf{C''}$};

\end{tikzpicture}

\captionof{figure}{Illustration of \textsc{SplitChildren} on collections. The collection $\bfC$ (above) with $n=5$ is split into two collections, $\mathbf{C'}$ composed of the elements of $\mathcal{C}_u$ and $\mathcal{C}_v$ (red nodes), and $\mathbf{C''}$ with the remaining nodes (below). Recursively, the children of red nodes, first, and the children of non-red nodes, second, will in turn be separated from the other nodes -- whether they are present in bags or collections.
} \label{fig:collec_exple}
\end{minipage}
\end{figure}

Furthermore, when recursively mapping the parents of two nodes, if the latter were present in a collection, the other nodes with which these parents share a subset must be put together in a bag. More precisely, let $u$ and $v$ be the parents to be mapped, $\bfC\in\collections$ and $n\in\supp(\bfC)$ such that $u\in P \in C_{1,n}$ and $v\in Q\in C_{2,n}$. In addition to mapping $u$ and $v$, we must remove $P$ and $Q$ from $\bfC_n$ (which may lead to new deductions) and create bag $(P\setminus\lbrace u \rbrace, Q\setminus \lbrace v\rbrace)$.

\paragraph{Size of the remaining search space} When the system only consists of bags, we have established in \eqref{size_bags_only} the number of possible completions of $\isom$, referred to as the size of the search space. Now that collections have been introduced, and some bags may have been created again (via Deduction Rule~\ref{ded:coll_cardinal}), we have to evaluate the size of the search space in terms of both $\bags$ and $\collections$. 

Each bag $\bfB\in\bags$ still provides $\#\bfB!$ possible mappings. For a given collection $\bfC\in\collections$ and a given integer $n\in\supp(\bfC)$, there are $\#\bfC_n!$ ways of forming new bags; each such bag leading to $n!$ possible mappings. The number of possible completions of $\isom$ is therefore given by
\begin{equation}\label{size}
 N(\bags,\collections)=\prod_{\bfB\in \bags} \#\bfB! \prod_{\bfC\in \collections} \left(\prod_{n\in\supp(\bfC)} \#\bfC_n! \times n!^{\#\bfC_n} \right).
\end{equation}

\FloatBarrier

\subsection{Backtracking}\label{ss:backtracking}

The purpose of the steps in Subsections~\ref{ss:ded:topol} and \ref{ss:ded:label} is to reduce the size of the search space by making deductions about the necessary mappings required for a tree ciphering to exist between $T_1$ and $T_2$. When no more deductions are possible, we propose to explore the remaining possibilities by making arbitrary mappings, which may prove unsuccessful and break the partial tree ciphering, in which case one must backtrack and make other choices if any remain; otherwise one can conclude that $T_1\nsim T_2$.

At this point, the system consists of nodes already mapped and the rest distributed among bags and collections. Since arbitrary mapping choices are to be made, the issue here is to define precisely what these choices consist of, and in what order to consider them. 

\paragraph{On the backtracking tree} The underlying structure behind a backtracking approach is an enumeration tree, whose nodes are called states. In the best case, we simply traverse it from the root to a leaf; in the worst case, we explore the entirety of it. Ensuring that the tree has a minimal number of possible states is therefore of utmost importance to guarantee the performance of the algorithm. The design of the exploration of the remaining search space is thus motivated by the minimization of the number of states of the related backtracking tree.

In our case, given $\isom$, $\bags$ and $\collections$, we know that there are $N(\bags,\collections)$ ways to complete $\isom$. The definition of the backtracking exploration consists in (i) mapping nodes taken from bags together, and (ii) forming bags from sets of nodes taken from collections. No matter how these options are implemented and sorted, the resulting backtracking tree has precisely $N(\bags,\collections)$ leaves. The aim here is to minimize the number of internal (excluding the leaves and the root) states in the backtracking tree. Prior to any optimization of this number, we have the following upper-bound.

\begin{proposition}\label{backtracking_tree_size}
Starting from initial system given by $\bags$ and $\collections$, the number of states of any backtracking tree aiming to complete $\isom$ into a tree ciphering isomorphism is upper-bounded by $2(e-1)N(\bags,\collections)$.
\end{proposition}
\begin{proof}
The proof can be found in Appendix~\ref{proof:backtracking_tree}. See in particular Appendix~\ref{ss:proof_backtracking_tree_size}.
\end{proof}

First, this result highlights the value of the previous phase of deductions on topology and labels, aimed at reducing the size $N(\bags,\collections)$ of the search space as much as possible.
In addition, since we know that there must be exactly $N(\bags,\collections)$ leaves, we see that the number of states of any backtracking tree is linear in the quantity of items to enumerate. In the sequel, the efforts that are made to minimize the number of inner states aim to reduce the value of the constant in front of $N(\bags,\collections)$.

\paragraph{Processing a single bag} We assume here that the system consists of only one bag $\bfB=(B_1,B_2)$. There are $\#\bfB!$ possible mappings between the nodes of $B_1$ and those of $B_2$. In particular, for any $u\in B_1$, there are $\#\bfB$ choices for mapping $u$ to a node $v\in B_2$. Process bag $\bfB$ means mapping an arbitrary couple $(u,v)\in B_1\times B_2$ with $\textsc{ExtBij}(\isom,u,v)$, which leads (if we ignore the possible deductions that can result from it) to a new state formed of a single bag $\bfB'$ of size $\#\bfB-1$ and for which we can recursively apply the same procedure until all nodes have been mapped or an error arises. Backtracking on the choice $(u,v)$ amounts to choosing another node $v'\in B_2\setminus\{v\}$ and mapping $(u,v')$. If there is no node $v\in B_2$ that can be mapped to $u$ without generating an error, then one can conclude that $T_1\nsim T_2$.

\paragraph{Processing a single collection} Similarly, we assume that the system consists of only one collection $\bfC$ with a single integer $n$ such that $\#\bfC_n >0$. There are $\#\bfC_n!$ possible mappings between the elements of $C_{1,n}$ and those of $C_{2,n}$. For any $P\in C_{1,n}$, there are $\#\bfC_n$ choices for selecting $Q\in C_{2,n}$ to create a bag $(P,Q)$. Once this bag is created, two strategies are possible:
\begin{enumerate}[(i)]
    \item continue to form new bags from the remaining elements in $\bfC_n$, then process all these bags one by one;
    \item process this new bag entirely (according to the previous procedure) before forming a new bag from the remaining elements in $\bfC_n$.
\end{enumerate}
By virtue of upcoming Proposition~\ref{prop:strategy_for_processing_collections}, we adopt the second strategy: process collection $\bfC_n$ means putting together into a bag an arbitrary couple $(P,Q)\in C_{1,n}\times C_{2,n}$. The system is then composed of a bag and a collection. The next paragraph says how to deal with a system composed of both elements: in consistency with the choice of the second strategy, bags are processed first. As a consequence, the new bag $(P,Q)$ is entirely processed before forming new bags from the rest of $\bfC_n$.
Backtracking on this choice amounts to choosing a different set $Q'\in C_{2,n}\setminus\{Q\}$. As for bags, if there is no $Q\in C_{2,n}$ that can be associated to $P$ without generating an error, then $T_1\nsim T_2$.

\begin{proposition}\label{prop:strategy_for_processing_collections}
The backtracking tree related to strategy (ii) has fewer states than the one related to strategy (i).
\end{proposition}
\begin{proof}
The proof can be found in Appendix~\ref{proof:backtracking_tree}. See in particular Lemma~\ref{strategy_collections}.
\end{proof}

\paragraph{Processings bags and collections} We consider the general case where there are several bags and collections to choose from to process. Several questions arise: between bags and collections, which to process first? Among several bags, which one to process first? Among several collections, which one to process first? These issues are addressed in the next theorem.

\begin{theorem}\label{backtracking_tree}
The number of states of the backtracking tree is minimal if:
\begin{itemize}
    \item bags are processed before collections;
    \item when processing bags, the bag $\bfB$ with the smallest cardinality $\#\bfB$ (among all bags) is processed first;
    \item if there are only collections left, the collections $\bfC$ with the largest $n=\max(\supp(\bfC))$ (among all collections) are processed first by increasing cardinality $\#\bfC_n$.
\end{itemize}
\end{theorem}
\begin{proof}
The proof can be found in Appendix~\ref{proof:backtracking_tree}. See in particular Appendix~\ref{ss:proof_backtracking_tree}.
\end{proof}

\paragraph{Algorithm} The pseudocode for the backtracking part is provided in Algorithm~\ref{backtrack} and implements precisely the strategy of Theorem~\ref{backtracking_tree}. During the deduction steps, if an operation went wrong, we could immediately conclude that $T_1\nsim T_2$; here it indicates that an unfortunate choice has been made and that another one must be tested, if any remain. The detection of these unsuccessful cases is done in the pseudocode via the tags \textbf{try} and \textbf{catch}: $\textbf{try}~X$ $\textbf{catch}~Y$ executes $X$ and, if an error is detected, executes $Y$ instead. A number of new procedures are also introduced: \textsc{SaveState}, \textsc{RestoreState} and \textsc{NextCandidates}. The first two are self-explanatory, while the last one, presented in Algorithm~\ref{candidates}, generates the list of possible choices for the next step of backtracking, following what was established earlier.

\begin{figure}[ht!]

\begin{minipage}[t]{.55\textwidth}
\begin{algorithm}[H]
\caption{\textsc{Backtracking}}\label{backtrack}
\SetKwIF{Try}{}{Catach}{try}{}{}{catch}{endif}
\SetKw{KwAssert}{assert}

\KwIn{$\bags,\collections,\isom,f$}
$\lambda,L\gets\textsc{NextCandidates}(\bags,\collections)$\\
$S\gets\textsc{SaveState}(\bags,\collections,\isom,f)$\\
\uIf{$\lambda=\text{\emph{flag}}(\bags)$}{
\For{$(u,v)\in L$}{
        \eTry{}{
        \KwAssert \textsc{MapNodes}$(u,v,\isom,f)$\\
     Apply Deduction~Rules~\ref{ded:bags} to \ref{ded:coll_cardinal}
    
        \Return \textsc{Backtracking}$(\bags,\collections,\isom,f)$\\
        }{
        $\bags,\collections,\isom,f\gets\textsc{RestoreState}(S)$
        }}
        \Return $\bot$ 
        }
        \uElseIf{$\lambda=\text{\emph{flag}}(\collections)$}{
\For{$(P,Q)\in L$}{
        \eTry{}{
        \KwAssert \textsc{ExtBij}$(f,\overline{P},\overline{Q})$\\
        Delete $P$ and $Q$ from their collection\\
        Add $(P,Q)$ to $\bags$\\
     Apply Deduction~Rules~\ref{ded:bags} to \ref{ded:coll_cardinal}

        \Return \textsc{Backtracking}$(\bags,\collections,\isom,f)$
        }{
        $\bags,\collections,\isom,f \gets \textsc{RestoreState}(S)$
        }}
        \Return $\bot$
        }
        \Else{
        \Return $\top$ 
        }
\end{algorithm}
\end{minipage}\hfill
\begin{minipage}[t]{.44\textwidth}
\begin{algorithm}[H]
\caption{\textsc{NextCandidates}}\label{candidates}
\SetKwIF{Try}{}{Catach}{try}{}{}{catch}{endif}
\KwIn{$\bags,\collections$}
\uIf{$\bags\neq\emptyset$}{
$\lambda\gets\text{flag}(\bags)$\\
$\bfB\gets\argmin_{\bfB\in\bags} \#\bfB$\\
Select $u\in B_1$\\
$L\gets\lbrace (u,v) : v\in B_2\rbrace$
}
        \uElseIf{$\collections\neq\emptyset$}{
        $\lambda\gets\text{flag}(\collections)$\\
        $A\gets\argmax_{\bfC\in\collections} \max(\supp(\bfC))$\\
        $\bfC\gets\argmin_{\bfC\in A} \#\bfC_{\max(\supp(\bfC))}$\\
        Select $P\in C_{1,\max(\supp(\bfC))}$\\
        $L\gets\lbrace (P,Q)\,:\,Q\in C_{2,n}\rbrace$
        }
        \Else{
        $\lambda\gets\text{flag}(\emptyset)$\\
        $L\gets\emptyset$\\
        }
\Return $\lambda,L$
\end{algorithm}
\end{minipage}

\end{figure}

\paragraph{Running example: backtracking} The deduction steps end with Figure~\ref{preproc_collec_deduction}, then the backtracking exploration can start. \textsc{NextCandidates}$(\bags,\collections)$ returns $\lambda=\text{flag}(\bags)$ and $L = \lbrace (u_9, v_7), (u_9, v_8)\rbrace$. After mapping $\isom(u_9)=v_7$, one can deduce $\isom(u_{10})=v_8$ and start the recursion, this time on collections. All mapping choices between the remaining nodes and labels are acceptable: the exploration of the remaining search space does not require to backtrack.

\subsection{Analysis of time-complexity}\label{ss:timecomplexity}

\paragraph{Complexity of \textsc{MapNodes}}
The backtracking exploration has been developed together with the analysis of the number of possible states. The time-complexity of the two parts of our tree ciphering algorithm requires to investigate the procedure \textsc{MapNodes} given in Algorithm~\ref{algo:mapnodes}. Below, we assume that the two trees of interest are topologically isomorphic and note $T$ their topology.

\begin{lemma}\label{prop:number_of_calls_mapnodes}
The number of calls to \textsc{MapNodes} is:
\begin{itemize}
\item bounded by $\#T$ during the deduction steps (prior to backtracking);
\item at least $\#T$ for the whole algorithm (deduction steps and backtracking) if $T_1\sim T_2$.
\end{itemize}
\end{lemma}
\begin{proof}
The first item is stated in \cite[Proposition 1]{ingels2021isomorphic}. In addition, if $T_1\sim T_2$, the algorithm maps all the nodes, which requires at least $\#T$ calls to \textsc{MapNodes}.
\end{proof}

If $T_1\nsim T_2$, the number of calls to \textsc{MapNodes} in the backtracking exploration can not be compared to $\#T$: it may be larger if the impossibility of building a tree ciphering is detected late in the backtracking, or smaller if it is detected early during the deduction steps.

Considering the time-complexity of \textsc{ExtBij} (given in Algorithm~\ref{algo:extbij}) as $O(1)$, the complexity of a call to \textsc{MapNodes} depends only on the complexity of the subroutine \textsc{SplitChildren} (concatenation of Algorithms~\ref{algo:split_children} and~\ref{split_children_bis}), which is provided in the next proposition. It should be recalled that, contrary to the version presented in \cite[3.3 Mapping nodes]{ingels2021isomorphic}, \textsc{SplitChildren} is recursive here.

\begin{proposition}\label{prop:split_children_complexity}
The time-complexity of \textsc{SplitChildren} is $O(\#T\theight(T))$.
\end{proposition}
\begin{proof}
First, we can see that Algorithms~\ref{algo:split_children} (line~2) and \ref{split_children_bis} (line~9) loop over all bags and collections. Actually, only bags and collections containing nodes of $\mathcal{C}_u$ and $\mathcal{C}_v$ are affected. If the implementation keeps a table that maps each node to the bag or collection where it is stored, we can restrict the loop to relevant bags and collections.

Next, note that the first time \textsc{SplitChildren} is called, the nodes have already been partitioned by depth (see \textit{B. Depth} in Subsection~\ref{ss:ded:topol}). Recursive calls on their children visit strictly deeper nodes, and thus necessarily different bags, which subsequently proves that the algorithm terminates. In the worst case, each different bag and collection is visited.

Partitioning bags or collections into two elements can be done in linear time, assuming that testing membership to $\mathcal{C}_u$ or $\mathcal{C}_v$ takes constant time, e.g. via hash tables. A bag $\bfB$ is then partitioned in $O(\#\bfB)$ and a collection $\bfC_n$ in $O(n\#\bfC_n)$.

The remaining question is: how many times can a bag or a collection be further partitioned through the various recursive calls? The worst possible case is when all recursive calls act on the same bag or collection, i.e. the children of the considered nodes are in the same object (see Figure~\ref{fig:proof_split_children} for an example with bags). A bag or a collection visited at recursion depth $k$ can be successively subdivided at most $k+1$ times. Roughly bounding $k$ by $\theight(T)$, we obtain, over all recursive calls, a complexity of 
$$O\left(\Bigg(\displaystyle\sum_{\bfB\in\bags} \#\bfB +\displaystyle\sum_{\bfC\in\collections}\displaystyle\sum_{n\in\supp(\bfC)} n \#\bfC_n\Bigg)\times (\theight(T)+1)\right).$$
Since the bags and collections are composed of the nodes of $T_1$ and $T_2$, the left-hand term can not be greater than $\#T$, concluding the proof.
\end{proof}

\begin{figure}[ht!]
    \centering

\def\xscale{0.45}
\def\yscale{0.5}
\def\nodescale{0.8}

\begin{tikzpicture}[xscale=\xscale,yscale=\yscale]
\tikzstyle{noeud}=[draw,circle,fill=white,scale=\nodescale*1]
\tikzstyle{attribut}=[scale=\nodescale*1,font=\bf]
\tikzstyle{arc}=[-,>=latex]
\tikzstyle{fleche}=[->,ultra thick]

\def\xshift{0}

\node[noeud,fill=lred] (u2) at (1+\xshift,0) {};
\node[noeud,fill=lgreen] (u3) at (2+\xshift,0) {};
\node[noeud,fill=lgreen] (u4) at (3+\xshift,0) {};

\draw[rounded corners,ultra thick,gray] (1-.5+\xshift,-.5) rectangle (3.5+\xshift,.5);

\node[noeud] (v1) at (0+\xshift,-2) {};
\node[noeud] (v2) at (1+\xshift,-2) {};
\node[noeud] (v3) at (2+\xshift,-2) {};
\node[noeud] (v4) at (3+\xshift,-2) {};
\node[noeud] (v5) at (4+\xshift,-2) {};
\node[noeud] (v6) at (5+\xshift,-2) {};

\draw[arc] (u2)--(v1);
\draw[arc] (u2)--(v2);
\draw[arc] (u3)--(v3);
\draw[arc] (u4)--(v4);
\draw[arc] (u4)--(v5);

\draw[rounded corners,ultra thick,gray] (-.5+\xshift,-2-.5) rectangle (5.5+\xshift,-2+.5);

\node[noeud] (w1) at (0+\xshift,-4) {};
\node[noeud] (w2) at (1+\xshift,-4) {};
\node[noeud] (w3) at (2+\xshift,-4) {};
\node[noeud] (w4) at (3+\xshift,-4) {};
\node[noeud] (w6) at (5+\xshift,-4) {};
\node[noeud] (w7) at (6+\xshift,-4) {};
\node[noeud] (w8) at (7+\xshift,-4) {};

\draw[rounded corners,ultra thick,gray] (-.5+\xshift,-4-.5) rectangle (7.5+\xshift,-4+.5);

\draw[arc] (v2)--(w1);
\draw[arc] (v2)--(w2);
\draw[arc] (v3)--(w3);
\draw[arc] (v4)--(w4);
\draw[arc] (v6)--(w6);
\draw[arc] (v6)--(w7);

\draw[fleche,black] (6+\xshift,-2)--(8+\xshift,-2);

\def\xshift{9}

\node[noeud,fill=lred] (u2) at (1+\xshift,0) {};
\node[noeud,fill=lgreen] (u3) at (2+\xshift,0) {};
\node[noeud,fill=lgreen] (u4) at (3+\xshift,0) {};

\draw[rounded corners,ultra thick,gray] (1-.5+\xshift,-.5) rectangle (1.5+\xshift,.5);
\draw[rounded corners,ultra thick,gray] (2-.5+\xshift,-.5) rectangle (3.5+\xshift,.5);

\node[noeud,fill=lred] (v1) at (0+\xshift,-2) {};
\node[noeud,fill=lred] (v2) at (1+\xshift,-2) {};
\node[noeud,fill=lgreen] (v3) at (2+\xshift,-2) {};
\node[noeud,fill=lgreen] (v4) at (3+\xshift,-2) {};
\node[noeud,fill=lgreen] (v5) at (4+\xshift,-2) {};
\node[noeud,fill=lorange] (v6) at (5+\xshift,-2) {};

\draw[arc] (u2)--(v1);
\draw[arc] (u2)--(v2);
\draw[arc] (u3)--(v3);
\draw[arc] (u4)--(v4);
\draw[arc] (u4)--(v5);

\draw[rounded corners,ultra thick,gray] (-.5+\xshift,-2-.5) rectangle (5.5+\xshift,-2+.5);

\node[noeud] (w1) at (0+\xshift,-4) {};
\node[noeud] (w2) at (1+\xshift,-4) {};
\node[noeud] (w3) at (2+\xshift,-4) {};
\node[noeud] (w4) at (3+\xshift,-4) {};
\node[noeud] (w6) at (5+\xshift,-4) {};
\node[noeud] (w7) at (6+\xshift,-4) {};
\node[noeud] (w8) at (7+\xshift,-4) {};

\draw[rounded corners,ultra thick,gray] (-.5+\xshift,-4-.5) rectangle (7.5+\xshift,-4+.5);

\draw[arc] (v2)--(w1);
\draw[arc] (v2)--(w2);
\draw[arc] (v3)--(w3);
\draw[arc] (v4)--(w4);
\draw[arc] (v6)--(w6);
\draw[arc] (v6)--(w7);

\draw[fleche,black] (6+\xshift,-2)--(8+\xshift,-2);

\def\xshift{18}

\node[noeud,fill=lred] (u2) at (1+\xshift,0) {};
\node[noeud,fill=lgreen] (u3) at (2+\xshift,0) {};
\node[noeud,fill=lgreen] (u4) at (3+\xshift,0) {};

\draw[rounded corners,ultra thick,gray] (1-.5+\xshift,-.5) rectangle (1.5+\xshift,.5);
\draw[rounded corners,ultra thick,gray] (2-.5+\xshift,-.5) rectangle (3.5+\xshift,.5);

\node[noeud,fill=lred] (v1) at (0+\xshift,-2) {};
\node[noeud,fill=lred] (v2) at (1+\xshift,-2) {};
\node[noeud,fill=lgreen] (v3) at (2+\xshift,-2) {};
\node[noeud,fill=lgreen] (v4) at (3+\xshift,-2) {};
\node[noeud,fill=lgreen] (v5) at (4+\xshift,-2) {};
\node[noeud,fill=lorange] (v6) at (5+\xshift,-2) {};

\draw[arc] (u2)--(v1);
\draw[arc] (u2)--(v2);
\draw[arc] (u3)--(v3);
\draw[arc] (u4)--(v4);
\draw[arc] (u4)--(v5);

\draw[rounded corners,ultra thick,gray] (-.5+\xshift,-2-.5) rectangle (1.5+\xshift,-2+.5);
\draw[rounded corners,ultra thick,gray] (2-.5+\xshift,-2-.5) rectangle (4.5+\xshift,-2+.5);
\draw[rounded corners,ultra thick,gray] (5-.5+\xshift,-2-.5) rectangle (5.5+\xshift,-2+.5);

\node[noeud,fill=lred] (w1) at (0+\xshift,-4) {};
\node[noeud,fill=lred] (w2) at (1+\xshift,-4) {};
\node[noeud,fill=lgreen] (w3) at (2+\xshift,-4) {};
\node[noeud,fill=lgreen] (w4) at (3+\xshift,-4) {};
\node[noeud,fill=lorange] (w6) at (5+\xshift,-4) {};
\node[noeud,fill=lorange] (w7) at (6+\xshift,-4) {};
\node[noeud] (w8) at (7+\xshift,-4) {};

\draw[rounded corners,ultra thick,gray] (-.5+\xshift,-4-.5) rectangle (7.5+\xshift,-4+.5);

\draw[arc] (v2)--(w1);
\draw[arc] (v2)--(w2);
\draw[arc] (v3)--(w3);
\draw[arc] (v4)--(w4);
\draw[arc] (v6)--(w6);
\draw[arc] (v6)--(w7);

\draw[fleche,black] (6+\xshift,-2)--(8+\xshift,-2);

\def\xshift{27}

\node[noeud,fill=lred] (u2) at (1+\xshift,0) {};
\node[noeud,fill=lgreen] (u3) at (2+\xshift,0) {};
\node[noeud,fill=lgreen] (u4) at (3+\xshift,0) {};

\draw[rounded corners,ultra thick,gray] (1-.5+\xshift,-.5) rectangle (1.5+\xshift,.5);
\draw[rounded corners,ultra thick,gray] (2-.5+\xshift,-.5) rectangle (3.5+\xshift,.5);

\node[noeud,fill=lred] (v1) at (0+\xshift,-2) {};
\node[noeud,fill=lred] (v2) at (1+\xshift,-2) {};
\node[noeud,fill=lgreen] (v3) at (2+\xshift,-2) {};
\node[noeud,fill=lgreen] (v4) at (3+\xshift,-2) {};
\node[noeud,fill=lgreen] (v5) at (4+\xshift,-2) {};
\node[noeud,fill=lorange] (v6) at (5+\xshift,-2) {};

\draw[arc] (u2)--(v1);
\draw[arc] (u2)--(v2);
\draw[arc] (u3)--(v3);
\draw[arc] (u4)--(v4);
\draw[arc] (u4)--(v5);

\draw[rounded corners,ultra thick,gray] (-.5+\xshift,-2-.5) rectangle (1.5+\xshift,-2+.5);
\draw[rounded corners,ultra thick,gray] (2-.5+\xshift,-2-.5) rectangle (4.5+\xshift,-2+.5);
\draw[rounded corners,ultra thick,gray] (5-.5+\xshift,-2-.5) rectangle (5.5+\xshift,-2+.5);

\node[noeud,fill=lred] (w1) at (0+\xshift,-4) {};
\node[noeud,fill=lred] (w2) at (1+\xshift,-4) {};
\node[noeud,fill=lgreen] (w3) at (2+\xshift,-4) {};
\node[noeud,fill=lgreen] (w4) at (3+\xshift,-4) {};
\node[noeud,fill=lorange] (w6) at (5+\xshift,-4) {};
\node[noeud,fill=lorange] (w7) at (6+\xshift,-4) {};
\node[noeud,fill=lblue] (w8) at (7+\xshift,-4) {};

\draw[rounded corners,ultra thick,gray] (-.5+\xshift,-4-.5) rectangle (1.5+\xshift,-4+.5);
\draw[rounded corners,ultra thick,gray] (2-.5+\xshift,-4-.5) rectangle (3.5+\xshift,-4+.5);
\draw[rounded corners,ultra thick,gray] (5-.5+\xshift,-4-.5) rectangle (6.5+\xshift,-4+.5);
\draw[rounded corners,ultra thick,gray] (7-.5+\xshift,-4-.5) rectangle (7.5+\xshift,-4+.5);

\draw[arc] (v2)--(w1);
\draw[arc] (v2)--(w2);
\draw[arc] (v3)--(w3);
\draw[arc] (v4)--(w4);
\draw[arc] (v6)--(w6);
\draw[arc] (v6)--(w7);

\end{tikzpicture}
    \caption{Illustration of recursive calls of \textsc{SplitChildren} on bags (only the $B_1$ part of bags is represented) to help the proof of Proposition~\ref{prop:split_children_complexity}. The color of nodes indicates the parent-child relationship between sets $S_u$ and $\mathcal{C}_u$ through successive calls. The starting bag is divided in two. The middle bag is successively divided twice: one for the red recursive call and one for the green recursive call (the orange nodes forming the remaining nodes). Finally, the bottom bag is successively divided three times: one for the red recursive call, one for the green recursive call, and one for the orange recursive call (the blue node forming the remaining node).}
    \label{fig:proof_split_children}
\end{figure}
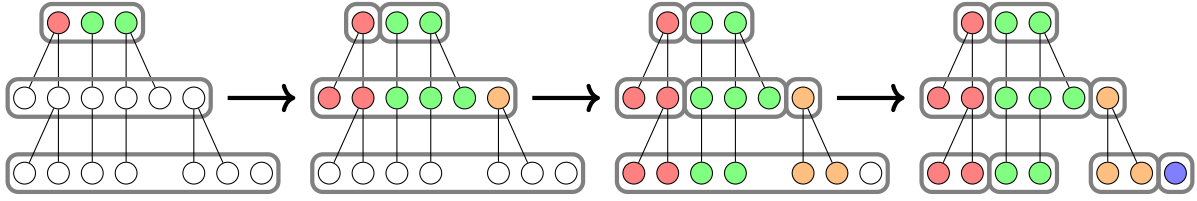

\paragraph{Experimental overall complexity} Given the intrinsic complexity of the tree ciphering isomorphism problem \cite{booth1979problems}, the worst-case complexity of the search algorithm that we propose in this paper is not relevant. However, a reasonable average time-complexity can be expected. Here we explore this question from simulated data, first when the trees are isomorphic $T_1\sim T_2$, then when they are not $T_1\nsim T_2$.

We suspect that the proportion of different labels found in the trees plays a major role in the complexity. Indeed, if all labels are different or if they are all equal, the problem becomes trivial: in those cases, any tree isomorphism is a tree ciphering. The simulation scheme therefore involves two parameters: the size of the trees $n$ and the proportion of different labels $1/n\leq p\leq 1$. 

The tree $T_1$ is simulated as follows. Its topology is a random recursive tree  with $n$ nodes \cite{zhang2015number}. $N=\max(\lfloor pn \rfloor,1)$ of them are picked at random and receive labels $1,\dots,N$; then a random label between $1$ and $N$ is assigned to each of the remaining nodes. The exact proportion of labels is $N/n\simeq p$.

The simulation of the tree $T_2$ depends on the considered scenario. In order to analyse the case $T_1\sim T_2$, $T_2$ is obtained by shuffling the children of each node of $T_1$. To test the case $T_1\nsim T_2$, $T_2$ shares the topology of $T_1$ but the labels are randomly shuffled. If $T_1$ and $T_2$ are still isomorphic, the operation is repeated. This procedure preserves the histogram of labels: the nonexistence of a tree ciphering will not be detected at step (i) of the algorithm. It should be noted that, if $p=1/n$ or $p=1$, it is impossible to build a tree $T_2$ that is not isomorphic to $T_1$ by following this method.

In each of these two cases, $500$ couples of trees $(T_1,T_2)$ were simulated for each pair $(n,p)$, with $n$ from $50$ to $500$ by step of $50$ and $p$ from $0.1$ to $0.9$ by step of $0.1$ (extreme values $p=1/n$ and $p=1$ are added only when $T_1\sim T_2$), which makes a total of $100\,000$ samples. The search algorithm was implemented in the \texttt{Python} library \texttt{treex} \cite{azais2019treex} and run on a Macbook Pro laptop with 2.6 GHz Intel Core i7 processors and 32 GB of RAM. Computation times required to test if trees $T_1$ and $T_2$ are isomorphic are presented in Figure~\ref{fig:treeciphering:comptimes}.

\begin{figure}[p!]
\thisfloatpagestyle{empty}
    \centering

     \begin{subfigure}[b]{\textwidth}
    \centering
        \caption*{\normalsize $T_1\sim T_2$}
    \label{fig:treeciphering:comptimes:isom}
    \begin{tabular}{ccc}
    \footnotesize\phantom{abcdefg}$p\leq1/3$
    & \footnotesize\phantom{abcde}$1/3<p\leq2/3$
    &\footnotesize\phantom{abcdef}$p>2/3$ \\
    \includegraphics[width=0.3\textwidth]{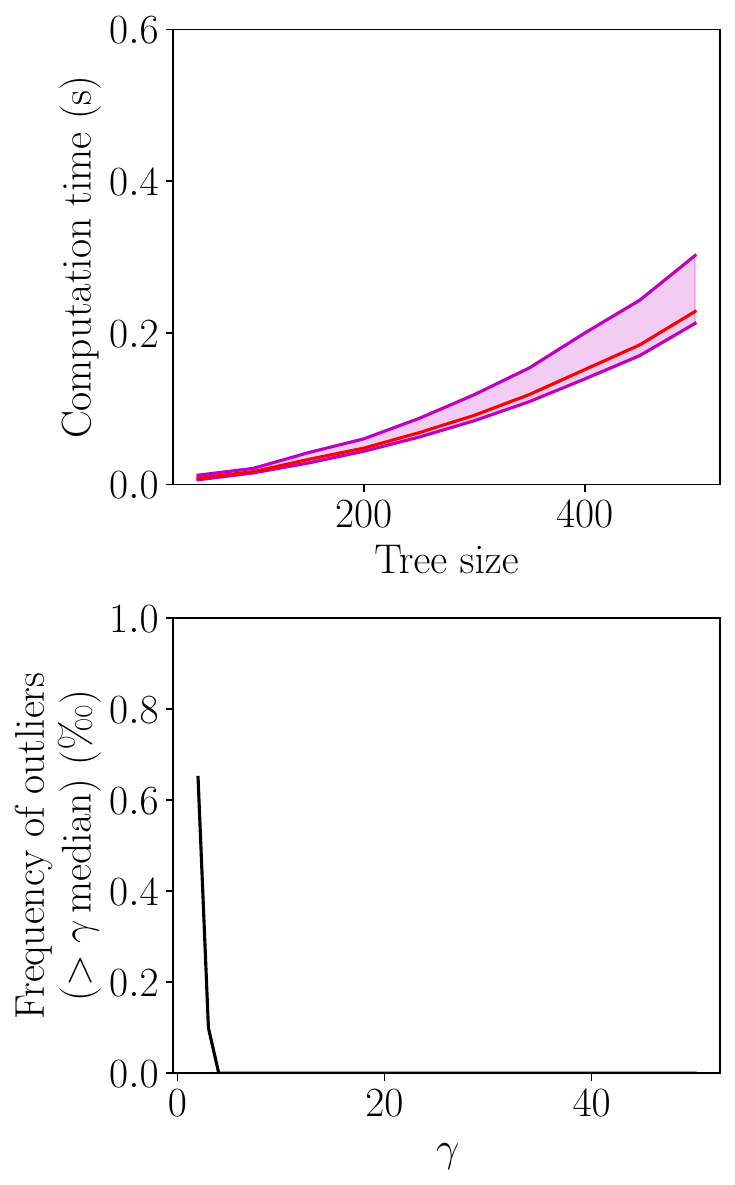}
    &  \includegraphics[width=0.3\textwidth]{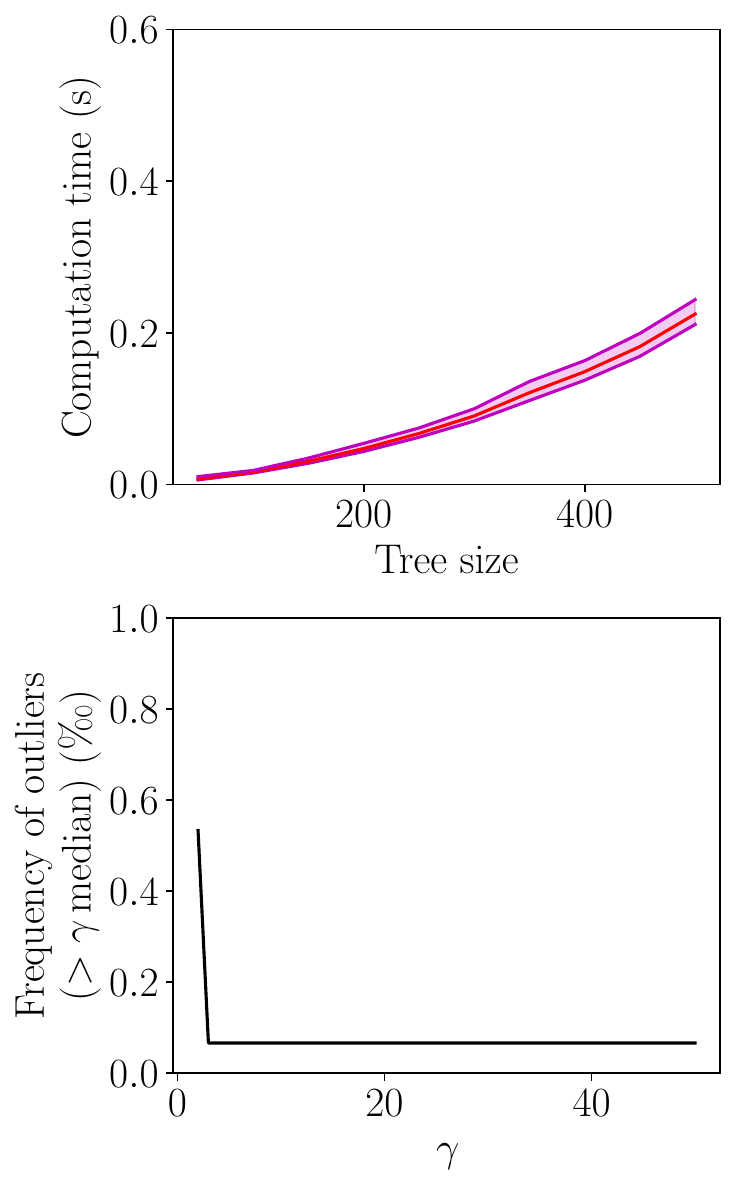}
    &  \includegraphics[width=0.3\textwidth]{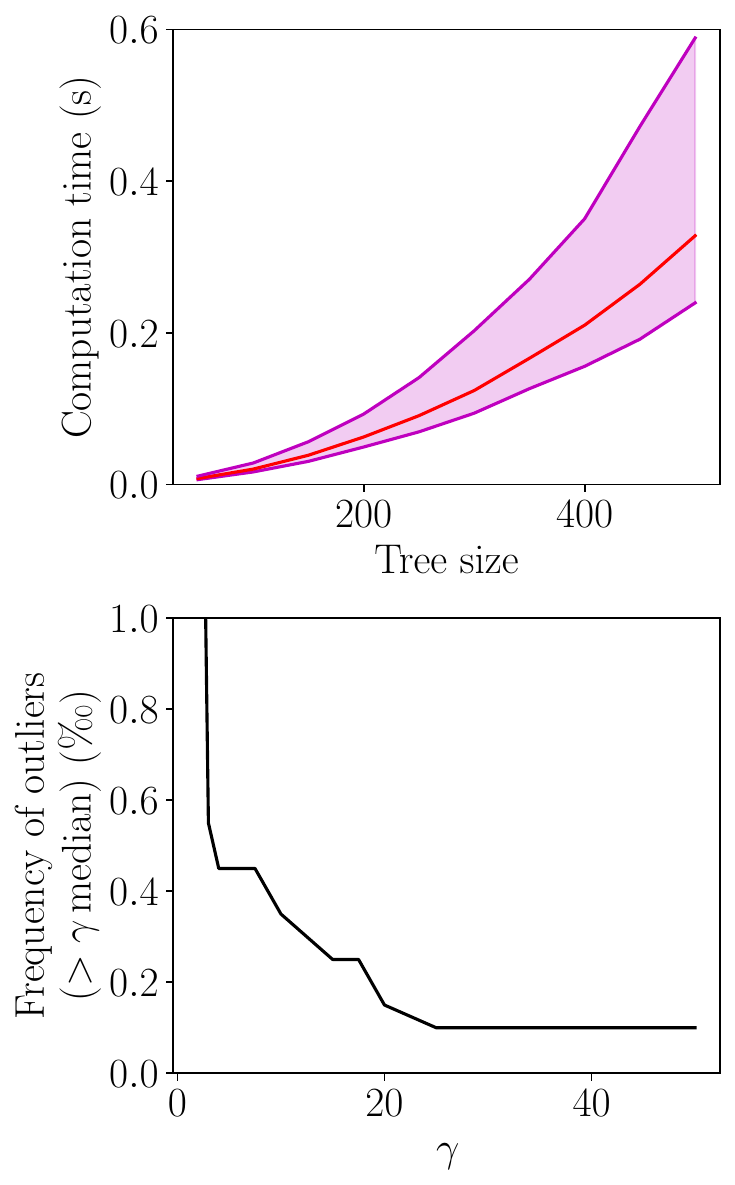}
    \end{tabular}

     \end{subfigure}

\bigskip

\begin{subfigure}[b]{\textwidth}
    \centering
        \caption*{\normalsize $T_1\nsim T_2$}
    \label{fig:treeciphering:comptimes:notisom}
    \begin{tabular}{ccc}
    \footnotesize\phantom{abcdefg}$p\leq1/3$
    & \footnotesize\phantom{abcde}$1/3<p\leq2/3$
    &\footnotesize\phantom{abcdef}$p>2/3$ \\
    \includegraphics[width=0.3\textwidth]{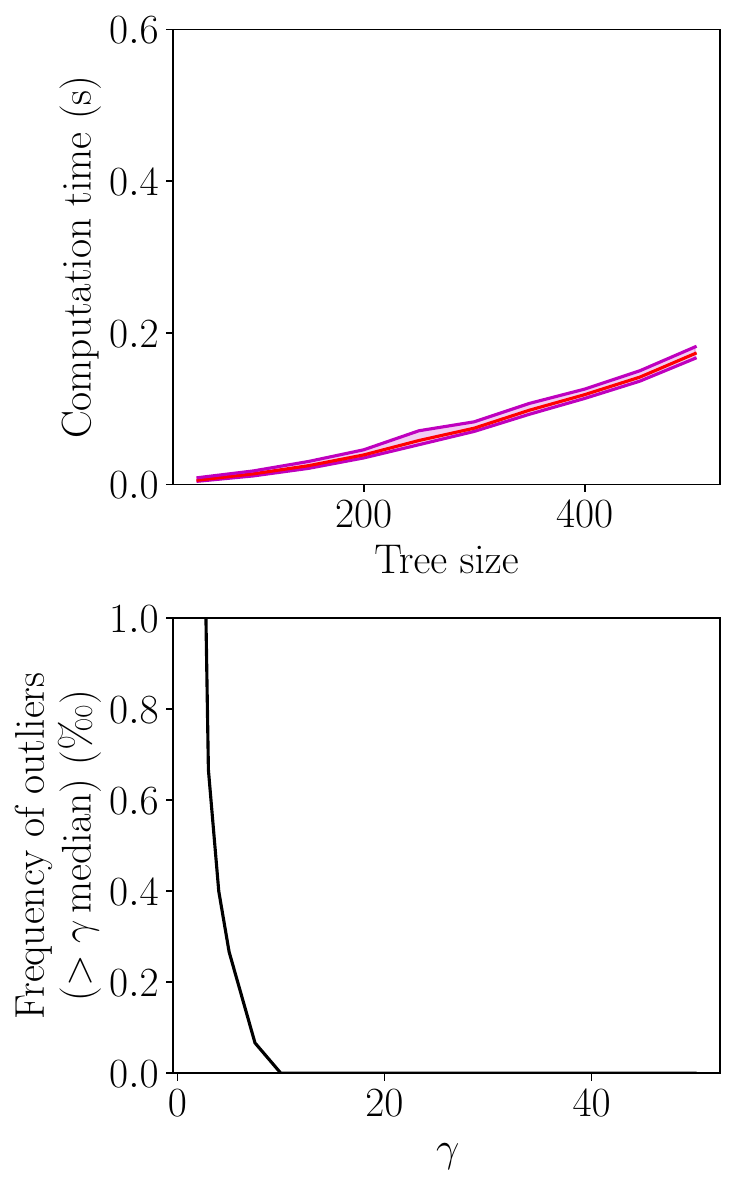}
    &  \includegraphics[width=0.3\textwidth]{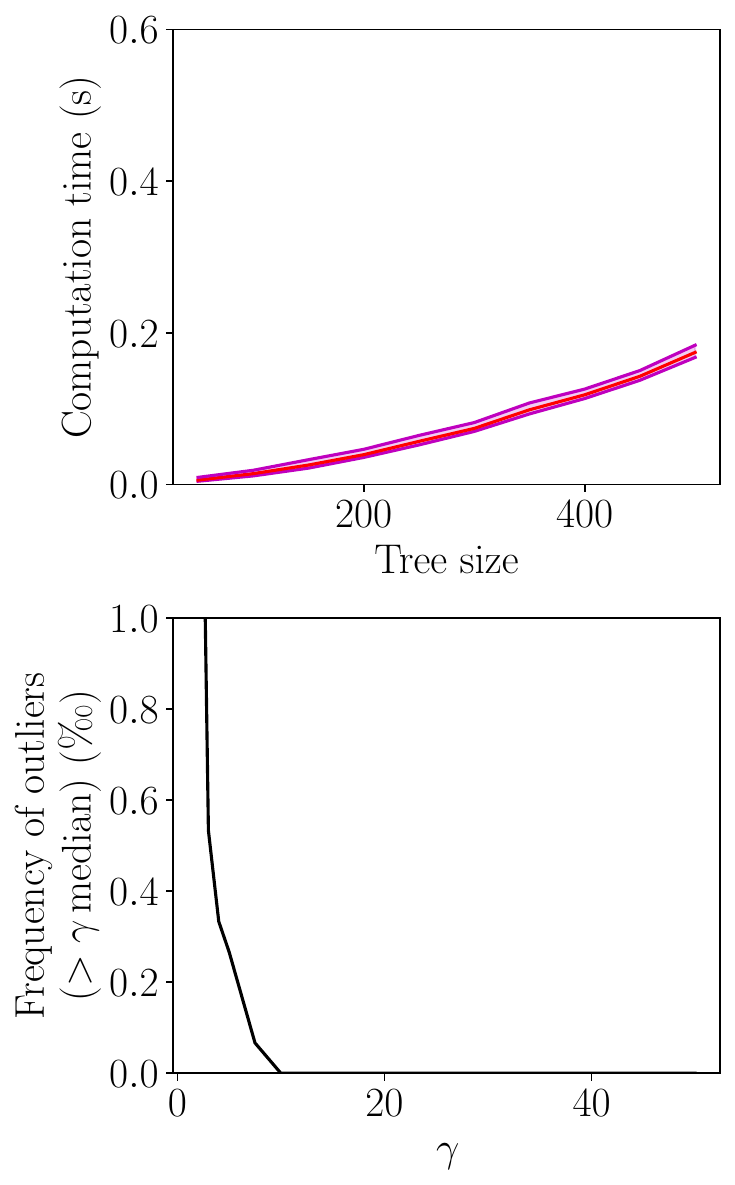}
    &  \includegraphics[width=0.3\textwidth]{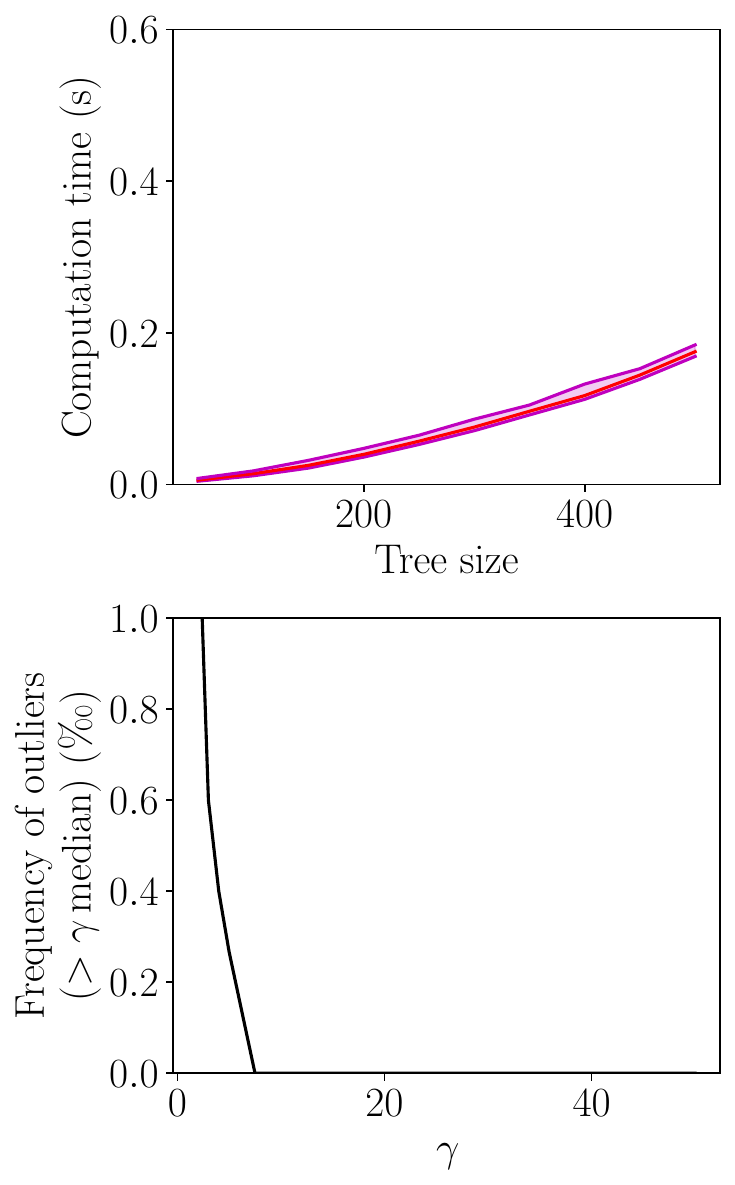}
    \end{tabular}

     \end{subfigure}
\caption{Computation time (5th quantile, median, and 95th quantile) required to decide on the tree ciphering isomorphism problem ($T_1\sim T_2$ (top) and $T_1\nsim T_2$ (bottom)) as a function of the tree size and the proportion $p$ of different labels, and proportion of outliers.}\label{fig:treeciphering:comptimes}
\end{figure}

In most cases, measured computation times are less than $0.3$ seconds for $95\%$ of the data, even for trees with $500$ nodes, which appears to be very reasonable for a problem with such a complexity. The worst-case scenario occurs when the trees are isomorphic with a proportion of different labels greater than $2/3$. the median computation time for a tree of size $500$ is around $0.35$ seconds and the 95th quantile is above $0.5$ seconds, which is still acceptable. However, the plot of the proportion of outliers shows that some examples resist, even after $50$ times the median computation time, which certainly corresponds to a deep backtracking search. The existence of such examples is of course expected but, in this proportion and with this effect, does not drastically limit the applicability of our approach, in particular to pattern mining.

\section{DAG-RW compression}
\label{s:dagrw}

In this section, we define DAG-RW compression of trees, a lossless compression scheme that eliminates their internal redundancies with respect to tree ciphering $\sim$. We also provide the related compression algorithm and discuss its complexity. Finally, we derive a method to detect frequent common subtrees with identical label distribution in a tree dataset and use it to analyse real data from the literature.

\subsection{Concept}

DAG-RW compression of a labelled unordered tree $T$ consists in a multigraph with labelled vertices and labelled edges $\red_\sim(T)=(V,E,L_V,L_E)$, where $V$ is the set of vertices, $E$ is the multiset of edges, $L_V$ is the label function of vertices, and $L_E$ is the label function of edges.

\paragraph{Topology} We first consider the topology $G$ of $\red_\sim(T)$. $G=(V,E)$ is the directed multigraph defined as follows.
\begin{itemize}
    \item The set of vertices $V$ is formed by the equivalence classes $[T(u)]_\sim$, denoted in short $[u]_\sim$, $u\in T$, with respect to the relation $\sim$.
    \item For any $(p,q)\in V^2$, $(p,q)$ appears in $E$ with multiplicity $k$ if there exists $(u,v_1,\dots,v_k)\in T^{k+1}$ such that $\{v_1,\dots,v_k\}\subset\children(u)$ and
\begin{align*}
[u]_\sim &=p,\\
\forall\,1\leq i\leq k,\quad [v_i]_\sim &= q.
\end{align*}
\end{itemize}

\begin{lemma}
The multigraph $G$ is a DAG with a unique source $S$.
\end{lemma}
\begin{proof}
    By construction, any path in $G$ makes the size of subtrees associated to visited vertices decrease strictly, which states that $G$ is a DAG. In addition, a vertex of $G$ without entering edge must be associated to a subtree whose root has no parent in $T$. $T$ is the only such subtree: this proves that $G$ has a unique source, which is associated to the whole tree $T$.
\end{proof}

\begin{lemma}\label{lem:dagrw:topology}
    $G$ is a lossless compression of the topology of $T$, i.e. $[T]_\simeq$ can be reconstructed from $G$.
\end{lemma}
\begin{proof}
One sorts the children of $\troot(T)$ by their equivalence class (with respect to tree ciphering): $n_i$ children of $\troot(T)$ are of class $p_i$ and $v_i$ denotes one of these children. By definition, $\red_\sim(T)$ has $n_i$ edges from its source $S$ to $p_i$. Furthermore, the multigraph made of $p_i$ and all its descendants in $G$ is $\red_\sim(T(v_i))$.
Consequently, by induction on the height of $T$, if the topology of $T(v_i)$ can be reconstructed from $\red_\sim(T(v_i))$, then the topology of $T$ can be obtained from $\red_\sim(T)$. Initialisation of induction at height $0$ is trivial, which states the result.
\end{proof}

\paragraph{Selection of a section} At this stage of the definition of $\red_\sim(T)$, there is no conceptual difference with DAG compression of trees based on equivalence relations $\simeq_l$ and $\simeq$, as given in Subsection~\ref{ss:intro:dagcompression}. Indeed, the set of vertices is the set of equivalence classes of subtrees of $T$ with respect to the considered equivalence relation, and the edges reflect the connection between them in $T$. However, unlike conventional DAG compression, $\red_\sim(T)$ needs to take into account labels of $T$ up to a cipher, which requires to select a specific section.

\begin{lemma}\label{lem:sectionexists}
    There exists a section $s:V\to T$ such that, for any $q\in V$,
    $$\exists\,p\in V,~(p,q)\in E~\Rightarrow~s(q)\in\children(s(p)).$$
    In the rest of the paper, $s$ denotes such a section.
\end{lemma}
\begin{proof}
The section $s$ recursively built top-down by $s(S) = \troot(T)$ and, if $s(p)=u$ then, for any $q\in V$ such that $(p,q)\in E$, $s(q)$ is selected among the set $\{v\in\children(u)\,:\,[v]_\sim=q\}$, obeys the imposed conditions. 
\end{proof}

As aforementioned, the multiplicity of edge $(p,q)$ in $E$ is related to nodes $v$ of class $q$ that share the same parent of class $p$. Now the section has been fixed, a natural candidate to be such a parent is $s(p)$. As a consequence, an edge of $G$ can be identified with the triplet $(p,q,v)$ where $v\in\children(s(p))$ such that $[v]_\sim=q$.

\paragraph{Labels} Labels of vertices of $G$ are defined through $s$: for any $q\in V$, $L_V(q) = \overline{s(q)}$. Edge $(p,q,v)$ carries the cipher $f_{\isom}$ related to the tree ciphering $\isom$ from $T(s(q))$ to $T(v)$ (which exists because $[v]_\sim=q$): $L_E(p,q,v) = f_{\isom}$. $G$ augmented with labels on its vertices and on its edges defines $\red_\sim(T)$.

\begin{lemma}\label{lem:dagrw:compo}~
    \begin{enumerate}[(i)]
        \item For any $q\neq S$, there exists an edge $(p,q,v)\in E$ such that $L_E(p,q,v)=\id$.
        \item Ciphers can be composed along paths of $\red_\sim(T)$.
    \end{enumerate}    
\end{lemma}
\begin{proof}
    (i) By definition of $s$, there exists an edge $(p,q,v)$ such that $s(q)\in\children(s(p))$. If one selects $v=s(q)$ (which is correct because $[s(q)]_\sim = q$), then the cipher carried by $(p,q,v)$ is related to the trivial isomorphism between $T(s(q))$ and itself: $L_E(p,q,v)=\id$.
    (ii) $\{(p,q,v),(q,r,w)\}$ denotes a path of length $2$ of $\red_\sim(T)$. By definition,
        \begin{align*}
        L_E(p,q,v)&:\attr(T(s(q)))\to\attr(T(v))\\
        L_E(q,r,w)&:\attr(T(s(r)))\to\attr(T(w))
        \end{align*}
        But $w\in\children(s(q))$, which implies $\attr(T(w))\subset\attr(T(s(q)))$. Then the composition $L_E(p,q,v)\circ L_E(q,r,w)$ is well-defined. By induction on the length, the result if true for any path.
\end{proof}

\begin{theorem}
$\red_\sim(T)$ defines a lossless compression of $T$.
\end{theorem}
\begin{proof}
    As stated in Lemma~\ref{lem:dagrw:topology}, the topology of $T$ can be obtained from $\red_\sim(T)$: in other words, we are already able to reconstruct $T$ without its labels. As it is transparent in the proof, each vertex $p$ of $\red_\sim(T)$ corresponds to a topological subtree of $T$ in the following sense: the multigraph made of $p$ and all its descendants in $\red_\sim(T)$ is $\red_\sim(T(v))$ for some $v\in T$. Here, one aims to state that the label of node $v$ can be read along a unique path from the source of $\red_\sim(T)$ to vertex $p$.
    
    To this end, we consider $(v_k,\dots,v_n)$ a path of $T$ (i.e., for any $k\leq i\leq n-1$, $v_{i+1}\in\children(v_i)$) such that
    \begin{align*}
        \forall\,k\leq i\leq n,\quad    &[v_i]_\sim = p_i,\\
                                   &v_k        = s(p_k), \\
        \forall\,k+1\leq i\leq n,\quad  &v_i  \neq s(p_i).
    \end{align*}
    We are interested in how to read the label of the nodes of this path from $\red_\sim(T)$. There is no difficulty with the first element: it is obvious that $\overline{v}_k = \overline{s(p_k)} = L_V(p_k)$. By definition, $v_{k+1}\in\children(s(p_k))$, thus edge $(p_k,p_{k+1},v_{k+1})$ carries the cipher $L_E(p_k,p_{k+1},v_{k+1})$ that transforms labels of $T(s(p_{k+1}))$ into labels of $T(v_{k+1})$. For instance,
\begin{align*}
    \overline{v}_{k+1} &=L_E(p_k,p_{k+1},v_{k+1})(\overline{s(p_{k+1})}) \\
    &= L_E(p_k,p_{k+1},v_{k+1})(L_V(p_{k+1}))
\end{align*}
    Now, we complete the path $(v_k,\dots,v_n)$ so that it starts from the root: $(v_1,\dots,v_n)$ with $v_1 = \troot(T)$ and, for any $1\leq i\leq n$, $[v_i]_\sim = p_i$. It should be noted that one has necessarily $v_i = s(p_i)$ for any $1\leq i\leq k$ (by definition of $s$). Consequently, edge $(p_i,p_{i+1},s(p_{i+1}))$, $1\leq i<k$, carries the cipher $L_E(p_i,p_{i+1},s(p_{i+1})) = \id$ (see Lemma~\ref{lem:dagrw:compo}). Therefore, the labels of $T(v_{k+1})$ can be obtained by applying
$$\underbrace{L_E(p_1,p_2,v_2)}_{=\,\id}~\circ~\cdots~\circ~\underbrace{L_E(p_{k-1},p_k,v_k)}_{=\,\id}~\circ~L_E(p_k,p_{k+1},v_{k+1})$$
to labels of $T(s(p_{k+1}))$.

This reasoning does not apply to any path: the first element of the path ($v_{k+2},\dots,v_n)$ is not a child of $s(p_{k+1})$ (whereas $v_{k+1}$ is a child of $s(p_k)$). Nevertheless, since $[v_{k+1}]_\sim = p_{k+1}$, there exists a path $(w_{k+2},\dots,w_n)$ in $T(s(p_{k+1}))$ such that, for any $i$, $[w_i]_\sim = [v_i]_\sim = p_i$ and $w_{k+2}\in\children(s(p_{k+1}))$. As a consequence, labels of $T(w_{k+2})$ can be obtained by applying
$$L_E(p_{k+1},p_{k+2},w_{k+2})$$
to labels of $T(s(p_{k+2}))$, which is a subtree of $T(s(p_{k+1}))$.
Finally, labels of $T(v_{k+2})$ can be obtained by applying
$$L_E(p_1,p_2,v_2)~\circ~\cdots~\circ~L_E(p_{k-1},p_k,v_k)~\circ~L_E(p_k,p_{k+1},v_{k+1})~\circ~ L_E(p_{k+1},p_{k+2},w_{k+2})$$
to labels of $T(s(p_{k+2}))$, where the composition makes sense in light of Lemma~\ref{lem:dagrw:compo}. By recursive induction, this shows that there is a one-to-one correspondence between nodes of $T$ and paths of $\red_\sim(T)$ starting from its source. In addition, the label of any node of $T$ can be read along the related path of $G$.
\end{proof}

An example of DAG-RW compression is provided with Figure~\ref{fig:compression_label_exple}.

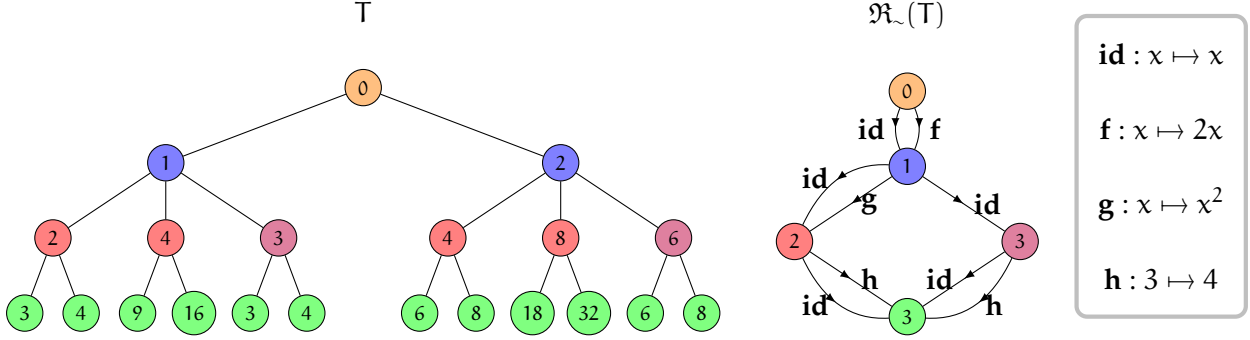
\begin{figure}[h]
\def\xscale{0.75}
\def\yscale{1}
\def\nodescale{0.7}
\centering
\begin{tikzpicture}[xscale=\xscale,yscale=\yscale]
\tikzstyle{noeud}=[draw,circle,fill=white,scale=\nodescale*1]
\tikzstyle{arc}=[->-,>=latex]
\tikzstyle{fleche}=[->,>=latex,red,thick]

\node[noeud,fill=lorange] (u0) at (0,3) {$0$};

\node[noeud,fill=lblue] (u1) at (-3.5,2) {$1$};
\node[noeud,fill=lblue] (u2) at (3.5,2) {$2$};

\node[noeud,fill=lred] (u3) at (-5.5,1) {$2$};
\node[noeud,fill=lred] (u4) at (-3.5,1) {$4$};
\node[noeud,fill=lpurple] (u5) at (-1.5,1) {$3$};
\node[noeud,fill=lred] (u6) at (1.5,1) {$4$};
\node[noeud,fill=lred] (u7) at (3.5,1) {$8$};
\node[noeud,fill=lpurple] (u8) at (5.5,1) {$6$};

\node[noeud,fill=lgreen] (u9) at (-6,0) {$3$};
\node[noeud,fill=lgreen] (u10) at (-5,0) {$4$};
\node[noeud,fill=lgreen] (u11) at (-4,0) {$9$};
\node[noeud,fill=lgreen] (u12) at (-3,0) {$16$};
\node[noeud,fill=lgreen] (u13) at (-2,0) {$3$};
\node[noeud,fill=lgreen] (u14) at (-1,0) {$4$};
\node[noeud,fill=lgreen] (u15) at (1,0) {$6$};
\node[noeud,fill=lgreen] (u16) at (2,0) {$8$};
\node[noeud,fill=lgreen] (u17) at (3,0) {$18$};
\node[noeud,fill=lgreen] (u18) at (4,0) {$32$};
\node[noeud,fill=lgreen] (u19) at (5,0) {$6$};
\node[noeud,fill=lgreen] (u20) at (6,0) {$8$};

\node (t1) at (0,4) {$T$};

\draw (u0)--(u1);
\draw (u0)--(u2);
\draw (u1)--(u3);
\draw (u1)--(u4);
\draw (u1)--(u5);
\draw (u2)--(u6);
\draw (u2)--(u7);
\draw (u2)--(u8);
\draw (u3)--(u9);
\draw (u3)--(u10);
\draw (u4)--(u11);
\draw (u4)--(u12);
\draw (u5)--(u13);
\draw (u5)--(u14);
\draw (u6)--(u15);
\draw (u6)--(u16);
\draw (u7)--(u17);
\draw (u7)--(u18);
\draw (u8)--(u19);
\draw (u8)--(u20);

\end{tikzpicture}\hfill
\begin{tikzpicture}[xscale=\xscale,yscale=\yscale]
\tikzstyle{noeud}=[draw,circle,fill=white,scale=\nodescale*1]
\tikzstyle{arc}=[->-,>=latex]
\tikzstyle{fleche}=[->,>=latex,red,thick]

\node[noeud,fill=lorange] (u0) at (0,3) {$0$};
\node[noeud,fill=lblue] (u1) at (0,2) {$1$};
\node[noeud,fill=lred] (u2) at (-2,1) {$2$};
\node[noeud,fill=lpurple] (u2b) at (2,1) {$3$};
\node[noeud,fill=lgreen] (u3) at (0,0) {$3$};

\node (t1) at (0,4) {$\red_\sim(T)$};

\draw[arc] (u0) to[bend right=30] node[left,midway]{$\id$} (u1);
\draw[arc] (u0) to[bend right=-30] node[right,midway] {$\textbf{f}$} (u1);
\draw[arc] (u1) to[bend right=30] node[left,midway]{$\id$} (u2);
\draw[arc] (u1) to node[right,midway] {$\textbf{g}$} (u2);
\draw[arc] (u1) to node[right,midway]{$\id$} (u2b);

\draw[arc] (u2) to[bend right=30] node[left,midway]{$\id$} (u3);
\draw[arc] (u2) to node[right,midway] {$\textbf{h}$} (u3);
\draw[arc] (u2b) to node[left,midway]{$\id$} (u3);
\draw[arc] (u2b) to[bend right=-30] node[right,midway] {$\textbf{h}$} (u3);

\def\x{3}

\node (u1) at (\x+1.5,3.5) {$\id : x \mapsto x$};
\node (u1) at (\x+1.5,2.5) {$\textbf{f} : x \mapsto 2x$};
\node (u2) at (\x+1.5,1.5) {$\textbf{g} : x \mapsto x^2$};;
\node (u3) at (\x+1.5,0.5) {$\textbf{h} : 3 \mapsto 4$};

\draw[ultra thick, draw=lightgray, rounded corners] (\x+0,0) rectangle (\x+3,4);

\end{tikzpicture}
\caption{A labelled tree $T$ (left) and its DAG-RW compression $\red_\sim(T)$ (right). Nodes are colored according to their equivalence class with respect to $\sim$. For the sake of readability, the ciphers $\id$, $\textbf{f}$ and $\textbf{g}$ are given as formulae but should be defined only on the alphabet of the subtree to which they point. In particular, the notation $\id$ is used several times but the corresponding ciphers are different.}
\label{fig:compression_label_exple}
\end{figure}

\begin{remark}
As explained in Subsection~\ref{ss:intro:dagcompression}, DAG compression $\red_\ast(T)$, $\ast\in\{\simeq,\simeq_l\}$, of a tree $T$, allows its reconstruction up to related isomorphism $\ast$. It is noteworthy that the situation is different for DAG-RW compression, which is lossless with respect to $\simeq_l$ while built from tree ciphering $\sim$: the addition of labels on vertices and edges makes possible the exact reconstruction of node labels. 
\end{remark}

\subsection{Compression algorithm}

We present here how to build the compressed version $\red_\sim(T)$ of $T$ from the topological compression $\red_\simeq(T)$ which is computable in linear time \cite{valiente2001}. Since the topology is already evaluated in $\red_\simeq(T)$, only the labels of $T$ need to be taken into account to construct $\red_\sim(T)$.

The main idea of the algorithm is, in a top-down approach, to replace each vertex $q$ of $\red_\simeq(T)$ by new vertices corresponding to the equivalence classes with respect to $\sim$ of subtrees of $T$ represented by $q$. By construction, each of the new vertices $r$ has the same parents $p$ as $q$. However, the value of the section $s(r)$ and the labels $L_V(r)$ and $L_E(p,r,v)$, $[v]_\sim=r$, need to be defined.

Since we aim to edit $\red_\simeq(T)$ from top to bottom, we assume that all the ascendants of $q$ correspond to vertices of $\red_\sim(T)$ and that $s$ is defined on them. To obtain the new vertices $r$ that will replace $q$, we first consider the set $A_q$ of nodes $u$ of $T$ such that $u$ is a child of $s(p)$ of class $[u]_\simeq = q$ for any parent $p$ of $q$. $A_q$ is then partitioned according to the equivalence classes with respect to $\sim$ using the algorithm for tree ciphering isomorphism given in Section~\ref{s:treeciphering}: $\mathcal{P}$ denotes this partition and the representative $\sigma(P)$ of the set $P\in\mathcal{P}$ is arbitrarily chosen to define section $s$ at this level.

Finally, the vertex $q$ is replaced by new vertices $r_P$, $P\in\mathcal{P}$, with label of $\sigma(P)$ and section $s(r_P)=\sigma(P)$, which obeys the condition of Lemma~\ref{lem:sectionexists}. For each node $v\in P$, we know by construction that there exists a tree ciphering $\isom$ between $T(v)$ and $T(s(r_p))$. By the recurrence hypothesis, the vertex corresponding to $[\parent(v)]_\sim$ exists: an edge bearing the cipher $f_{\isom}$ related to $\isom$ from it to $r_p$ can be created.

Algorithm~\ref{algo:dag_rw} provides the detail of the compression method, while Figure~\ref{fig:compression_label_iter} illustrates it on the tree of Figures~\ref{fig:dag_compression_example} and \ref{fig:compression_label_exple}.

\begin{algorithm}[ht!]
\SetKw{KwFrom}{from}
\SetKw{KwST}{so that}
\caption{\textsc{DAG-RW}}\label{algo:dag_rw}
\KwIn{$T$, $\red_\simeq(T)$}
\KwOut{$\red_\sim(T)$}
Init. $R$ as $\red_\simeq(T)$ and $s:\emptyset\to\emptyset$\\
\ForEach{\emph{vertex $q$ of $R$ from its source in breadth-first search order}}{
\eIf{$q$ \emph{is the source of} $R$}{
% \For{$h$ \KwFrom $\theight(D)$ \KwTo $0$}{
% \For{$q \in D$ \KwST $\theight(q)=h$  }{
% \eIf{$h=\theight(D)$}{
$A_q\gets\lbrace\troot(T) \rbrace$ %\tcp{Only the source of $D$ exists at this height}
}{
$A_q\gets\displaystyle\bigcup_{\text{edge}\,(p,q)\,\text{of}\,R} \lbrace u \in \children(s(p)) : [u]_\simeq =q \rbrace $
}
Init. $\mathcal{P}$ as an empty list and $\sigma: \emptyset \to \emptyset$\\
\For{$u\in A_q$}{
\eIf{$\exists\,P\in\dom(\sigma)$ \emph{such that} $T(u)\sim T(\sigma(P))$}{
Add $u$ to $P$\\
Store the related cipher $f_{\sigma(P),u}$
}{
$P\gets\{u\}$\\
Add $P$ to $\mathcal{P}$\\
$\textsc{ExtBij}(\sigma,P,u)$
}}
\For{$P\in\mathcal{P}$}{
Add a new vertex $r_P$ with label $\overline{\sigma(P)}$ to $R$\\
Identify $r_P$ as $[\sigma(P)]_\sim$\\
\ForEach{\emph{edge $(q,c)$ of $R$}}{
Add a new edge from $r_P$ to $c$
}
\For{$u\in P$}{
Add a new edge from $[\parent(u)]_\sim$ to $r_P$ with label $f_{\sigma(P),u}$\\
}
$\textsc{ExtBij}(s,r_P,\sigma(P))$
}
Delete vertex $q$}
%}
\Return{$R$}
\end{algorithm}

\begin{remark}
The size of the piece of information stored on the edges (line 16 of Algorithm~\ref{algo:dag_rw}) is linear in the number of labels of the related subtree. In some contexts, the ciphers can be encoded as formulae instead of by listing values. For instance, one of the ciphers is the identity from the alphabet of the subtree to itself. It can also be encoded as an abstract identity function that returns its argument without specifying the domain. Such an approach can reduce the amount of information stored on the edges. An example is given in Figure~\ref{fig:compression_label_exple}.
\end{remark}

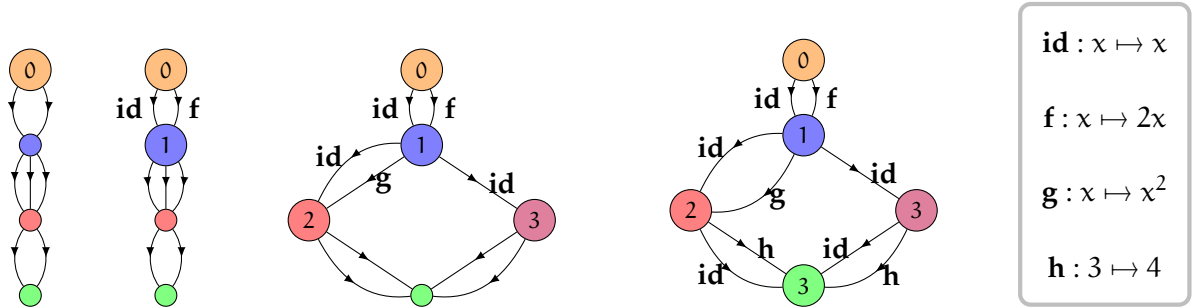
\begin{figure}[h]
\def\xscale{0.75}
\def\yscale{1}
\def\nodescale{0.8}
\begin{subfigure}[t]{0.09\textwidth}
\centering
\begin{tikzpicture}[xscale=\xscale,yscale=\yscale]
\tikzstyle{noeud}=[draw,circle,fill=white,scale=\nodescale*1]
\tikzstyle{arc}=[->-,>=latex]
\tikzstyle{fleche}=[->,>=latex,red,thick]

% \node[noeud,fill=lorange,label=15:$p_0$] (u0) at (0,3) {$0$};
\node[noeud,fill=lorange] (u0) at (0,3) {$0$};
%\node[noeud,fill=lblue,label=0:$q_1$] (u1) at (0,2) {};
\node[noeud,fill=lblue] (u1) at (0,2) {};
%\node[noeud,fill=lred,label=0:$q_2$] (u2) at (0,1) {};
\node[noeud,fill=lred] (u2) at (0,1) {};
%\node[noeud,fill=lgreen,label=-90:$q_3$] (u3) at (0,0) {};
\node[noeud,fill=lgreen] (u3) at (0,0) {};

\draw[arc,white] (u0) to[bend right=30] node[left,midway,white]{$\id$} (u1);

\draw[arc] (u0) to[bend left=45] (u1);
\draw[arc] (u0) to[bend right=45] (u1);
\draw[arc] (u1) to[bend left=45] (u2);
\draw[arc] (u1) to[bend right=45] (u2);
\draw[arc] (u1) to (u2);
\draw[arc] (u2) to[bend left=45] (u3);
\draw[arc] (u2) to[bend right=45] (u3);
\end{tikzpicture}
% \caption{\label{compression_label_iter_1}}
\end{subfigure}\hfill
\begin{subfigure}[t]{0.09\textwidth}
\centering
\begin{tikzpicture}[xscale=\xscale,yscale=\yscale]
\tikzstyle{noeud}=[draw,circle,fill=white,scale=\nodescale*1]
\tikzstyle{arc}=[->-,>=latex]
\tikzstyle{fleche}=[->,>=latex,red,thick]

% \node[noeud,fill=lorange,label=15:$p_0$] (u0) at (0,3) {$0$};
\node[noeud,fill=lorange] (u0) at (0,3) {$0$};
%\node[noeud,fill=lblue,label=0:$p_1$] (u1) at (0,2) {$1$};
\node[noeud,fill=lblue] (u1) at (0,2) {$1$};
%\node[noeud,fill=lred,label=0:$q_2$] (u2) at (0,1) {};
\node[noeud,fill=lred] (u2) at (0,1) {};
%\node[noeud,fill=lgreen,label=-90:$q_3$] (u3) at (0,0) {};
\node[noeud,fill=lgreen] (u3) at (0,0) {};

\draw[arc] (u0) to[bend right=30] node[left,midway]{$\id$} (u1);
\draw[arc] (u0) to[bend right=-30] node[right,midway] {$\textbf{f}$} (u1);

\draw[arc] (u1) to[bend left=45] (u2);
\draw[arc] (u1) to[bend right=45] (u2);
\draw[arc] (u1) to (u2);
\draw[arc] (u2) to[bend left=45] (u3);
\draw[arc] (u2) to[bend right=45] (u3);
\end{tikzpicture}
% \caption{\label{compression_label_iter_2}}
\end{subfigure}\hfill
\begin{subfigure}[t]{0.28\textwidth}
\centering
\begin{tikzpicture}[xscale=\xscale,yscale=\yscale]
\tikzstyle{noeud}=[draw,circle,fill=white,scale=\nodescale*1]
\tikzstyle{arc}=[->-,>=latex]
\tikzstyle{fleche}=[->,>=latex,red,thick]

% \node[noeud,fill=lorange,label=15:$p_0$] (u0) at (0,3) {$0$};
\node[noeud,fill=lorange] (u0) at (0,3) {$0$};
%\node[noeud,fill=lblue,label=0:$p_1$] (u1) at (0,2) {$1$};
\node[noeud,fill=lblue] (u1) at (0,2) {$1$};
%\node[noeud,fill=lred,label=180:$p_2$] (u2) at (-2,1) {$2$};
\node[noeud,fill=lred] (u2) at (-2,1) {$2$};
%\node[noeud,fill=lpurple,label=0:$p'_2$] (u2b) at (+2,1) {$3$};
\node[noeud,fill=lpurple] (u2b) at (+2,1) {$3$};
%\node[noeud,fill=lgreen,label=-90:$q_3$] (u3) at (0,0) {};
\node[noeud,fill=lgreen] (u3) at (0,0) {};

\draw[arc] (u0) to[bend right=30] node[left,midway]{$\id$} (u1);
\draw[arc] (u0) to[bend right=-30] node[right,midway] {$\textbf{f}$} (u1);
\draw[arc] (u1) to[bend right=30] node[left,midway]{$\id$} (u2);
\draw[arc] (u1) to node[right,midway] {$\textbf{g}$} (u2);
\draw[arc] (u1) to node[right,midway]{$\id$} (u2b);

\draw[arc] (u2) to (u3);
\draw[arc] (u2) to[bend right=30] (u3);
\draw[arc] (u2b) to[bend left=30] (u3);
\draw[arc] (u2b) to (u3);

\end{tikzpicture}
% \caption{\label{compression_label_iter_3}}
\end{subfigure}\hfill
\begin{subfigure}[t]{0.28\textwidth}
\centering
\begin{tikzpicture}[xscale=\xscale,yscale=\yscale]
\tikzstyle{noeud}=[draw,circle,fill=white,scale=\nodescale*1]
\tikzstyle{arc}=[->-,>=latex]
\tikzstyle{fleche}=[->,>=latex,red,thick]

%\node[noeud,fill=lorange,label=15:$p_0$] (u0) at (0,3) {$0$};
\node[noeud,fill=lorange] (u0) at (0,3) {$0$};
%\node[noeud,fill=lblue,label=0:$p_1$] (u1) at (0,2) {$1$};
\node[noeud,fill=lblue] (u1) at (0,2) {$1$};
%\node[noeud,fill=lred,label=180:$p_2$] (u2) at (-2,1) {$2$};
\node[noeud,fill=lred] (u2) at (-2,1) {$2$};
%\node[noeud,fill=lpurple,label=0:$p'_2$] (u2b) at (+2,1) {$3$};
\node[noeud,fill=lpurple] (u2b) at (+2,1) {$3$};
%\node[noeud,fill=lgreen,label=-90:$p_3$] (u3) at (0,0) {$3$};
\node[noeud,fill=lgreen] (u3) at (0,0) {$3$};

\draw[arc] (u0) to[bend right=30] node[left,midway]{$\id$} (u1);
\draw[arc] (u0) to[bend right=-30] node[right,midway] {$\textbf{f}$} (u1);
\draw[arc] (u1) to[bend right=30] node[left,midway]{$\id$} (u2);
\draw[arc] (u1) to[bend right=-30] node[right,midway] {$\textbf{g}$} (u2);
\draw[arc] (u1) to node[right,midway]{$\id$} (u2b);

\draw[arc] (u2) to[bend right=30] node[left,midway]{$\id$} (u3);
\draw[arc] (u2) to node[right,midway] {$\textbf{h}$} (u3);
\draw[arc] (u2b) to node[left,midway]{$\id$} (u3);
\draw[arc] (u2b) to[bend right=-30] node[right,midway] {$\textbf{h}$} (u3);

\end{tikzpicture}
% \caption{\label{compression_label_iter_4}}
\end{subfigure}\hfill
\begin{subfigure}[t]{0.15\textwidth}
\centering
\begin{tikzpicture}[xscale=\xscale,yscale=\yscale]
\tikzstyle{noeud}=[draw,circle,fill=white,scale=\nodescale*1]
\tikzstyle{arc}=[->-,>=latex]
\tikzstyle{fleche}=[->,>=latex,red,thick]

\def\x{0}

\node (u1) at (\x+1.5,3) {$\id : x \mapsto x$};
\node (u1) at (\x+1.5,2) {$\textbf{f} : x \mapsto 2x$};
\node (u2) at (\x+1.5,1) {$\textbf{g} : x \mapsto x^2$};;
\node (u3) at (\x+1.5,0) {$\textbf{h} : 3 \mapsto 4$};

\draw[ultra thick, draw=lightgray, rounded corners] (\x+0,-0.5) rectangle (\x+3,3.5);

\end{tikzpicture}
% \caption{\label{compression_label_cipher}}
\end{subfigure}
\caption{From left to right: state of the graph obtained from $\red_\simeq(T)$ after each iteration of Algorithm~\ref{algo:dag_rw} on the example of Figures~\ref{fig:dag_compression_example} and \ref{fig:compression_label_exple}. For the sake of readability, the ciphers $\id$, $\textbf{f}$ and $\textbf{g}$ are given as formulae. Vertices are colored according to their equivalence class with respect to $\sim$ for labelled vertices and to $\simeq$ for unlabelled ones.}
\label{fig:compression_label_iter}
\end{figure}

The following result gives the complexity of Algorithm~\ref{algo:dag_rw} as a function of $C_\sim(n)$ that denotes the complexity of the algorithm used for tree ciphering isomorphism (line 9) on trees of size $n$.

\begin{proposition}\label{prop:compression_cipher}
Let $T$ be a tree made of $n$ different topological subtrees: subtree $i$, $1\leq i\leq n$, is of size $s_i$ and appears $n_i$ times. The time-complexity of Algorithm~\ref{algo:dag_rw} is $O\left(\#T\sum_{i=1}^n n_iC_\sim(s_i)\right)$.
\end{proposition}
\begin{proof}
Any vertex $q$ of $\red_\simeq(T)$ corresponds to a subtree of $T$ that is of size $s_q$ and appears $n_q$ times. Processing $q$ requires to evaluate $O(\#A_q^2)$ isomorphism tests (whose complexity is $C_\sim(s_q)$) and to copy each of the ciphers in $O(\#\attr(T))$ time. Since the number of ciphers is $O(\#A_q)$, the cost of treating $q$ is $O(\#A_q^2 C_\sim(s_q) + \#A_q\#\attr(T))$. $\#A_q\leq n_q$ and $\sum_q n_q = \#T$ yield the result.
\end{proof}

\begin{remark}
    Since the tree ciphering isomorphism problem is graph isomorphism complete \cite{booth1979problems}, one can not expect to derive a fast compression algorithm. However, one can choose to limit the search for isomorphisms (line 9 of Algorithm~\ref{algo:dag_rw}), either in time or in number of operations, resulting in a possibly incomplete compression. Numerical results of Subsection~\ref{ss:timecomplexity} seem to show that the worst-case complexity of our algorithm for tree ciphering isomorphism is due to infrequent pathological cases. One can therefore reasonably anticipate good compression results even with limited computation time.
\end{remark}

\subsection{Frequent pattern mining}

In this subsection, we show how DAG-RW compression can be used to detect frequent common subtrees with identical label distribution in a tree dataset.

\paragraph{DAG-RW compression of a dataset}

We consider a dataset of labelled rooted trees $(T_1,\dots,T_n)$ and are interested in the detection of frequent patterns of the following types: unlabelled subtrees, labelled subtrees, and subtrees with identical label distribution, i.e. subtrees up to tree ciphering. In other words, we would like to know, for each of these types of patterns, which ones appear frequently in the dataset. For the first two, conventional DAG compression allows the preliminary enumeration of all the patterns, whose frequency is evaluated in a second step \cite{azais2020weight}. Here we explore how DAG-RW compression solves the problem for subtrees with identical label distribution.

To this end, we consider the tree $\mathbb{T}$ whose root has $n$ children: the roots of the $T_i$'s. Then we evaluate its DAG-RW compression $\red_\sim(\mathbb{T})$. A vertex $q$ of $\red_\sim(\mathbb{T})$ represents a subtree that appears up to tree ciphering at least once in the dataset. The data which contain this subtree, $\origin(q)$, can be recursively computed in one top-down traversal of the DAG,
$$
\origin(q) =
\bigcup_{\text{edge\,$(p,q)$\,of}\,\red_\sim(\mathbb{T})} \origin(p),
$$
after the following initialisation step: if $p$ is the source of $\red_\sim(\mathbb{T})$, $\origin(p)=\emptyset$, and, if $p$ is the $i$th vertex connected to the source, $\origin(p)=\{i\}$. Then $\#\origin(p)/n$ gives the frequency of subtrees with the same label distribution as $p$ in the dataset. Finally, $\red_\sim(\mathbb{T})$ can be filtered to keep the patterns more frequent than a given threshold.

\paragraph{Real data analysis}
%INEX 2005 and INEX 2006
%https://www.esann.org/sites/default/files/proceedings/legacy/es2016-172.pdf

We conduct frequent pattern mining on two reference datasets from the literature, INEX~2005 and INEX~2006 \cite{denoyer2006inex}. The patterns that we are looking for are topological subtrees, subtrees with common label distribution, and labelled subtrees. Topological subtrees and labelled subtrees are detected by conventional DAG compression \cite{azais2020weight}, while DAG-RW compression is performed for the detection of subtrees with common label distribution. Both algorithms were implemented as modules of the \texttt{Python} library \texttt{treex} \cite{azais2019treex}. Table~\ref{tab:sum:inex} summarises the numbers of patterns and frequent patterns (detected in at least $5\%$ of the data) found in the two datasets. Frequent patterns observed in INEX 2005 and INEX 2006 are given in Figure~\ref{fig:INEX_both}.

\begin{table}[ht]
    \centering
    \begin{tabular}{c|ccc|ccc}
            & \multicolumn{3}{c|}{INEX 2005} & \multicolumn{3}{c}{INEX 2006}  \\
            & \# patterns & \# frequent patterns & \# data & \# patterns & \# frequent patterns & \# data\\ \hline
$\simeq$    & 2073 & 9 &  & 1615 & 11 &  \\
$\sim$      & 2212 & 9 & 4820 & 2430 & 13 & 6053 \\
$\simeq_l$  & 5121 & 144 &  & 3165 & 35 &  \\ % \hline
    \end{tabular}
    \caption{Summary of numbers of patterns and frequent patterns (observed in at least $5\%$ of the data) detected in datasets INEX\,2005 and INEX\,2006 for each of the equivalence relations considered in the paper: $\simeq$ for unlabelled subtrees, $\sim$ for subtrees with identical label distribution, and $\simeq_l$ for labelled subtrees.}
    \label{tab:sum:inex}
\end{table}

\begin{figure}[p!]
\thisfloatpagestyle{empty}
    \centering
    \caption*{\normalsize INEX 2005}
\label{fig:inex2005}
     \begin{subfigure}[b]{\textwidth}
\centering
\begin{tabular}{c|c|c}
\inexI & \inexA & \inexB \\
\inexC & \inexD & \inexE \\
 \inexF & \inexG &\inexH
\end{tabular}

     \end{subfigure}

\bigskip
     
\begin{subfigure}[b]{\textwidth}
\centering
\caption*{\normalsize INEX 2006}
\label{fig:inex2006}
\begin{tabular}{c|c|c}
\inexxM &\inexxA & \inexxB \\
\inexxE & \inexxF & \inexxC \\
\inexxD & \inexxH & \inexxI \\
\inexxJ & \inexxG & \inexxK \\
        & \inexxL &
\end{tabular}

     \end{subfigure}
\caption{Common subtrees with identical label distribution appearing in at least $5\%$ of trees from INEX 2005 (top) and INEX 2006 (bottom) datasets: gray bar gives the frequency of each of them (scale from $0\%$ to $100\%$). Each pattern corresponds to at least 1 subtree with labels under equivalence relation $\simeq_l$: black bars give their frequency when larger than $5\%$ (horizontal red line). }\label{fig:INEX_both}
\end{figure}

In light of the inclusion relation \eqref{eq:eqrels} between the three types of patterns, theory says that the number of topological subtrees is less than the number of subtrees with common label distribution, which is less than the number of labelled subtrees. Of course, we observe this phenomenon in the data but also note that the number of labelled subtrees found in INEX 2005 is larger than the number of data itself, which says that shared labelled subtrees are not so frequent in this dataset. The relation between tree isomorphisms does not tell us anything about the frequency of observed patterns, which highly depends on the distribution of the data. We can nevertheless expect more frequent patterns from a more frequent pattern type. The number of frequent topological subtrees is $9$ in INEX 2005 and $11$ in INEX 2006, which gives parsimonious representations of the datasets. With respectively $9$ and $13$ frequent occurrences, subtrees with identical label distribution provide a representation as parsimonious as topological subtrees but preserve the label distribution, which gives an additional insight on the data. Labelled subtrees keep the information of labels but require respectively $144$ and $35$ frequent patterns to describe the data, which is, in particular for INEX 2005, much less parsimonious.

\FloatBarrier

The subtree with only one node, whatever its label, is present in any subtree under both $\simeq$ and $\sim$, which implies a frequency of $100\%$ in the two datasets. This hides a diversity of labels ($119$ frequent labelled leaves in INEX 2005 and $11$ in INEX 2006), which can be investigated without taking into account the topology (trivial in this case). This phenomenon occurs, to a lesser extent, for other patterns, e.g. $(1,3)$ (first row, third column) in INEX 2005 and $(1,3)$ in INEX 2006. DAG-RW compression allows to discover, without losing the information of the label distribution, that some subtrees are very frequent. Typically, subtree $(1,2)$ in INEX 2005 is present, up to tree ciphering, in $96\%$ of data, while the frequency of its most frequent labelled version is only $25\%$. In addition, it avoids to miss patterns: $5$ frequent subtrees in INEX 2005 and $1$ in INEX 2006 have no labelled version more frequent than $5\%$. To sum up, in both datasets, identifying frequent labelled subtrees lead to much more patterns and less diversity of topologies than subtrees with identical label distribution. Furthermore, contrary to topological patterns, DAG-RW compression separates common topologies with different label distributions, e.g. patterns $(1,2)$ and $(1,3)$ in INEX 2006. As a conclusion, with a representation as parsimonious as from topology only (without losing the label distribution), and a much more parsimonious representation together with a greater diversity of topologies than from labelled patterns, DAG-RW compression achieves a relevant frequent pattern mining on these datasets.

\bibliographystyle{acm}
\bibliography{biblio}

\begin{thebibliography}{10}

\bibitem{aggarwal2014frequent}
{\sc Aggarwal, C.~C., Bhuiyan, M.~A., and Hasan, M.~A.}
\newblock Frequent pattern mining algorithms: A survey.
\newblock In {\em Frequent pattern mining}. Springer, 2014, pp.~19--64.

\bibitem{aho1974design}
{\sc Aho, A., Hopcroft, J., and Ullman, J.}
\newblock {\em The Design and Analysis of Computer Algorithms}.
\newblock Addison-Wesley series in computer science and information processing.
  Addison-Wesley Publishing Company, 1974.

\bibitem{asai2003discovering}
{\sc Asai, T., Arimura, H., Uno, T., and Nakano, S.-i.}
\newblock Discovering frequent substructures in large unordered trees.
\newblock In {\em Discovery Science\/} (Berlin, Heidelberg, 2003), G.~Grieser,
  Y.~Tanaka, and A.~Yamamoto, Eds., Springer Berlin Heidelberg, pp.~47--61.

\bibitem{astrachan2003bubble}
{\sc Astrachan, O.}
\newblock Bubble sort: an archaeological algorithmic analysis.
\newblock {\em ACM Sigcse Bulletin 35}, 1 (2003), 1--5.

\bibitem{avis1996reverse}
{\sc Avis, D., and Fukuda, K.}
\newblock Reverse search for enumeration.
\newblock {\em Discrete Applied Mathematics 65}, 1 (1996), 21--46.
\newblock First International Colloquium on Graphs and Optimization.

\bibitem{azais2019treex}
{\sc Aza{\"\i}s, R., Cerutti, G., Gemmerl{\'e}, D., and Ingels, F.}
\newblock treex: a {Python} package for manipulating rooted trees.
\newblock {\em Journal of Open Source Software 4}, 38 (2019), 1351.

\bibitem{azais2020weight}
{\sc Aza{\"\i}s, R., and Ingels, F.}
\newblock The weight function in the subtree kernel is decisive.
\newblock {\em The Journal of Machine Learning Research 21}, 1 (2020),
  2535--2570.

\bibitem{babai2016graph}
{\sc Babai, L.}
\newblock Graph isomorphism in quasipolynomial time.
\newblock In {\em Proceedings of the forty-eighth annual ACM symposium on
  Theory of Computing\/} (2016), pp.~684--697.

\bibitem{bmmv2002}
{\sc Babilon, R., Matou{\v{s}}ek, J., Maxov{\'a}, J., and Valtr, P.}
\newblock Low-distortion embeddings of trees.
\newblock In {\em Graph Drawing\/} (2002), Springer Berlin Heidelberg,
  pp.~343--351.

\bibitem{booth1979problems}
{\sc Booth, K.~S., and Colbourn, C.~J.}
\newblock {\em Problems polynomially equivalent to graph isomorphism}.
\newblock Computer Science Department, Univ., 1979.

\bibitem{denoyer2006inex}
{\sc Denoyer, L., Gallinari, P., and Vercoustre, A.-M.}
\newblock {Report on the XML Mining Track at INEX 2005 and INEX 2006,
  Categorization and Clustering of XML Documents}.
\newblock In {\em {5th International Workshop of the Initiative for the
  Evaluation of XML Retrieval, INEX 2006}\/} (Dagstuhl, Germany, Dec. 2006),
  N.~Fuhr, M.~Lalmas, S.~Malik, and G.~Kazai, Eds., vol.~4518 of {\em Lecture
  Notes in Computer Science}, {Springer}, pp.~432--443.

\bibitem{fortin1996graph}
{\sc Fortin, S.}
\newblock The graph isomorphism problem.
\newblock Tech. rep., University of Alberta, Canada, 1996.

\bibitem{grohe2021deep}
{\sc Grohe, M., Schweitzer, P., and Wiebking, D.}
\newblock Deep {W}eisfeiler {L}eman.
\newblock In {\em Proceedings of the 2021 ACM-SIAM Symposium on Discrete
  Algorithms (SODA)\/} (2021), SIAM, pp.~2600--2614.

\bibitem{ingels2021isomorphic}
{\sc Ingels, F., and Aza\"is, R.}
\newblock Isomorphic unordered labeled trees up to substitution ciphering.
\newblock In {\em Combinatorial Algorithms\/} (2021), pp.~385--399.

\bibitem{ingels2022enumeration}
{\sc Ingels, F., and Aza{\"\i}s, R.}
\newblock Enumeration of irredundant forests.
\newblock {\em Theoretical Computer Science\/} (2022).

\bibitem{jimenez2010}
{\sc Jim\'enez, A., Berzal, F., and Cubero, J.-C.}
\newblock Frequent tree pattern mining: A survey.
\newblock {\em Intell. Data Anal. 14\/} (01 2010), 603--622.

\bibitem{knuth2005art}
{\sc Knuth, D.~E.}
\newblock {\em The Art of Computer Programming, Volume 4, Fascicle 2:
  Generating All Tuples and Permutations (Art of Computer Programming)}.
\newblock Addison-Wesley Professional, 2005.

\bibitem{mckay2014practical}
{\sc McKay, B.~D., and Piperno, A.}
\newblock Practical graph isomorphism, ii.
\newblock {\em Journal of symbolic computation 60\/} (2014), 94--112.

\bibitem{schoning1988}
{\sc Sch\"oning, U.}
\newblock Graph isomorphism is in the low hierarchy.
\newblock {\em Journal of Computer and System Sciences 37}, 3 (1988), 312--323.

\bibitem{valiente2001}
{\sc Valiente, G.}
\newblock An efficient bottom-up distance between trees.
\newblock In {\em Proceedings Eighth Symposium on String Processing and
  Information Retrieval\/} (2001), pp.~212--219.

\bibitem{valiente2013algorithms}
{\sc Valiente, G.}
\newblock {\em Algorithms on Trees and Graphs}.
\newblock Springer Science \& Business Media, 2013.

\bibitem{weisfeiler1968reduction}
{\sc Weisfeiler, B., and Lehman, A.}
\newblock A reduction of a graph to a canonical form and an algebra arising
  during this reduction.
\newblock {\em Nauchno-Technicheskaya Informatsia 2}, 9 (1968), 12--16.

\bibitem{zaki2002}
{\sc Zaki, M.~J.}
\newblock Efficiently mining frequent trees in a forest.
\newblock In {\em Proceedings of the Eighth ACM SIGKDD International Conference
  on Knowledge Discovery and Data Mining\/} (New York, NY, USA, 2002), KDD '02,
  Association for Computing Machinery, pp.~71--80.

\bibitem{zhang2015number}
{\sc Zhang, Y., and Zhang, Y.~Z.}
\newblock On the number of leaves in a random recursive tree.
\newblock {\em Brazilian Journal of Probability and Statistics 29}, 4 (2015),
  897--908.

\end{thebibliography}

\pagebreak
\appendix

\section{Evolution of the size of the search space}\label{app:size_search_size}

This appendix investigates how the various manipulations of the system presented in Section~\ref{s:treeciphering} modify the size of the search space, defined in \eqref{size} as the number of possible ways to complete $\isom$ from a distribution of nodes (not already mapped via $\isom$) between bags and collections.

\subsection{Reduction factors}

\paragraph{Splitting bags} When a bag of size $n$ is split in two due to the action of \textsc{SplitChildren}, say into bags of size $p$ and $n-p$, then we go from $n!$ possibilities to $p!(n-p)!$, reducing the space by a factor $\binom{n}{p}$.

\paragraph{Splitting collections} When \textsc{SplitChildren} is applied to a collection, several sets can be split and put together in the same collection. Assume that we split $k$ sets of size $n$ into sets of size $p$ and $n-p$, $k\leq\min(\#\bfC_n,\degr(T))$. Then we go from $\#\bfC_n! \times n!^{\#\bfC_n}$ possibilities to $(\#\bfC_n-k)!\times n!^{\#\bfC_n-k}\times k!p!^k\times k!(n-p)!^k$. The space is therefore modified by a factor $\displaystyle\frac{1}{k!}\binom{\#\bfC_n}{k}\binom{n}{p}^k$, which is not always greater than $1$. However, this does not show that the search space can grow because all the parameters are highly interrelated: the values for which the reduction factor is less than $1$ may not be feasible. We actually exhibit in Appendix~\ref{app:a:counterexample} an instance of a tree for which the search space locally grows along the algorithm.

\paragraph{Mapping parents} When recursively mapping the parents of two nodes, if they were present in a bag $\bfB$, we go from $\#\bfB!$ to $(\#\bfB-1)!$, reducing the space by a factor $\#\bfB$; if they were present in a collection $\bfC_n$, then we go from $\#\bfC_n!\times n!^{\#\bfC_n}$ possibilities to $(\#\bfC_n-1)!\times n!^{\#\bfC_n-1}\times (n-1)!$, reducing the space by a factor $n\#\bfC_n$.

\paragraph{Deductions rules} When applying Deduction~Rule~\ref{ded:coll_label} and splitting a collection $\bfC_n$ into two new collections $\mathbf{C}'_n$ and $\mathbf{C}''_n$, say $\#\bfC_n = p + q$, $\#\mathbf{C}'_n=p$ and $\#\mathbf{C}''_n=q$, then we go from $(p+q)! \times n!^{p+q}$ possibilities to $p!\times n!^p\times q!\times n!^q$, reducing the space by a factor $\binom{p+q}{q}$. On the other hand, Deduction~Rules~\ref{ded:bags}, \ref{ded:coll_label_new} and \ref{ded:coll_cardinal} do not modify the size of the search space. However, these rules still modify the system in an irrevocable way: mapping nodes (Deduction Rule~\ref{ded:bags}), mapping labels (Deduction Rule~\ref{ded:coll_label_new}) or transferring nodes from a collection to a bag (Deduction Rule~\ref{ded:coll_cardinal}).

\subsection{Does the search space only shrink? A counterexample}
\label{app:a:counterexample}

The following example shows that the search space can locally grow through the deduction steps given in Subsections~\ref{ss:ded:topol} and \ref{ss:ded:label}.

\bigskip

\begin{minipage}[c]{0.475\textwidth}
\centering

Step 1/6 (histogram of labels)

$N(\bags,\collections)= 12!=479,001,600$

\def\xscale{0.35}
\def\yscale{0.7}
\def\nodescale{0.7}
\begin{tikzpicture}[xscale=\xscale,yscale=\yscale]
\tikzstyle{noeud}=[draw,circle,fill=white,scale=\nodescale*1]
\tikzstyle{attribut}=[scale=\nodescale*1,font=\bf]
\tikzstyle{arc}=[-,>=latex]
\tikzstyle{fleche}=[->,>=latex,red,thick]

\def\x{2}
\def\y{1.5}

\node[noeud,label=15:$u_1$] (u1) at (5*\x,2*\y) {$A$};
\node[noeud,label=-15:$u_2$] (u2) at (2.5*\x,1*\y) {$A$};
\node[noeud,label=-15:$u_3$] (u3) at (7.5*\x,1*\y) {$B$};
\node[noeud,label=-90:$u_4$] (u4) at (10*\x,1*\y) {$B$};
\node[noeud,label=-90:$u_5$] (u5) at (1*\x,0) {$C$};
\node[noeud,label=-90:$u_6$] (u6) at (2*\x,0) {$D$};
\node[noeud,label=-90:$u_7$] (u7) at (3*\x,0) {$E$};
\node[noeud,label=-90:$u_8$] (u8) at (4*\x,0) {$F$};
\node[noeud,label=-90:$u_9$] (u9) at (6*\x,0) {$C$};
\node[noeud,label=-90:$u_{10}$] (u10) at (7*\x,0) {$D$};
\node[noeud,label=-90:$u_{11}$] (u11) at (8*\x,0) {$E$};
\node[noeud,label=-90:$u_{12}$] (u12) at (9*\x,0) {$F$};

\draw[arc] (u1)--(u2) ;
\draw[arc] (u1)--(u3) ;
\draw[arc] (u1)--(u4) ;
\draw[arc] (u2)--(u5) ;
\draw[arc] (u2)--(u6) ;
\draw[arc] (u2)--(u7) ;
\draw[arc] (u2)--(u8) ;
\draw[arc] (u3)--(u9) ;
\draw[arc] (u3)--(u10) ;
\draw[arc] (u3)--(u11) ;
\draw[arc] (u3)--(u12) ;

\draw[rounded corners,ultra thick,gray] (\x-1,-1) rectangle (10*\x+1,2*\y+1);

\end{tikzpicture}
\end{minipage}\hfill
\begin{minipage}[c]{0.475\textwidth}
\centering

Step 2/6 (depth)

$N(\bags,\collections)= 3!\times 8! = 241,920$

\def\xscale{0.35}
\def\yscale{0.7}
\def\nodescale{0.7}
\begin{tikzpicture}[xscale=\xscale,yscale=\yscale]
\tikzstyle{noeud}=[draw,circle,fill=white,scale=\nodescale*1]
\tikzstyle{attribut}=[scale=\nodescale*1,font=\bf]
\tikzstyle{arc}=[-,>=latex]
\tikzstyle{fleche}=[->,>=latex,red,thick]

\def\x{2}
\def\y{1.5}

% \node at (5*\x,2.5*\y) {$T_1$};

\node[noeud,draw=lred,ultra thick,label=15:$u_1$] (u1) at (5*\x,2*\y) {$A$};
\node[noeud,label=-15:$u_2$] (u2) at (2.5*\x,1*\y) {$A$};
\node[noeud,label=-15:$u_3$] (u3) at (7.5*\x,1*\y) {$B$};
\node[noeud,label=-90:$u_4$] (u4) at (10*\x,1*\y) {$B$};
\node[noeud,label=-90:$u_5$] (u5) at (1*\x,0) {$C$};
\node[noeud,label=-90:$u_6$] (u6) at (2*\x,0) {$D$};
\node[noeud,label=-90:$u_7$] (u7) at (3*\x,0) {$E$};
\node[noeud,label=-90:$u_8$] (u8) at (4*\x,0) {$F$};
\node[noeud,label=-90:$u_9$] (u9) at (6*\x,0) {$C$};
\node[noeud,label=-90:$u_{10}$] (u10) at (7*\x,0) {$D$};
\node[noeud,label=-90:$u_{11}$] (u11) at (8*\x,0) {$E$};
\node[noeud,label=-90:$u_{12}$] (u12) at (9*\x,0) {$F$};

\draw[arc] (u1)--(u2) ;
\draw[arc] (u1)--(u3) ;
\draw[arc] (u1)--(u4) ;
\draw[arc] (u2)--(u5) ;
\draw[arc] (u2)--(u6) ;
\draw[arc] (u2)--(u7) ;
\draw[arc] (u2)--(u8) ;
\draw[arc] (u3)--(u9) ;
\draw[arc] (u3)--(u10) ;
\draw[arc] (u3)--(u11) ;
\draw[arc] (u3)--(u12) ;

\draw[rounded corners,ultra thick,gray] (\x-1,-1) rectangle (9*\x+1,0.5);

\draw[rounded corners,ultra thick,gray] (2.5*\x-1,\y-1) rectangle (10*\x+1,\y+0.5);

\end{tikzpicture}
\end{minipage}

\bigskip

\begin{minipage}[c]{0.475\textwidth}
\centering

Step 3/6 (equivalence class)\\
and step 4/6 (parents)

$N(\bags,\collections)= 2!\times 8! = 80,640$

\def\xscale{0.35}
\def\yscale{0.7}
\def\nodescale{0.7}
\begin{tikzpicture}[xscale=\xscale,yscale=\yscale]
\tikzstyle{noeud}=[draw,circle,fill=white,scale=\nodescale*1]
\tikzstyle{attribut}=[scale=\nodescale*1,font=\bf]
\tikzstyle{arc}=[-,>=latex]
\tikzstyle{fleche}=[->,>=latex,red,thick]

\def\x{2}
\def\y{1.5}

% \node at (5*\x,2.5*\y) {$T_1$};

\node[noeud,draw=lred,ultra thick,label=15:$u_1$] (u1) at (5*\x,2*\y) {$A$};
\node[noeud,label=-15:$u_2$] (u2) at (2.5*\x,1*\y) {$A$};
\node[noeud,label=-15:$u_3$] (u3) at (7.5*\x,1*\y) {$B$};
\node[noeud,draw=lblue,ultra thick,label=-90:$u_4$] (u4) at (10*\x,1*\y) {$B$};
\node[noeud,label=-90:$u_5$] (u5) at (1*\x,0) {$C$};
\node[noeud,label=-90:$u_6$] (u6) at (2*\x,0) {$D$};
\node[noeud,label=-90:$u_7$] (u7) at (3*\x,0) {$E$};
\node[noeud,label=-90:$u_8$] (u8) at (4*\x,0) {$F$};
\node[noeud,label=-90:$u_9$] (u9) at (6*\x,0) {$C$};
\node[noeud,label=-90:$u_{10}$] (u10) at (7*\x,0) {$D$};
\node[noeud,label=-90:$u_{11}$] (u11) at (8*\x,0) {$E$};
\node[noeud,label=-90:$u_{12}$] (u12) at (9*\x,0) {$F$};

\draw[arc] (u1)--(u2) ;
\draw[arc] (u1)--(u3) ;
\draw[arc] (u1)--(u4) ;
\draw[arc] (u2)--(u5) ;
\draw[arc] (u2)--(u6) ;
\draw[arc] (u2)--(u7) ;
\draw[arc] (u2)--(u8) ;
\draw[arc] (u3)--(u9) ;
\draw[arc] (u3)--(u10) ;
\draw[arc] (u3)--(u11) ;
\draw[arc] (u3)--(u12) ;

\draw[rounded corners,ultra thick,gray] (\x-1,-1) rectangle (9*\x+1,0.5);

\draw[rounded corners,ultra thick,gray] (2.5*\x-1,\y-.75) rectangle (7.5*\x+2.5,\y+0.5);

\end{tikzpicture}
\end{minipage}\hfill
\begin{minipage}[c]{0.475\textwidth}
\centering

Step 5/6 (from bags to collections)

$N(\bags,\collections)= \left(2!\times 1!^2\right)\times\left(4!\times 2!^4\right)=768$

\def\xscale{0.7}
\def\yscale{0.7}
\def\nodescale{0.7}
\begin{tikzpicture}[xscale=\xscale,yscale=\yscale]
\tikzstyle{noeud}=[draw,circle,fill=white,scale=\nodescale*1]
\tikzstyle{attribut}=[scale=\nodescale*1,font=\bf]
\tikzstyle{arc}=[->-,>=latex]

\def\x{1.5}
\def\y{2}

%%%%% Collection 1 %%%%%

\node[noeud,label=-15:$u_2$] (u3) at (\x,1*\y) {$A$};
\draw[rounded corners,ultra thick,gray] (\x-.5,\y-.75) rectangle (1*\x+1.25,\y+.5);

\node[noeud,label=-15:$u_3$] (u4) at (2.5*\x,1*\y) {$B$};
\draw[rounded corners,ultra thick,gray] (2.5*\x-.5,\y-.75) rectangle (2.5*\x+1.25,\y+.5);

\node at (1.75*\x,\y+1) {$n=1$};
\node at (0,\y) {$C_1$};
\draw[ultra thick,lightgray,dashed] (-.5,\y/2)--(2.5*\x+1.5,\y/2);

\draw[ultra thick,lightgray,rounded corners] (-.5,\y/2)--(-.5,\y+1.5)--(2.5*\x+1.5,\y+1.5)--(2.5*\x+1.5,\y/2);

\node[ultra thick,circle,draw=lightgray,fill=white] at (-.5,\y+1.5) {$\bfC$};

%%%%% Collection 3 %%%%%

\def\xstep{0*\x}
\def\ystep{-1.75*\y}

\node[noeud,label=-15:$u_{5}$] (u11) at (\x+\xstep,1*\y+\ystep) {$C$};
\node[noeud,label=-15:$u_{9}$] (u14) at (2*\x+\xstep,1*\y+\ystep) {$C$};
\draw[rounded corners,ultra thick,gray] (\x-.5+\xstep,\y-.75+\ystep) rectangle (2*\x+1.25+\xstep,\y+.5+\ystep);

\draw[rounded corners,ultra thick,gray] (3.5*\x-.5+\xstep,\y-.75+\ystep) rectangle (4.5*\x+1.25+\xstep,\y+.5+\ystep);
\node[noeud,label=-15:$u_{6}$] (u11) at (3.5*\x+\xstep,1*\y+\ystep) {$D$};
\node[noeud,label=-15:$u_{10}$] (u14) at (4.5*\x+\xstep,1*\y+\ystep) {$D$};

\draw[rounded corners,ultra thick,gray] (\x-.5+\xstep,-.75+\ystep) rectangle (2*\x+1.25+\xstep,.5+\ystep);
\node[noeud,label=-15:$u_{7}$] (u11) at (\x+\xstep,\ystep) {$E$};
\node[noeud,label=-15:$u_{11}$] (u14) at (2*\x+\xstep,\ystep) {$E$};

\draw[rounded corners,ultra thick,gray] (3.5*\x-.5+\xstep,-.75+\ystep) rectangle (4.5*\x+1.25+\xstep,.5+\ystep);
\node[noeud,label=-15:$u_{8}$] (u11) at (3.5*\x+\xstep,\ystep) {$F$};
\node[noeud,label=-15:$u_{12}$] (u14) at (4.5*\x+\xstep,\ystep) {$F$};

\node at (2.75*\x+\xstep,\y+1+\ystep) {$n=2$};
\node at (+\xstep,\y/2+\ystep) {$C'_1$};
\draw[ultra thick,lightgray,dashed] (-.5+\xstep,-\y/2+\ystep)--(4.5*\x+1.5+\xstep,-\y/2+\ystep);

\draw[ultra thick,lightgray,rounded corners] (-.5+\xstep,-\y/2+\ystep)--(-.5+\xstep,\y+1.5+\ystep)--(4.5*\x+1.5+\xstep,\y+1.5+\ystep)--(4.5*\x+1.5+\xstep,-\y/2+\ystep);

\node[ultra thick,circle,draw=lightgray,fill=white] at (-.5+\xstep,\y+1.5+\ystep) {$\mathbf{C'}$};

\end{tikzpicture}
\end{minipage}

\bigskip

At this stage, Deduction Rule~\ref{ded:coll_label} splits $u_2$ and $u_3$ from $\bfC_1$ using mapped labels ($A$ and $B$), which yields $N(\bags,\collections) = 4!\times 2!^4 = 384$ (from $\mathbf{C}'_2$ only).

From one of these new collections (say the one with $u_2$), one can create a bag that will map $u_2$ and its counterpart thanks to Deduction Rule~\ref{ded:bags}. This does not change the size of the search space.

\textsc{SplitChildren} is then called on $\mathbf{C}'_2$ to isolate nodes $u_5$, $u_6$, $u_7$ and $u_8$ (children of $u_2$) and nodes $u_9$, $u_{10}$, $u_{11}$ and $u_{12}$ (children of $u_3$). One obtains the following partition.

\def\xscale{0.7}
\def\yscale{0.7}
\def\nodescale{0.7}
\begin{tikzpicture}[xscale=\xscale,yscale=\yscale]
\tikzstyle{noeud}=[draw,circle,fill=white,scale=\nodescale*1]
\tikzstyle{attribut}=[scale=\nodescale*1,font=\bf]
\tikzstyle{arc}=[->-,>=latex]

\def\x{1.5}
\def\y{2}

\node[noeud,label=-15:$u_5$] (u3) at (\x,1*\y) {$C$};
\draw[rounded corners,ultra thick,gray] (\x-.5,\y-.75) rectangle (1*\x+1.25,\y+.5);

\node[noeud,label=-15:$u_6$] (u4) at (2.5*\x,1*\y) {$D$};
\draw[rounded corners,ultra thick,gray] (2.5*\x-.5,\y-.75) rectangle (2.5*\x+1.25,\y+.5);

\node[noeud,label=-15:$u_7$] (u4) at (4*\x,1*\y) {$E$};
\draw[rounded corners,ultra thick,gray] (4*\x-.5,\y-.75) rectangle (4*\x+1.25,\y+.5);

\node[noeud,label=-15:$u_8$] (u4) at (5.5*\x,1*\y) {$F$};
\draw[rounded corners,ultra thick,gray] (5.5*\x-.5,\y-.75) rectangle (5.5*\x+1.25,\y+.5);

\node at (3.25*\x,\y+1) {$n=1$};
\node at (0,\y) {$C''_1$};
\draw[ultra thick,lightgray,dashed] (-.5,\y/2)--(5.5*\x+1.5,\y/2);

\draw[ultra thick,lightgray,rounded corners] (-.5,\y/2)--(-.5,\y+1.5)--(5.5*\x+1.5,\y+1.5)--(5.5*\x+1.5,\y/2);

\node[ultra thick,circle,draw=lightgray,fill=white] at (-.5,\y+1.5) {$\mathbf{C''}$};

\def\xstep{8*\x}

\node[noeud,label=-15:$u_9$] (u3) at (\x+\xstep,1*\y) {$C$};
\draw[rounded corners,ultra thick,gray] (\x+\xstep-.5,\y-.75) rectangle (1*\x+\xstep+1.25,\y+.5);

\draw[rounded corners,ultra thick,gray] (2.5*\x+\xstep-.5,\y-.75) rectangle (2.5*\x+\xstep+1.25,\y+.5);
\node[noeud,label=-15:$u_{10}$] (u4) at (2.5*\x+\xstep,1*\y) {$D$};

\draw[rounded corners,ultra thick,gray] (4*\x+\xstep-.5,\y-.75) rectangle (4*\x+\xstep+1.25,\y+.5);
\node[noeud,label=-15:$u_{11}$] (u4) at (4*\x+\xstep,1*\y) {$E$};

\draw[rounded corners,ultra thick,gray] (5.5*\x+\xstep-.5,\y-.75) rectangle (5.5*\x+\xstep+1.25,\y+.5);
\node[noeud,label=-15:$u_{12}$] (u4) at (5.5*\x+\xstep,1*\y) {$F$};

\node at (3.25*\x+\xstep,\y+1) {$n=1$};
\node at (0+\xstep,\y) {$C'''_1$};
\draw[ultra thick,lightgray,dashed] (-.5+\xstep,\y/2)--(5.5*\x+\xstep+1.5,\y/2);

\draw[ultra thick,lightgray,rounded corners] (-.5+\xstep,\y/2)--(-.5+\xstep,\y+1.5)--(5.5*\x+1.5+\xstep,\y+1.5)--(5.5*\x+1.5+\xstep,\y/2);

\node[ultra thick,circle,draw=lightgray,fill=white] at (-.5+\xstep,\y+1.5) {$\mathbf{C'''}$};

\end{tikzpicture}

\bigskip

The related cardinality is $N(\bags,\collections)= \left(4! \times 1!^4\right)^2 = 576$, which illustrates a local growth of the search space. It should be however noticed that any mapping of labels in one of the collections yields a deduction in the other: the 576 mappings are not all feasible.

\section{Optimizing the number of states of the backtracking tree}
\label{proof:backtracking_tree}

When it comes to backtracking in Subsection~\ref{ss:backtracking}, the aim is to explore the various possible associations between, on the one hand, nodes within bags and, on the other hand, sets within collections. Any backtracking strategy for exploring these choices implicitly induces a tree structure, with leaves representing terminal states. The smaller the tree, the more efficient the algorithm is likely to be. In this appendix, we develop a mathematical model of the size of this tree and derive a backtracking strategy to minimise it.

\subsection{Preliminaries}\label{app:prelim}

\paragraph{Enumeration of permutations} Enumerating all possible mappings between the nodes of a bag, or all the bags that can be constructed by associating the sets of a collection, amounts in both cases to enumerating permutations. We assume that we have two sets of $n\geq 2$ objects to be associated together. The associated enumeration tree has $n!$ leaves (one for each permutation of $\lbrace 1,\dots, n\rbrace$) and as many internal nodes as necessary to build these $n!$ permutations by mapping only one additional pair of objects at a time.

The size of the enumeration tree (except the root) $a_n$ is the number of operations required to construct all the permutations of $n$ elements \cite{knuth2005art},
\begin{equation}\label{a_n}
a_n=n!\sum_{i=1}^{n-1} \frac{1}{i!}.
\end{equation}
It should be noted that $a_n$ satisfies the recurrence equation $a_n=n (1 + a_{n-1})$ with $a_1=0$.

\paragraph{Nesting enumeration trees} The tree introduced above only allows to build mappings between one pair of sets of objects, but, in our case, the system is made of several pairs of sets of objects to be mapped together. We assume that we have two sets of $p\geq2$ objects on the one hand, and two sets of $q\geq 2$ objects on the other.

Since we have to enumerate the $p!$ permutations of the first pair, and the $q!$ permutations of the second pair, and since these permutations are independent, we can build mappings incrementally, one pair of objects at a time, in any desired way.

Our aim here is to minimise the size of the overall enumeration tree. Let $a_{p,q}$ be the size of the optimal enumeration tree for those two pairs of respective size $p$ and $q$. It is easy to see that $a_{2,2}=6$, $a_{p,1}=a_p$ and $a_{1,q}=a_q$. In the general case of an instance of size $(p,q)$, we can map two objects from the pair of size $p$ (which leads to a new instance of size $(p-1,q)$), or from the pair of size $q$ (which leads to a new instance of size $(p,q-1)$). Therefore, $a_{p,q}$ verifies the following recurrence relation, echoing the one of $a_n$,
$$a_{p,q} = \min \big[ p(1+a_{p-1,q}), q(1+a_{p,q-1})\big].$$
We have the following result.

\begin{lemma}\label{lemma:two_bags}For any $p\geq2$ and $q\geq 2$, $a_{p,q}=a_{\min(p,q)} + (\min(p,q))! a_{\max(p,q)}$.
\end{lemma}
\begin{proof}
We assume $p\leq q$ without loss of generality and proceed by induction on $p$ and $q$. One can remark that $a_{2,2}= a_2 + 2! a_2$ and $a_{1,q}= a_1+ 1! a_q$. From now on, we assume that the result is verified for all couples $(p',q')$ so that either $p'\leq p$ and $q'<q$, or $p'<p$ and $q'\leq q$.

If we perform a mapping on the pair of size $p$, then $p-1\leq q$ and, by induction, $a_{p-1,q}=a_{p-1}+(p-1)!a_q$, i.e. $p(1+a_{p-1,q}) = a_p+p!a_q$ after simplification.

If we perform a mapping on the pair of size $q$, then two cases can occur: either $q-1<p$ (in which case, $p=q$) or $p\leq q-1$. In the first case, with $p=q$, we have $a_{p,q-1}=a_{q-1}+(q-1)!a_p$ by induction, i.e. $q(1+a_{p,q-1})=a_q+q!a_p=a_p+p!a_q$ after simplification. Consequently, $a_{p,q}=a_p+p!a_q$ as claimed.

In the second case $p\leq q-1$, we have $a_{p,q-1}=a_p+p!a_{q-1}$ by induction, then, with $qa_{q-1}=a_q-q$, $q(1+a_{p,q-1})=q(1+a_p+p!a_{q-1})=q(1+a_p) + p!(a_q-q)$. Thus, we need to show
$$q(1+a_p)+p!(a_q-q) \geq a_p+p!a_q,$$
which is equivalent to
$$\frac{1}{p!} + \left(1-\frac{1}{q}\right) \frac{a_p}{p!}\geq 1.$$
Since $p\leq q$, $1-1/q\geq 1-1/p$, so it suffices to prove that $\displaystyle\frac{1}{p!} + \left(1-\frac{1}{p}\right)\sum_{i=1}^{p-1}\frac{1}{i!}\geq 1$ for any $p\geq 2$. This holds for $p=2$. In addition, $\left(1-\frac{1}{p}\right)\displaystyle\sum_{i=1}^{p-1}\frac{1}{i!}$ is an increasing function of $p$ whose value at $p=3$ is $1$, which states the result for any $p\geq 3$.
\end{proof}

\begin{minipage}[c]{0.55\textwidth}
Consequently, $a_{p,q}$ is the size of the enumeration tree that begins with the permutations of the set of size $\min(p,q)$, followed, for each of its $\min(p,q)!$ leaves by the permutations of set of size $\max(p,q)$ (see Figure~\ref{fig:nesting_construction_tree}). It is therefore optimal to enumerate first the permutations of the smallest set, then the ones of the largest set.

\medskip

Furthermore, by immediate induction, if we have to process more than two pairs of sets of objects, the permutation tree size is minimal if we process them by increasing size.
\end{minipage}
\hfill
\begin{minipage}[c]{0.4\textwidth}

    \begin{tikzpicture}[scale=0.5]

\draw[ultra thick,lred,fill=lred!50!white,rounded corners] (0,0)--(-2,-2)--(2,-2)--cycle;

\node (r) at (0,0) {$\bullet$};

\node at (0,-1.25) {$a_p$};

\node (l1) at (-2,-2) {$\bullet$};
\node at (0,-3) {$\dots$};
\node (l2) at (2,-2) {$\bullet$};

\draw[ultra thick,lblue,fill=lblue!50!white,rounded corners] (-2,-2)--(-3,-4)--(-1,-4)--cycle;

\draw[ultra thick,lblue,fill=lblue!50!white,rounded corners] (2,-2)--(1,-4)--(3,-4)--cycle;

\node at (-2,-3.25) {$a_q$};
\node at (2,-3.25) {$a_q$};

\draw [decorate,decoration={brace,amplitude=5pt,raise=-2ex}] (3,-5)-- (-3,-5);

\node at (0,-5.25) {$p!$};

\end{tikzpicture}
\centering

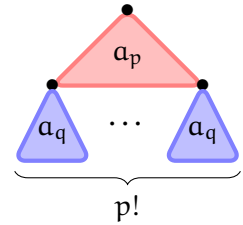
\captionof{figure}{Optimal nesting of two permutation trees for two pairs of sets of respective size $p$ and $q$, $p\leq q$.}
\label{fig:nesting_construction_tree}
\end{minipage}

This result can easily be related to the original problem of backtracking: if there are only bags in the system, then they must be processed in increasing size in order to obtain the smallest possible enumeration tree. However, this does not allow to simultaneously manage bags and collections at this point.

\subsection{Model of the size of the backtracking tree}\label{app:variadic}

Given a finite sequence of integer tuples $(n_i,\alpha_i)$, $1\leq i\leq p$, with $n_i\geq 2$ and $\alpha_i\geq 1$, we define the following variadic function,
\begin{equation}\label{variadic_f}
f\left((n_1,\alpha_1),\dots,(n_p,\alpha_p)\right)=\begin{cases}
\alpha_1 a_{n_1} & \text{if } $p=1$,\\
\alpha_1\left[a_{n_1}+n_1!f\left((n_2,\alpha_2),\dots,(n_p,\alpha_p)\right)\right] & \text{otherwise}.
\end{cases}
\end{equation}

When appropriate tuples are chosen, the value of $f$ corresponds to the size of the backtracking tree (up to the root) associated with the enumeration of mappings from bags and collections, according to the following principles.

\paragraph{Link between $f$ and bags} We saw in the previous section that, if we have two bags of sizes $n$ and $m$, treating them in this order induces a backtracking tree of size $1 + a_n + n! a_m$. If we now have $p$ bags of respective sizes $n_1,\dots,n_p$, and we process them in this order, then the resulting tree has 
$1+f\left((n_1,1),\dots,(n_p,1)\right)$ states. In other words, a bag of size $n$ processed in position $k$ contributes to the backtracking tree as the tuple $(n,1)$ as $k$th input variable of $f$. It should be recalled that Lemma~\ref{lemma:two_bags} proves that the backtracking tree size is minimal when bags are processed in increasing order.

\paragraph{Link between $f$ and collections} We explained in Subsection~\ref{ss:backtracking} that two strategies can be considered for processing collections. Let $\bfC$ be a collection and $n\in\supp(\bfC)$:
\begin{itemize}
    \item either we create all $\#\bfC_n!$ possible bags by associating the subsets of $\bfC_n$, then we process these bags one after the other;
    \item or we create a bag, we process it entirely, then we create an other bag (from the collection that now contains $\#\bfC_n-1$ subsets), which we process entirely, and so on.
\end{itemize}

In terms of backtracking tree size, by definition \eqref{variadic_f} of $f$, the first strategy amounts to considering as input variables of $f$ the tuple $(\#\bfC_n,1)$ first then $\#\bfC_n$ times the tuple $(n,1)$. The second strategy amounts to considering the tuple $(n,\#\bfC_n)$ first (one branches on $\#\bfC_n$ possibilities -- mapping the first set of $C_{1,n}$ with each of the others from $C_{2,n}$ --, then process the tree generated by a bag of size $n$), then the tuple $(n,\#\bfC_n-1)$, and so on to the tuple $(n,1)$. We state in the following lemma that the second strategy minimises $f$ and thus the size of the resulting backtracking tree, establishing incidentally Proposition~\ref{prop:strategy_for_processing_collections}.

\begin{lemma}\label{strategy_collections}
For any $n\geq2$ and $\alpha\geq 2$,
$$f\big((n,\alpha),(n,\alpha-1),\dots,(n,1)\big) < 
f\big((\alpha,1),\underbrace{(n,1),\dots,(n,1)}_{\alpha}\big).
$$
\end{lemma}
\begin{proof}
By induction on $\alpha$,
\begin{eqnarray*}
    f\big((n,\alpha),(n,\alpha-1),\dots,(n,1)\big) &=&a_n\displaystyle\sum_{k=0}^{\alpha-1}\frac{\alpha!}{(\alpha-1-k)!} (n!)^k,\\
    f\big(\underbrace{(n,1),\dots,(n,1)}_{\alpha}\big)&=&a_n \displaystyle\sum_{k=0}^{\alpha-1}(n!)^k,\\
\end{eqnarray*}
which shows
$$ f\big((\alpha,1),\underbrace{(n,1),\dots,(n,1)}_{\alpha}\big) = a_\alpha + \alpha! a_n \displaystyle\sum_{k=0}^{\alpha-1}(n!)^k.$$
The conclusion is immediate when comparing term by term each of the two sums.
\end{proof}

The second strategy is therefore adopted in Algorithm~\ref{backtrack}. Consequently, a collection $\bfC_n$ contributes to the backtracking tree size as the tuples $(n,\#\bfC_n),\dots,(n,1)$ as input variables of $f$. The size of the backtracking tree generated by $\bfC_n$ is thus expressed as $1+f\big((n,\#\bfC_n),(n,\#\bfC_n-1),\dots,(n,1)\big)$.

\paragraph{General case: bags and collections} When the remaining search space is organised with both bags and collections, a bag $\bfB$ processed in position $k$ contributes to the backtracking tree size as $\left(b+\sum_{i=1}^c \#\bfC_{n_i}\right)$-th input variable $(\#\bfB,1)$, with $b$ the number of bags processed before $k$ and $c$ the number of collections $\bfC_{n_i}$ processed before $k$; a collection $\bfC$ contributes as input variables $(n,\#\bfC_n),\dots,(n,1)$ at positions $b+\sum_{i=1}^c \#\bfC_{n_i}$, \dots, $b+\sum_{i=1}^c \#\bfC_{n_i}+\bfC_n-1$.

\begin{remark}
The construction of the backtracking tree whose size is estimated by $f$ assumes that the permutation trees obtained from bags and collections are independent, i.e. that no branches are pruned thanks to deductions, recursive parent mapping, or \textsc{SplitChildren}, contrary to what is implemented in Algorithm~\ref{backtrack}. In other words, the model developed in this section provides an upper-bound of the actual backtracking tree size.
\end{remark}

\subsection{Proof of Proposition~\ref{backtracking_tree_size}}\label{ss:proof_backtracking_tree_size}

We begin with the following lemma.

\begin{lemma}\label{lemma:bound_f}
For any $p\geq1$, $(n_i)_{1\leq i\leq p}$ and $(\alpha_i)_{1\leq i\leq p}$ with $n_i\geq2$ and $\alpha_i\geq1$,
$$f\left((n_1,\alpha_1),\dots,(n_p,\alpha_p)\right)\leq (e-1)\displaystyle\prod_{i=1}^p \alpha_i n_i! \displaystyle\sum_{i=0}^{p-1}\frac{1}{2^i}.$$
\end{lemma}
\begin{proof}
The proof is made by induction on $p$. One has $f\left((n,\alpha)\right)=\alpha a_n \leq (e-1)\alpha n!$, which shows the result for $p=1$. Now, if $p\geq 2$, we have
\begin{align*}
f\left((n_1,\alpha_1),\dots,(n_p,\alpha_p)\right)&=\alpha_1\left[a_{n_1}+n_1!f\left((n_2,\alpha_2),\dots,(n_p,\alpha_p)\right)\right]\\
&\leq (e-1)\alpha_1 n_1! \left(1+\prod_{i=2}^p \alpha_i n_i!\sum_{i=0}^{p-2}\frac{1}{2^i}\right)\\
&=(e-1)\prod_{i=1}^p \alpha_i n_i! \left(\frac{1}{\prod_{i=2}^p \alpha_i n_i!}+\sum_{i=0}^{p-2}\frac{1}{2^i}\right)\\
&\leq (e-1)\prod_{i=1}^p \alpha_i n_i! \sum_{i=0}^{p-1}\frac{1}{2^i},
\end{align*}
with $\prod_{i=2}^p \alpha_i n_i!\geq 2^{p-1}$ for $\alpha_i\geq1$ and $n_i\geq 2$.
\end{proof}

We now connect the stated inequality with the system made of bags and collections according to the principles developed in Appendix~\ref{app:variadic}. As previously explained, replacing the input tuples of $f$ by well-chosen parameters, the left-hand term models the size of the backtracking tree. In addition, the factor $\prod_{i=1}^p \alpha_i n_i!$ of the upper-bound can be interpreted as follows.
\begin{itemize}
    \item A bag $\bfB$ is represented by the tuple $(\#\bfB,1)$ and therefore contributes as $\#\bfB!$ to the product.
    \item A collection $\bfC_n$, $n\in\supp(\bfC)$, is related to the tuples $(n,\#\bfC_n)$, $(n,\#\bfC_n-1)$, \dots, $(n,1)$, contributing as $ n!^{\#\bfC_n}\#\bfC_n!$ to the product. As aforementioned, a collection can be processed in an alternative way corresponding to the tuple $(\#\bfC_n,1)$ first then $\#\bfC_n$ times the tuple $(n,1)$. This leads to a larger backtracking tree but contributes as the same factor.
\end{itemize}
As a consequence, we retrieve formula \eqref{size} of $N(\bags,\collections)$. Bounding the remaining sum by $2$ leads to the expected result. It should be noted that this upper-bound is valid whether or not we select the order of tuples to minimise $f$ (and also does not depend on the strategy used to process collections).

\subsection{Proof of Theorem~\ref{backtracking_tree}}\label{ss:proof_backtracking_tree}

Given $p\geq 2$ tuples $(n_i,\alpha_i)$, $n_i\geq2$ and $\alpha_i\geq1$, and denoting the first two $(m,\alpha)$ and $(n,\beta)$, we are interested in the conditions on $m$, $n$, $\alpha$ and $\beta$ under which one has
$$\Delta f=f\left((m,\alpha),(n,\beta),\dots,(n_p,\alpha_p)\right)-f\left((n,\beta),(m,\alpha),\dots,(n_p,\alpha_p)\right)\leq 0.$$
If $\Delta f\leq 0$, the objects corresponding to tuples $(m,\alpha)$ and $(n,\beta)$ must be processed in this order to minimise the size of the backtracking tree. On the other hand, if $\Delta f>0$, they have to be swapped.

\begin{remark}By virtue of the bubble sort principle \cite{astrachan2003bubble} and taking into account the recursive form \eqref{variadic_f} of $f$, it is enough to know when to swap the first two arguments to globally minimise $f$. Indeed, we can recursively permute the subsequent elements and determine the absolute minimum of $f$.
\end{remark}

With some elementary manipulations on the definition \eqref{variadic_f} of $f$, $\Delta f$ can be rewritten as
\begin{equation}\label{delta_f}
\Delta f = \beta a_n(\alpha m!-1) - \alpha a_m (\beta n!-1).
\end{equation}
We now assume $m\geq n$ without loss of generality and examine the following cases.
\begin{itemize}
    \item If $m=n$, then $\Delta f=a_m(\alpha-\beta)$ which shows that the tuple $(m,\min(\alpha,\beta))$ must be placed first to minimise $f$.
    \item If $m>n$ and $\beta\geq2$, then, in light of the first item of upcoming Lemma~\ref{beta_12}, the tuple $(m,\alpha)$ must be placed first to minimise $f$, whatever the value of $\alpha$.
    \item If $m>n$ and $\beta=1$, then, in light of the second item of upcoming Lemma~\ref{beta_12}, the tuple $(n,\beta)=(n,1)$ must be placed first to minimise $f$, whatever the value of $\alpha$.
\end{itemize}

\begin{lemma}\label{beta_12}
We assume $m>n$.
\begin{itemize}
    \item If $\beta\geq2$, then $\Delta f<0$.
    \item If $\beta=1$, then $\Delta f>0$.
\end{itemize}
\end{lemma}
\begin{proof}
The proof can be found in Appendix~\ref{lemma_proof}.
\end{proof}

\medskip

The strategy for dealing with bags and collections implemented in Algorithm~\ref{backtrack} has been derived from the previous results to minimise the size of the backtracking tree:
\begin{itemize}
    \item If the system is composed of bags and collections, since bags correspond to tuples $(n,1)$ and collections to tuples $(m,\alpha)$ with $\alpha\geq 2$, then bags have to be processed first.
    \item If there are only bags left, i.e. $\alpha_i=1$ for all $i$, they have to be processed by increasing size (as already stated in Lemma~\ref{lemma:two_bags}).
    \item If there are only collections $\bfC_n$ and $\mathbf{C}_m'$ with $n>m$ left, then $\bfC_n$ has to be processed first. If $n=m$, then $\bfC_n$ has to be processed first if $\#\bfC_n\leq\#\mathbf{C}'_m$.
\end{itemize}

\subsection{Proof of Lemma~\ref{beta_12}}\label{lemma_proof}

For any $n\geq1$ and $k\geq1$, we denote
$$
    b_n = \sum_{i=1}^{n-1}\frac{1}{i!},\qquad
    r_n = \sum_{i=n}^{\infty}\frac{1}{i!}\qquad\text{and}\qquad
    s_{n,k} = b_{n+k} - b_n.
$$
Consequently, with \eqref{a_n}, $a_n = n! b_n$, $b_n+r_n=e-1$, and $n!\leq a_n\leq (e-1)n!$. In addition, together with \eqref{delta_f},
$$\Delta f = \alpha\beta b_n n!m!\left[1-\frac{1}{\alpha m!} - \frac{b_m}{b_n}\left(1-\frac{1}{\beta n!}\right)\right].$$
If we assume moreover that $m>n$, then, with $k=m-n$,
\begin{equation}\label{delta_lambda}
\Delta f = -\alpha\beta b_n n!(n+k)!\left(\frac{1}{\alpha(n+k)!}+\frac{s_{n,k}(\beta n!-1)-b_n}{b_n\beta n!}\right).
\end{equation}

\paragraph{First item} In light of \eqref{delta_lambda}, $\Delta f<0$ is equivalent to 
$$\displaystyle\frac{1}{\alpha(n+k)!}+\frac{s_{n,k}(\beta n!-1)-b_n}{b_n\beta n!}> 0.$$
It suffices to show that the right-hand term is nonnegative, i.e. $s_{n,k}(\beta n! -1)\geq b_n$. Since $s_{n,k}$ is increasing with $k$, we only need to investigate the case $k=1$, which is equivalent to $\beta -1/n!\geq b_n$ because $s_{n,1}=1/n!$. This is equivalent to $\beta\geq b_{n+1}$, which is true because $b_n$ is bounded by $e-1$.

\paragraph{Second item} In light of \eqref{delta_lambda}, $\Delta f>0$ is equivalent after some elementary manipulations to
$$
\frac{1}{s_{n,k}}\left(1-\frac{n!}{\alpha(n+k)!}\right) > \frac{n!-1}{b_n},
$$
which we shall establish by a disjunction of cases on $n$.
\begin{itemize}
\item We assume $n=2$. It should be remarked that it is sufficient to investigate the case $\alpha=1$. Assuming $\alpha=1$ and multiplying each side of the inequality by $s_{2,k}$, the expected result is equivalent to
$$\sum_{i=2}^{k+2} \frac{1}{i!} < 1  - \frac{1}{(k+2)!},$$
which is true because the left-hand term is upper-bounded by $e-2$ and the right-hand term is lower-bounded by $5/6$.
\item From now on, we assume $n\geq3$. Since $\alpha\geq1$ and $k\geq 1$, it is easy to see that
$$\frac{1}{s_{n,k}}\left(1-\frac{n!}{\alpha(n+k)!}\right) >\frac{n}{s_{n,k}(n+1)}.$$
Consequently, it is sufficient to prove
$$\frac{1}{s_{n,k}}\frac{n}{n+1} >\displaystyle\frac{n!-1}{b_n}.$$
The left-hand term being a decreasing function of $k$, the inequality holds if we can establish it at the limit. As a consequence, we want to prove
$$
(n!-1)\left(1+\frac{1}{n}\right)\frac{r_n}{e-1-r_n} <1.
$$
\begin{itemize}
    \item If $n=3$, with $r_3=e-5/2$, the left-hand term is approximately evaluated to $0.97<1$.
    \item If $n\geq4$, with $n!-1<n!$, we bound the left-hand term by $(1+1/n)n! r_n / (e-1-r_n)$ which is a decreasing function of $n$ (as a product of decreasing functions of $n$). In addition, the term corresponding to $n=4$ is approximately evaluated to $0.92<1$.
\end{itemize}
\end{itemize}

\end{document}